\documentclass[12pt]{article}
\usepackage[utf8]{inputenc}
\usepackage[margin=1in]{geometry}
\usepackage{bm}
\usepackage{amsmath, color, natbib}
\usepackage{amssymb}
\usepackage{amsthm}
\usepackage{graphicx}
\usepackage{amsfonts}
\usepackage{mathrsfs}
\numberwithin{equation}{section}
\usepackage{booktabs}
\usepackage{float}
\usepackage{adjustbox, url}
\usepackage[table]{xcolor}
\usepackage[doublespacing]{setspace}
\usepackage{tikz}
\usepackage{authblk}
\usetikzlibrary{shapes,arrows,positioning}
\usepackage{algorithm}
\usepackage{algpseudocode}
\usepackage{changebar}
\usepackage{multirow}
\usepackage{hyperref}
\usepackage{subcaption} 
\usepackage{array} 
\usepackage{mathtools}
\usepackage{authblk}

\usepackage{enumitem}

\hypersetup{hidelinks}

\usepackage{bibunits}
\defaultbibliographystyle{apalike}   
\defaultbibliography{references}

\usepackage[english]{babel}
\theoremstyle{definition}

\newtheorem{theorem}{Theorem}

\newcommand{\yx}{\color{blue}\em } 

\newcommand{\xx}{\color{black}\rm } 

\newcommand{\bD}{{\mathbf{D}}}

\newcommand{\wst}{{w_h^*}}

\newcommand{\tilx}{{\tilde{x}}}

\newcommand{\email}[1]{\href{mailto:#1}{#1}}

\renewcommand{\yx}{\color{black}}

\renewcommand{\xx}{\color{black}}

\makeatletter
\@ifpackageloaded{chngcntr}{}{\usepackage{chngcntr}}
\makeatother

\begin{document}

\title{\bf WAIC-Optimized Weight Gating for Mixture Priors in External Data Borrowing}

\author{\small
Shouhao Zhou$^{1,2,3,\dagger,*}$, 
Qiuxin Gao$^{4,\dagger}$, Chenqi Fu$^{1}$, and 
Yanxun Xu$^{4,5,\dagger,**}$ \\
\small $^{1}$Department of Public Health Sciences, Pennsylvania State University, Hershey, Pennsylvania, U.S.A. \\
\small $^{2}$Department of Statistics, Pennsylvania State University, University Park, Pennsylvania, U.S.A. \\
\small $^{3}$Penn State Cancer Institute, Hershey, Pennsylvania, U.S.A. \\
\small $^{4}$Department of Applied Mathematics and Statistics, Johns Hopkins University, Baltimore, Maryland, U.S.A. \\
\small $^{5}$Division of Quantitative Sciences, Johns Hopkins School of Medicine, Baltimore, Maryland, U.S.A. \\
\small $^{\dagger}$These authors contributed equally to this work. \\
\small \textit{*email: \email{szhou1@pennstatehealth.psu.edu}}\\
\small \textit{**email: \email{yanxun.xu@jhu.edu}}
}
\date{}
\maketitle

\begin{abstract}
The integration of external data using Bayesian mixture priors has become a powerful approach in clinical trials, offering significant potential to improve trial efficiency. Despite their strengths in analytical tractability and practical flexibility, existing methods such as the robust meta-analytic-predictive (rMAP) and self-adapting mixture (SAM) often presume borrowing without rigorously assessing whether external information is appropriate to incorporate. When external and concurrent data are discordant, excessive borrowing can bias estimation and lead to misleading conclusions. To address this, we introduce WOW, a Kullback–Leibler-based gating strategy guided by the widely applicable information criterion (WAIC). 
\yx
Within the mixture-prior framework, WAIC-Optimized Weighting (WOW) conducts a preliminary compatibility assessment between external and concurrent trial data to determine eligibility for borrowing. Only if this gating criterion is satisfied does borrowing proceed; a downstream mixture prior procedure, using user-specified fixed or adaptive weights, can then be applied to determine the amount of borrowing. \xx Simulation studies demonstrate that incorporating the WOW strategy before Bayesian mixture prior borrowing methods effectively mitigates excessive borrowing and improves estimation accuracy. \yx  A real-data illustration further highlights the feasibility and interpretability of the proposed gate-then-borrow strategy. By providing a practical safeguard against inappropriate borrowing, WOW strengthens the reliability of mixture-prior methods and supports better decision-making in clinical trials.\xx
\end{abstract}


\section{Introduction}
Integrating external data into clinical trials holds significant potential for improving efficiency, particularly in rare diseases, pediatric populations, and studies with ethical constraints.
By leveraging historical trials and real-world data (RWD), researchers can  supplement control arm evidence and optimize resource allocation \citep{li_novel_2022}. 
However, the use of external controls introduces major validity concerns, as non-concurrent randomization may lead to bias from unobserved confounding or temporal shifts \citep{spanakis2023addressing}. Regulatory agencies, including the U.S. Food and Drug Administration (FDA) and the European Medicines Agency (EMA), recognize the great potential but emphasize the critical need for rigorous methodologies to ensure reliable integration \citep{EMA2020,us2021considerations}.

Statistical methods for integrating external data have advanced substantially, with Bayesian and frequentist approaches offering distinct strengths \citep{lesaffre2024review}. Bayesian methods are especially popular for their flexibility in adjusting the extent of borrowing based on the relevance of external data. For example, power priors use a discounting parameter to control the contribution of external data \citep{ibrahim_power_2000}. Commensurate priors quantify similarity between external and concurrent data through a commensurability parameter, thereby dynamically modulating the degree of borrowing \citep{hobbs_hierarchical_2011}. Hierarchical models enable information borrowing across multiple sources while accounting for study heterogeneity through exchangeability assumptions  \citep{berry2013bayesian}.


Among various developments, mixture prior methods have emerged as an attractive framework due to their practical flexibility and interpretability in handling prior-data conflicts. 
This approach represents the prior as a weighted mixture of an informative component (derived from external data) and a non-informative component, with the mixture weight controlling the extent of borrowing.  By varying this weight between 0 (no borrowing) and 1 (full borrowing), the method enables continuous, data-driven adjustment of external information integration. In conjugate settings, mixture priors maintain analytical tractability, yielding interpretable posterior distributions that preserve the mixture form.  

Earlier methods, such as the robust meta-analytic-predictive (rMAP) prior \citep{schmidli2014robust}, rely on fixed mixture weights informed by clinical judgment. However, such weights may not be readily or reliably available. To overcome this limitation, data-driven methods, also known as dynamic borrowing, have been proposed to determine mixture weights through conflict diagnostics that automatically downweight incompatible external data. \yx  For example, the predictive informative prior (PIP) calibrates mixture weights using prior-predictive conflict checks \citep{egidi2022avoiding}; the self-adapting mixture (SAM) prior uses posterior probability ratios to reduce borrowing as prior-data conflict increases \citep{yang2023sam}; and the empirical Bayes robust MAP (EB-rMAP) prior uses Box's prior predictive $p$-values to balance borrowing strength against potential incompatibility \citep{zhang2023adaptively}. \xx
Compared with alternative Bayesian approaches such as power priors and commensurate priors, adaptive mixture priors offer a flexible and interpretable framework for improving efficiency while mitigating bias under prior-data conflict.


Despite these advancements, adaptive mixture priors still depend on pre-specified tuning parameters to assess prior-data conflict, which may lead to overly aggressive borrowing when these parameters are misspecified. Moreover, existing dynamic borrowing methods primarily focus on adjusting the amount of borrowing after external information has entered the mixture-prior framework, rather than first evaluating whether its inclusion improves estimation or predictive performance for the concurrent trial data. This limitation can lead to misleading inference, highlighting the need for a principled, data-driven procedure to assess borrowing appropriateness before applying adaptive weighting schemes.

\yx To address this gap, we propose the WAIC-Optimized Weight (WOW) gating strategy, which uses the widely applicable information criterion (WAIC) to assess borrowing appropriateness before applying downstream mixture-prior methods. \xx WAIC estimates out-of-sample posterior predictive performance, making it well suited for assessing whether external-data borrowing improves fit to the concurrent trial data. By comparing WAIC under borrowing and no borrowing, WOW provides a principled, data-driven gate. \yx
If the historical-data prior worsens predictive performance for the concurrent control data, borrowing is avoided; otherwise, a downstream mixture-prior method, such as SAM or EB-rMAP, determines the degree of borrowing. Thus, WOW separates borrowing eligibility from borrowing magnitude, improving transparency and supporting more justifiable use of external data. \xx

This paper proceeds as follows. Section \ref{sec:methods} reviews mixture-prior borrowing methods, discusses weight specification, and presents an illustrative example showing how adaptive approaches can lead to overly aggressive borrowing.  In Section \ref{sec:WOW}, we introduce the proposed WOW gating strategy and its key properties. Section \ref{sec:wowbinary} develops the WOW procedure for binary and continuous endpoints. Section \ref{sec:simu} presents simulation studies that evaluate and compare the performance of WOW against existing methods, \yx and Section \ref{sec:real_data} illustrates its application using the RBesT binary responder example. \xx Finally, Section \ref{sec:conc} concludes with a discussion of broader implications and future extensions.

\section{Mixture-Prior Framework, Weight Specification, and Borrowing Challenges}
\label{sec:methods}

This section reviews Bayesian mixture priors as a flexible framework for borrowing external control information and examines the specification of the mixture weight. In existing approaches, the mixture weight is typically treated either as a fixed quantity, pre-specified or estimated from the data, or as a random variable assigned a prior distribution. \yx We show that, under the linear mixture-prior construction considered here, these formulations are equivalent for marginal inference on the parameter of interest, providing a useful simplification for subsequent methodological development. \xx Finally, we present an illustrative example highlighting how adaptive weighting methods can still lead to overly aggressive borrowing under prior-data conflict.

\subsection{Background} \label{sec:2.1}
Consider a standard two-arm randomized controlled trial (RCT) designed to compare a new treatment against a control. Let \( D = \{y_i\}_{i=1}^{n} \) denote the concurrent control data, where \( y_i \) denotes the outcome for subject \( i \) and follows density $f(\cdot\mid \theta)$,  \( i = 1, \dots, n \).  Suppose that historical or external data
are available only for the control group, denoted by \( D_h = \{y_{h,i}\}_{i=1}^{n_h} \). The goal is to properly leverage this external information to inform clinical trial analysis, particularly in estimating key parameters such as the treatment effect.

\cite{schmidli2014robust} introduced the robust meta-analytic predictive (rMAP) prior, which uses a fixed-weight mixture of priors to address potential conflict between historical and concurrent control data. The rMAP prior is expressed as:
\begin{equation}
    \pi(\theta| w_h) = w_h \pi_h(\theta) + (1 - w_h) \pi_0(\theta),
    \label{eq:mixture prior}
\end{equation}
where \( \pi_h(\theta) \) is an informative prior derived from historical data, \( \pi_0(\theta) \)  is  a vague or weakly informative prior, and \( w_h \) is a fixed weight representing the prior belief in the compatibility of historical and concurrent data. This weight controls the extent to which information is borrowed from historical data.

The mixture prior in (\ref{eq:mixture prior}) spans a continuum of borrowing strategies for the control arm.  At one extreme, complete borrowing (\( w_h = 1 \), i.e., \( \theta \sim \pi_h \)) assumes full compatibility between historical and concurrent data. At the other extreme, no borrowing (\( w_h = 0 \), i.e., \( \theta \sim \pi_0 \)) excludes historical data \( D_h \) completely. Intermediate values of \( w_h \) represent partial borrowing, while alternative formulations may treat \( w_h \) as a random variable \citep{yang2023sam}. The informative prior  \( \pi_h(\theta) \) itself can be flexible, for example representing a posterior distribution from a single historical study or a pooled prior from multiple studies.
\yx More recent adaptive approaches differ mainly in how \( w_h \) is determined from the data. For example, PIP, SAM, and EB-rMAP use different prior-data conflict diagnostics to reduce borrowing when historical and concurrent data appear incompatible.  In particular, SAM uses a user-specified clinically meaningful difference threshold $\delta$, whereas EB-rMAP uses a prior predictive $p$-value (PPV) threshold $\gamma$, both of which control the borrowing weight in response to prior-data conflict.\xx 


\subsection{Weight specification and marginal equivalence}
A central modeling choice in mixture-prior borrowing is how to specify the weight parameter $w_h$, as it governs the extent to which external data $D_h$ contribute to inference based on the concurrent trial data $D$. In many commonly used approaches, $w_h$ is treated as a fixed scalar, either pre-specified by user input or estimated from the data. This formulation is computationally simple and easy to interpret.

However, as noted by \citet{schmidli2014robust}, mis-specification of $w_h$ can inflate type I error and bias inference, particularly in the presence of prior-data conflict. To address uncertainty in the degree of borrowing, an alternative formulation treats $w_h$ as a random variable with prior distribution $\pi(w_h)$ \citep{yang2023sam}. Under this setting, inference for the parameter of interest $\theta$ is based on the marginal posterior 
$$\theta \sim p_M(\theta \mid D,D_h) = \int_{w_h} p(\theta, w_h \mid D,D_h) \, dw_h.$$ 

\yx Although this formulation appears more flexible, under the linear mixture-prior construction considered here, its effect on marginal inference for \(\theta\) is determined only by the prior mean of \(w_h\). The following result formalizes this property. \xx
\vspace{-5pt}

\begin{theorem}[Marginal prior-mean dependency]
Under the linear mixture-prior construction in (\ref{eq:mixture prior}) with conjugate mixture component distributions, the marginal posterior distribution \(p_M(\theta \mid D,D_h)\) depends on the prior distribution \(\pi(w_h)\) only through
$
\overline{w}_h = \int w_h \, \pi(w_h)\,dw_h,$
the prior mean of \(w_h\). Therefore, for marginal inference on \(\theta\), the random-weight formulation induces the same posterior as fixing \(w_h=\overline{w}_h\).
\label{theo:no w prior}
\end{theorem}

\begin{proof}
Given the prior distribution \( \pi(w_h) \), the joint prior distribution of \((\theta, w_h)\) is:
\[
\pi(\theta, w_h \mid D_h) = \left[ w_h \pi_h(\theta) + (1 - w_h) \pi_0(\theta) \right] \pi(w_h) .
\]

Integrating out \( w_h \), the marginal posterior distribution \( p(\theta \mid D, D_h) \) becomes:
\begin{equation*}
\begin{aligned}
p_M(\theta \mid D, D_h) &= \int_{w_h} \pi(\theta, w_h\mid D_h) \prod_{i=1}^{n} f(y_i \mid \theta)  \, dw_h ,\\ &= \overline{w}_h^* \, p_h(\theta \mid D, D_h) + (1 - \overline{w}_h^*) \, p_0(\theta \mid D),
\end{aligned}\label{eq: post.margin}
\end{equation*}
where the weight \( \overline{w}_h^* \) is given by
$\overline{w}_h^* = \frac{\overline{w}_hz_h}{\overline{w}_{h}z_h + (1 - \overline{w}_h) z_0},$
with $z_h=\int \pi_h(\theta) \prod_{i=1}^n f(y_i \mid \theta)\, d\theta$ and $z_0=\int \pi_0(\theta) \prod_{i=1}^n f(y_i \mid \theta)\, d\theta,$ and $$
p_h(\theta \mid D) = \pi_h(\theta) \prod_{i=1}^n f(y_i \mid \theta) / z_h, \quad
p_0(\theta \mid D) =\pi_0(\theta) \prod_{i=1}^n f(y_i \mid \theta) / z_0,
$$
are the posteriors under full borrowing of historical data and no borrowing, respectively.
\end{proof}

\yx This result should be interpreted as a marginal equivalence result for \(\theta\), not as a full equivalence between random-weight and fixed-weight formulations.  At the joint level, the random-weight formulation remains distinct because it induces posterior uncertainty in \(w_h\). For example, the alternative approach considered by \citet{yang2023sam}, which assigns \(w_h\) a noninformative uniform prior, is marginally equivalent to an rMAP prior with fixed weight $w_h=0.5$ for posterior inference on $\theta$. \xx


In light of Theorem \ref{theo:no w prior}, the remainder of this paper focuses on fixed or data-adaptively selected borrowing weights. Given a fixed prior weight $w_h$, the posterior distribution of $\theta$ is 
\begin{equation}
    \theta \sim p(\theta \mid D,D_h) = \wst \, p_h(\theta \mid D, D_h) + (1 - \wst) \, p_0(\theta \mid D),
    \label{eq: post.final}
\end{equation}
where $\wst={w_h z_h}/{[w_h z_h+(1-w_h) z_0]}$ is the posterior mixture weight.
This simplification allows us to concentrate on borrowing decisions and their practical implications within a unified mixture-prior framework.

\subsection{An illustrative toy 
example}  
 \yx We first present a toy example to illustrate how existing mixture prior methods can lead to inappropriate borrowing in the presence of prior-data conflict. \xx
Consider a trial with a binary endpoint, where the concurrent control group has a response rate of $\theta=0.435$, and the treatment group has $\theta_t=0.466$. Historical data, intended to supplement the concurrent control, has a response rate of $\theta_h=0.4$. In this case, a large sample size from the historical data might provide rich information, but could also indicate a lack of overlap with the concurrent control group. 

\begin{figure}[htbp]
  \centering
  \includegraphics[width=0.8\textwidth]{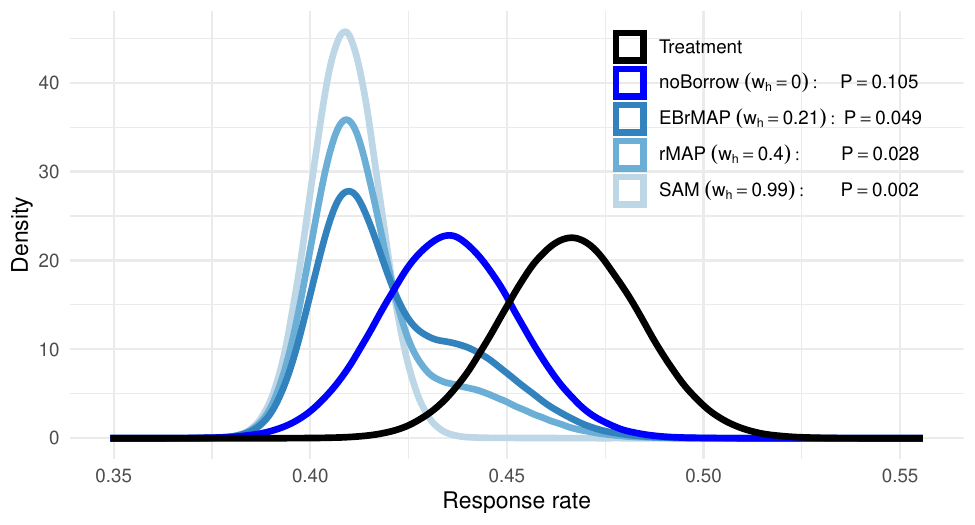}
  \caption{Motivating Example. This figure demonstrates how borrowing from incompatible historical data can artificially inflate significance and result in misleading inference. It compares the posterior distributions of the treatment effect across four borrowing strategies. }
  \label{fig:motivation}
\end{figure}

As shown in Figure~\ref{fig:motivation}, a standard data analysis for group comparison yields a non-significant treatment effect ($p$ = 0.105). However, when applying existing approaches: (1) SAM with $\delta$ = 0.1, (2) EB-rMAP with a PPP threshold $\gamma = 0.9$, and (3) rMAP with a fixed weight $w_h = 0.4$, all three methods incorporate non-negligible weights, either through dynamic weighting or fixed borrowing.  Despite observed incompatibility between historical and concurrent control data, these methods reduce the estimated control response rate through borrowing, thereby artificially inflating the estimated treatment difference. As a result, the originally non-significant finding crosses the conventional significance threshold ($p < 0.05$), raising concerns about biased inference due to inappropriate borrowing.



\section{WAIC-Optimized Weight (WOW) Gating Strategy}  
\label{sec:WOW}

\yx Motivated by the borrowing challenges described above, we propose the WOW gating strategy. WOW is designed to assess the empirical compatibility between historical data $D_h$ and concurrent trial data $D$ before external information is incorporated into the mixture prior, thereby enabling a data-driven pre-borrowing decision. \xx 

\subsection{A two-step gating + borrowing structure}
The WOW gating strategy implements a structured two-step process that separates the \textit{decision to borrow} from the \textit{method of borrowing}. In the first ``gating'' step, WAIC \citep{watanabe2021waic} is used to assess whether incorporating historical data into the prior improves out-of-sample predictive performance for the concurrent control data. Only when the gating criterion is satisfied does the second ``borrowing'' step proceed, in which historical and concurrent data are combined using a Bayesian mixture prior with adaptive weighting.

In contrast to existing borrowing methods that directly assign a borrowing weight, WOW provides a preliminary safeguard against inappropriate use of external information. The downstream borrowing step remains flexible and can accommodate a variety of weighting schemes, including rMAP, PIP, EB-rMAP, SAM, and user-defined alternatives. Importantly, the gating step is independent of the downstream weighting rule, allowing WOW to be integrated into existing mixture-prior borrowing frameworks. Implementing WOW requires formalizing the compatibility assessment and the resulting gating decision rule. In the next subsection, we present the mathematical formulation of the WOW procedure and characterize key properties of the proposed gate.

\subsection{WAIC specification for the gating decision}
WAIC is a Bayesian model selection criterion grounded in the predictive Kullback–Leibler divergence.  
It quantifies out-of-sample predictive accuracy by computing the expected log pointwise predictive density (elppd) across control observations $y_i$ in the concurrent dataset $D$.  Unlike conventional Bayesian model selection tools that rely heavily on regularity assumptions, WAIC is well-suited for dynamic borrowing settings where mixture models can be singular or misspecified \citep{watanabe2021waic}.

Given the posterior distribution \( p(\theta \mid D, D_h) \) in Eq.~(\ref{eq: post.final}), WAIC is specified as:
\begin{equation}
\text{WAIC}(w_h, D, D_h) =
-2 \sum_{i=1}^n \mathbb{E}_{\theta \sim p(\theta \mid D, D_h)} \log f(y_i \mid \theta)
+ 2 \sum_{i=1}^n \mathrm{Var}_{\theta \sim p(\theta \mid D, D_h)} \log f(y_i \mid \theta),
\label{eq: WAIC1}
\end{equation}
where $f(y_i \mid \theta)$ is the density function of $y_i$, and the expectation and variance are computed with respect to the posterior distribution $p(\theta \mid D, D_h)$. By evaluating WAIC as a function of the borrowing weight $w_h$, we determine whether adaptive integration of historical data  $D_h$ improves the model's ability to capture the true data-generating process of $D$.

Evaluating WAIC across all values of \( w_h \) in \([0,1]\) can pose computational challenges. The following result shows that, for the purpose of the WOW gating decision, it is sufficient to compare the two boundary cases corresponding to no borrowing and full borrowing.
\vspace{-3pt}

\begin{theorem}[Minimization of WAIC]
\yx For independent datasets \(D\) and \(D_h\), \(\text{WAIC}\) in Eq.~(\ref{eq: WAIC1}) is a concave quadratic function of the posterior mixture weight \(w_h^*\). Since \(w_h^*\) is a monotone transformation of the prior borrowing weight \(w_h\) that maps \([0,1]\) onto itself, the WAIC-based gating decision reduces to comparing the boundary cases \(w_h=0\) and \(w_h=1\).\xx \label{prop: 1}
\end{theorem}

The proof is provided in Supplementary Materials Section~A. \yx This result implies that the WOW gate can be constructed by comparing the no-borrowing model, $w_h=0$, with the full-borrowing model,  $w_h=1$.  \xx
Given $D_h$, we define the borrowing exclusion region  $\Omega_0$ for $D$ as:
\begin{equation}
    \Omega_0=\{ D \in \Omega: \text{WAIC}(w_h=0, D, D_h) < \text{WAIC}(w_h=1, D, D_h)\}, \label{borrowing exclusion region}
\end{equation} 
which is a compact subset of the sample space $\Omega$. \yx When $D \in \Omega_0$, the no-borrowing model is preferred because borrowing does not improve predictive performance. When $D \in \Omega \setminus \Omega_0$, borrowing is considered admissible. This does not imply that full pooling is appropriate; rather, a downstream mixture-prior method is still used to determine the amount of borrowing. 

Algorithm~\ref{alg:WOW} summarizes the proposed WOW-assisted gating-and-borrowing procedure. When the WOW gate supports borrowing, the prespecified downstream borrowing rule is applied with its corresponding weight; otherwise, the borrowing weight is set to zero. Thus, WOW can be viewed as a truncated borrowing rule that retains the downstream borrowing weight when the gate supports compatibility and truncates it to zero otherwise.
\begin{algorithm}[!h]
\caption{\ WOW-assisted two-step gating and borrowing procedure}
\label{alg:WOW} \vspace{12pt}
\begin{algorithmic}[1]\yx
\Require historical data $D_h$, concurrent control data $D$, prior $\pi_0(\theta)$ for mixture prior $\pi(\theta; w) = (1-w)\pi_0(\theta) + w\,\pi_h(\theta \mid D_h)$ with $w\in [0,1]$, 
and a prespecified downstream borrowing rule $\mathcal{A}$ (e.g., SAM) that determines the borrowing weight when the WOW gate is open. 

\hspace{-46pt} \textbf{Step 1:} WAIC-Optimized Weight (WOW) gating

\hspace{-46pt} Compute WAIC using Eq.~(\ref{eq: WAIC1}) with posterior distributions under no borrowing and full borrowing. Construct borrowing exclusion region $\Omega_0$ using Eq.~(\ref{borrowing exclusion region}).

\hspace{-46pt} \textbf{if} {$D \in \Omega_0$}\  \textbf{then}

\hspace{-26pt} Gate \textit{closed}: {Stop} borrowing with $ w_h = \wst = 0 $; posterior 
$p^{\star}(\theta \mid D, D_h) \leftarrow p_0(\theta \mid D)$

\hspace{-46pt} \textbf{else}\quad {Gate remains \textit{open}: continue to \textbf{Step 2} for borrowing}


\hspace{-46pt} \textbf{Step 2:} external borrowing

\hspace{-46pt} Apply the user-prespecified downstream borrowing rule $\mathcal{A}$ to determine $w_h$, and obtain
$p^{\star}(\theta \mid D, D_h) \leftarrow p(\theta \mid D, D_h)$ using Eq.~(\ref{eq: post.final})


\hspace{-46pt} \Return $p^{\star}(\theta \mid D, D_h)$
\end{algorithmic}
\end{algorithm}
\xx

The WOW gate can be applied prospectively during trial design or retrospectively after data collection. Prospectively, WAIC can be evaluated over possible realizations of \( D \) within the sample space $\Omega$. This enables identification of the borrowing exclusion region, where 
WAIC yields its minimum at $w_h=0$, indicating that borrowing is not supported. Retrospectively, once $D$ is observed, WAIC can be computed directly to determine whether borrowing is warranted for the specific dataset.

\section{Gating Strategy for Binary and Continuous Endpoints}
\label{sec:wowbinary}

In this Section, we demonstrate the WOW gating strategy for binary endpoints, where both concurrent data $D$ and historical data $D_h$ follow binomial distributions with conjugate beta mixture priors. 
The continuous endpoint case is discussed in Supplementary Materials Section F. 

Consider a binary outcome, \( y_i \sim \text{Bernoulli}(\theta) \), where \( \theta \) is the probability of success. Let \( x = \sum_{i=1}^n y_i \) denote the number of successes observed among $n$ individuals in the concurrent control group. Suppose that for the historical data \( D_h \), there are \( x_h \) successes out of \( n_h \) individuals. A commonly-used informative prior for \( \theta \) derived from $D_h$, assuming full borrowing, is:  $\pi_h(\theta\mid D_h) \propto \pi_0(\theta)p(D_h\mid \theta) =  \text{Beta}(a + x_h, b + n_h - x_h)$, where $\pi_0(\theta) = \text{Beta}(a, b)$ is a non-informative or vague prior. \yx Therefore, in Eq.~(\ref{eq: post.final}) we have
$p_h(\theta \mid D,D_h)=\mathrm{Beta}(a+x+x_h,\; b+n+n_h-x-x_h)$ and
$p_0(\theta \mid D)=\mathrm{Beta}(a+x,\; b+n-x)$, 
with
$z_0=\dfrac{B(a+x,\; n-x+b)}{B(a,b)}$
and
$z_h=\dfrac{B(a+x_h+x,\; b+n_h+n-x_h-x)}{B(a+x_h,\; b+n_h-x_h)}$,
which are used to compute the posterior mixture weight $\wst$.\xx

Given the historical data $(n_h,x_h)$ and a fixed control sample size $n$, the WAIC for binary outcomes can be expressed as a quadratic function of $\wst$:
\begin{equation}
\mathrm{WAIC}_B(\wst,D,D_h)
=
-I_1\,\wst^2+I_2\,\wst+I_3,
\label{eq:WAIC_binary}
\end{equation}
where \yx $I_1=2\Bigl\{(n-x)\Delta_{1-\theta}^{2}
+
x\Delta_{\theta}^{2}\Bigr\}$, $I_2=2(n-x)\Bigl\{V_{1-\theta}+\Delta_{1-\theta}^{2}-\Delta_{1-\theta}\Bigr\}
+
2x\Bigl\{V_{\theta}+\Delta_{\theta}^{2}-\Delta_{\theta}\Bigr\}$, and $I_3$ is constant with respect to $\wst$. Here, $\Delta_{\theta}$ and $\Delta_{1-\theta}$ denote the differences in posterior expectations of $\log\theta$ and $\log(1-\theta)$, respectively, between the full-borrowing and no-borrowing posteriors; $V_{\theta}$ and $V_{1-\theta}$ denote the corresponding differences in posterior variances. Detailed expressions are provided in Supplementary Materials Section~B. This quadratic form leads to the following structural property of the borrowing region. \xx



\begin{theorem}[Existence of a Single Connected Borrowing Region]
Given historical data $(n_h, x_h)$ and a fixed concurrent control sample size $n$,  there exists a single connected region \( G = [x_L^*, x_U^*] \subseteq [0, n] \) such that borrowing from the historical data is beneficial if and only if \( x \in G = \Omega \setminus \Omega_0\), where \( x \) is the number of observed successes in the concurrent control.
\end{theorem}
The proof is provided in Supplementary Materials Section C.

This result provides the foundation for a simple and practical gating rule: the borrowing region is always a single connected interval rather than a disjoint subset of possible outcomes. \xx Given  \( n_h \), \( x_h \), and \( n \), the thresholds \( x_L^* \) and \( x_U^* \) defining the borrowing region $G=\Omega \setminus \Omega_0$ can be computed numerically before trial implementation. Specifically, this is done by evaluating $\mathrm{WAIC}_B(0, D, D_h)$ and $\mathrm{WAIC}_B(1, D, D_h)$ across all possible values of $x$, and identifying the smallest and largest $x$ for which borrowing improves predictive performance. 

\yx  Algorithm S1 in Supplementary Materials Section~D summarizes the resulting binary-endpoint implementation of the two-step WOW procedure. In practice, the Step~1 gating thresholds defining $G$ can be pre-tabulated in advance. Figure~\ref{fig:SAM and EB binary} illustrates $[x_L^*, x_U^*]$ across varying historical sample sizes $n_h$ for SAM and EB-rMAP priors. Importantly, the proposed Step 1 gating region is invariant to the Step 2 borrowing rule, whether fixed or adaptive. This allows WOW to serve as a general pre-borrowing safeguard for existing mixture-prior methods. As shown in Figure~\ref{fig:SAM and EB binary}, EB-rMAP can exhibit asymmetric borrowing behavior, while SAM does not directly account for the historical sample size \(n_h\) in its weighting rule, making its borrowing behavior sensitive to the user-specified threshold \(\delta\). In contrast, the WOW borrowing region narrows as  \(n_h\) increases, reflecting greater precision in the historical data and a stricter evidence-based compatibility requirement. By incorporating both empirical compatibility and historical-data precision before applying any downstream mixture-prior rule, WOW improves the robustness and interpretability of external-data borrowing. Further discussion is provided in Supplementary Materials Section~E.  \xx
\begin{figure}[htbp]
  \centering
  \hspace{-0.05\textwidth}  
  \includegraphics[width=1.05\textwidth]{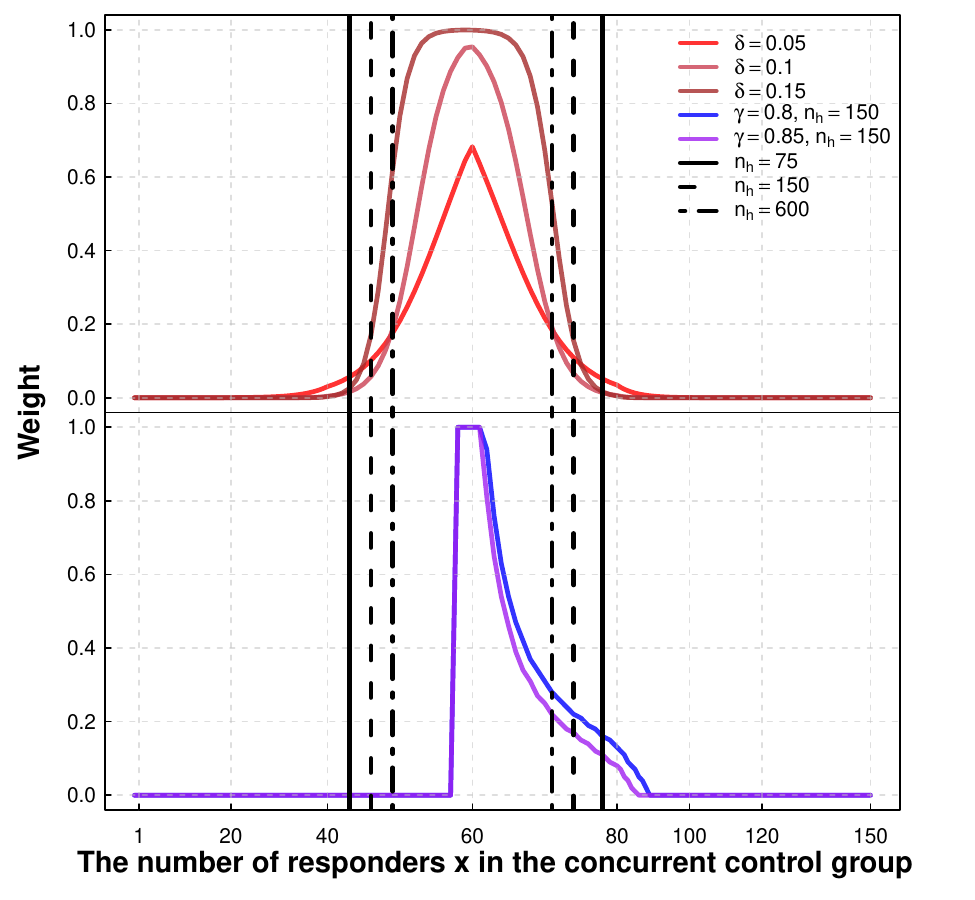}
  \caption{Comparison of borrowing regions and weight behavior across different borrowing strategies. Upper panel: SAM prior with varying $\delta$ values; lower panel: EB-rMAP prior with varying PPP threshold $\gamma$. Vertical black lines indicate the borrowing region determined by the WOW gating strategy. }
  \label{fig:SAM and EB binary}
\end{figure}

\yx 
\section{Simulation Study}
\label{sec:simu}

We conducted simulation studies to evaluate the WOW gating strategy for fixed weight (Mix50 with $w_h = 0.5$) and adaptive weight (PIP, SAM, and EB-rMAP) mixture priors by comparing each method with its WOW-gated counterpart. No borrowing (NP) and test-then-pool (TTP) were included as benchmarks. The remainder of this Section focuses on binary endpoints; parallel continuous endpoint simulations, reported in Supplementary Materials Section H, yielded similar conclusions: WOW reduced bias and MSE under prior-data conflict while preserving comparable operational characteristics performance under compatibility. \xx

\subsection{Simulation setup}

We generated historical control data as $D_h \sim \mathrm{Binomial}(n_h,\theta_h)$ and concurrent control data as $D \sim \mathrm{Binomial}(n,\theta)$, where $\theta$ was systematically varied to induce different levels of prior-data conflict. In all scenarios, the concurrent control arm sample size was fixed at $n=150$, with a 2:1 randomization ratio yielding $n_t=300$ for the treatment arm. 

\yx 
To assess borrowing performance for estimating $\theta$, we fixed $\theta_h=0.3$ and considered historical sample sizes $n_h=150,600$, and $1500$. The concurrent control rate $\theta$ varied from $0.1$ to $0.5$, with 2000 simulation replicates at each setting. Methods were evaluated using relative bias, ratio mean squared error (MSE), and absolute bias. To complement these point-estimation metrics, we also evaluated interval-estimation performance using coverage probability and interval score. We further considered a stochastic historical-data setting at  \(n_h=1500\), in which \(D_h\) was generated from its sampling distribution in each replicate to reflect realistic historical-data collection. This contrasts with fixed historical-data settings, which isolate the effect of the borrowing strategy from random historical-sampling variability. 

For trial operating characteristics, we considered three calibrated-power scenarios at $n_h=600$: fixed historical data with $\theta_h=0.30$, fixed historical data with $\theta_h=0.40$, and stochastic historical data generated under $\theta_h=0.40$. 
Treatment efficacy was declared when
$\Pr(\theta_t-\theta>0\mid D,D_t,D_h)>C$, where $C$ was calibrated separately for each method to control the type I error rate at 5\% under the corresponding null setting. 
Implementation details, including hyperparameters, tuning parameters for all comparator methods, definitions of performance metrics, and additional simulation settings, are provided in Supplementary Materials Section~G.

\subsection{Simulation results}

Figure~\ref{fig:RelativeBias} reports relative bias in estimating $\theta$ for SAM, EB-rMAP, Mix50, and PIP, together with their WOW-gated versions, over $\theta\in[0.1,0.5]$ with $\theta_h=0.3$. Results are shown for fixed historical sample sizes and for an additional stochastic historical-data scenario; TTP is included as a benchmark. Non-gated methods generally exhibit larger relative bias as $\theta$ departs from $\theta_h$, particularly in moderate-conflict regions where borrowing remains non-negligible despite increasing incompatibility. SAM and Mix50 are nearly unbiased at $\theta=\theta_h$, but their bias increases as $\theta$ moves away from $\theta_h$, reaches a peak, and then decreases. PIP shows a similar pattern near $\theta_h$ but often has larger bias under conflict, whereas EB-rMAP exhibits a comparable but more asymmetric pattern due to its PPP-based weight adjustment. The benefit of WOW gating becomes more pronounced as $n_h$ increases, because larger historical samples make the informative prior more influential and amplify bias when borrowing is inappropriate. By restricting borrowing when concurrent and historical controls are incompatible, WOW-gated methods substantially reduce bias and generally outperform TTP across most of the displayed range. The stochastic historical-data scenario closely follows the fixed $n_h=1500$ setting, with WOW-gated methods continuing to reduce bias relative to their non-gated counterparts.

\begin{figure}[htbp]
\centering
\includegraphics[width=1\textwidth]{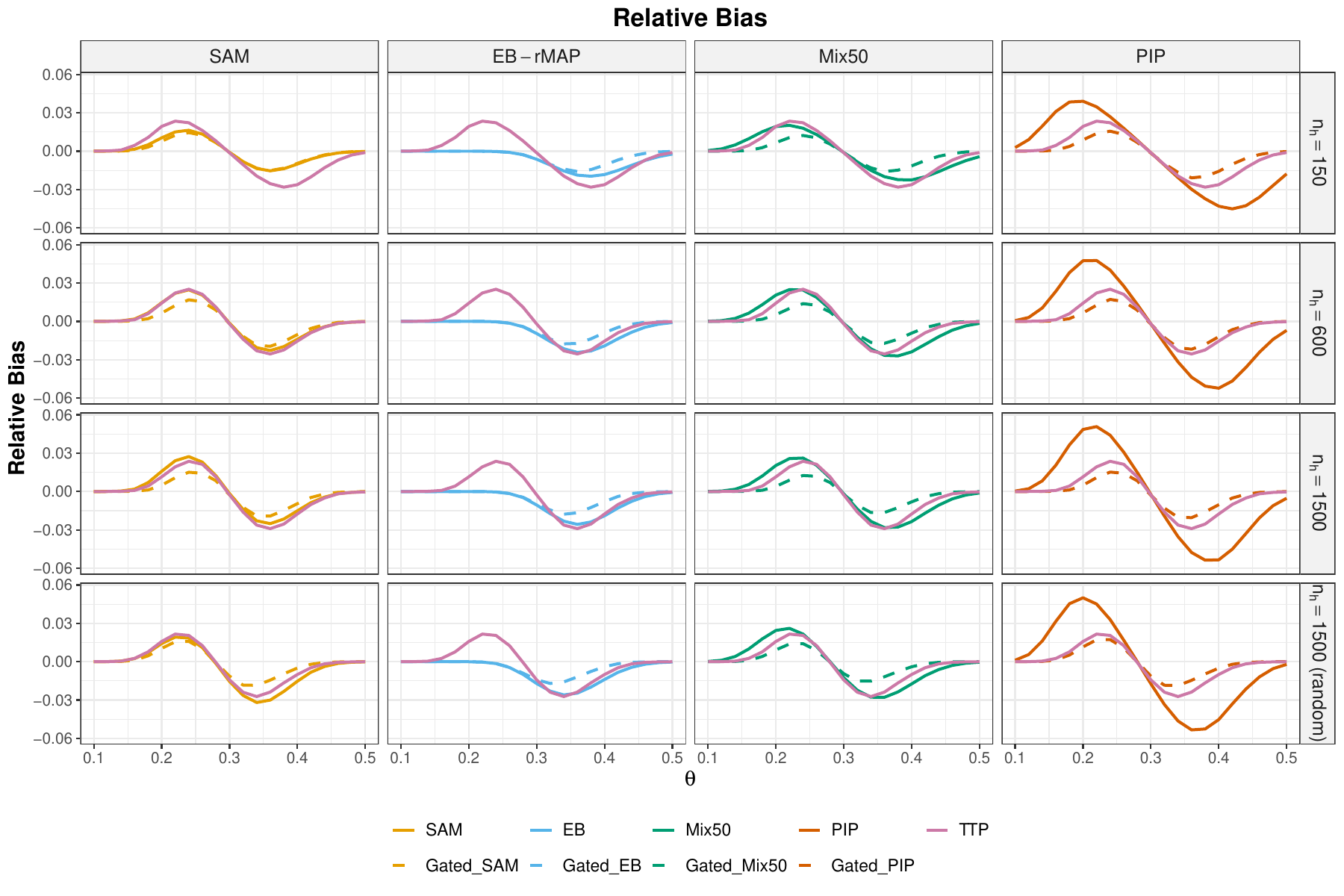}
  \caption{Relative bias in estimating the concurrent control response rate $\theta$ for the binary endpoint, comparing the original and WOW-gated versions of SAM, EB-rMAP, Mix50, and PIP across different historical sample sizes and an additional stochastic historical-data scenario. TTP is included as a benchmark.
}
  \label{fig:RelativeBias}
\end{figure}

Figure~\ref{fig:RatioMSE} reports the corresponding ratio MSE. 
The findings are consistent with the bias results: WOW-gated methods reduce MSE under prior--data conflict while maintaining comparable performance when $\theta$ is close to $\theta_h$. The improvement is most apparent for larger $n_h$, where non-gated methods have greater potential to overweight incompatible historical information. Additional results for estimating \(\theta\), including absolute bias, coverage probability, and interval score, are reported in Figure~S1 and Tables~S1--S2; additional results for the treatment effect \(\theta_t-\theta\), including bias and coverage, are reported in Tables~S3--S4. These results support the same conclusion.

\begin{figure}[htbp]
  \centering
  \includegraphics[width=1\textwidth]{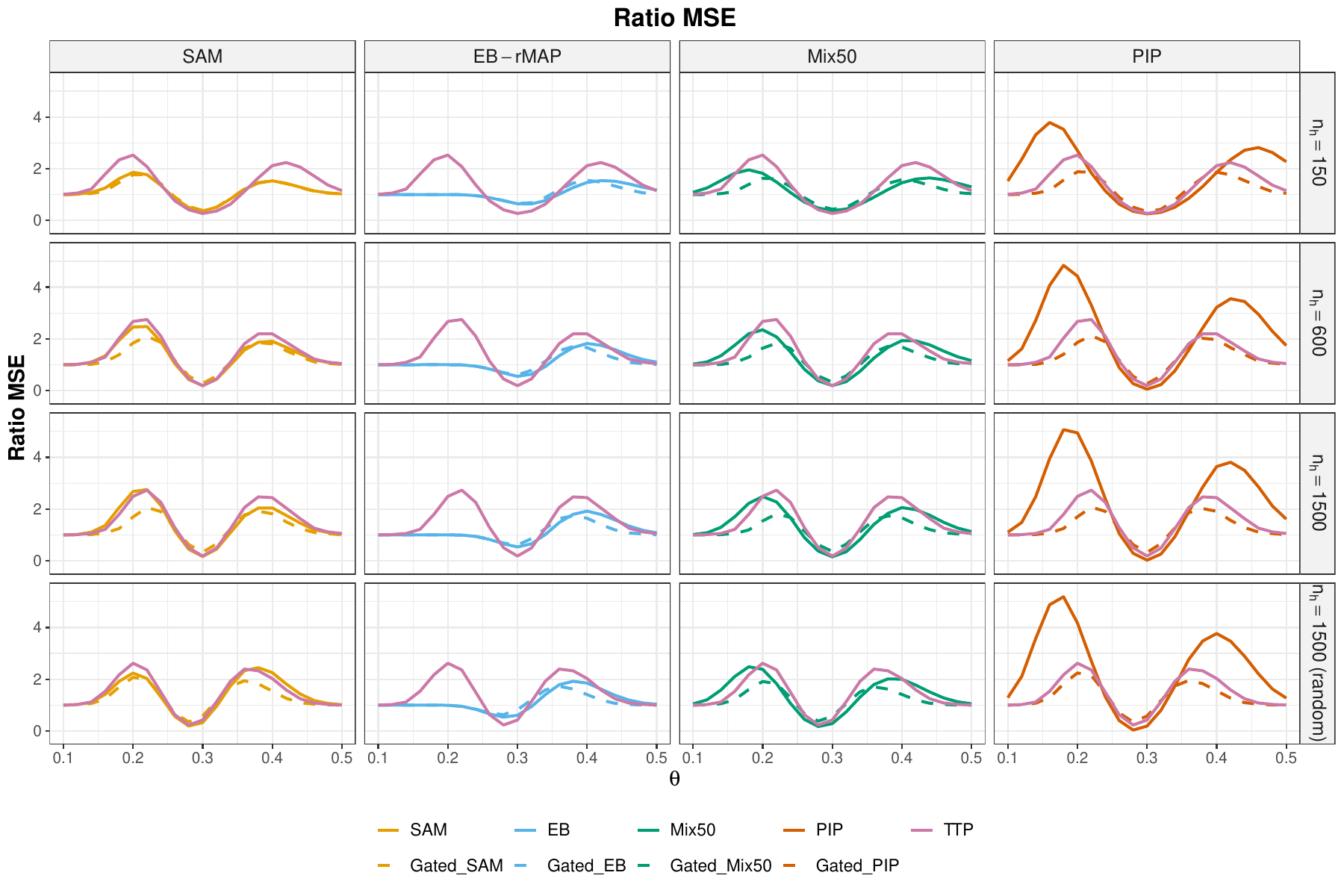}
  \caption{Ratio MSE in estimating the concurrent control response rate $\theta$ for the binary endpoint, comparing the original and WOW-gated versions of SAM, EB-rMAP, Mix50, and PIP across different historical sample sizes and an additional stochastic historical-data scenario. TTP is included as a benchmark. 
}
  \label{fig:RatioMSE}
\end{figure}

We next examined whether the estimation gains from WOW translate into improved trial-level operating characteristics across three scenarios (Table~\ref{tab:main:bin:calibrated-power-revised}). In Case 1, when $\theta=0.16$, the concurrent control rate is well below $\theta_h=0.3$, so substantial borrowing is inappropriate. Borrowing pulls the posterior control rate toward the historical value, leading to underestimation of the treatment effect. In this setting, WOW-gated methods improve over their non-gated counterparts, with Gated SAM outperforming SAM and larger gains observed for Gated Mix50 and Gated PIP relative to Mix50 and PIP. Similar improvements appear at $\theta=0.20$, and the gated methods also outperform TTP at $\theta=0.20$ and $0.22$, where TTP can still lead to full pooling under moderate conflict and consequently lower power. When $\theta=\theta_h=0.3$, the gated methods achieve power similar to their original versions. Notably, Gated SAM, Gated EB-rMAP, and Gated PIP still outperform NP, showing that WOW preserves efficiency when borrowing is appropriate while adding protection under conflict.

\begin{table}[htbp]
\centering
\caption{Calibrated power results for binary endpoints using a non-informative prior (NP) and four borrowing methods (SAM, EB-rMAP, Mix50, and PIP), their WOW-gated versions, and TTP as a benchmark. Boldface indicates the better-performing method within each original/WOW-gated pair for a given scenario. }

\label{tab:main:bin:calibrated-power-revised}
\resizebox{\textwidth}{!}{%
\begin{tabular}{ccccccccccccc}
\toprule
Scenario & $\theta$ & $\theta_t$ & NP & SAM & Gated SAM & EB-rMAP & Gated EB-rMAP & Mix50 & Gated Mix50 & PIP & Gated PIP & TTP \\
\midrule
\multicolumn{13}{c}{\textbf{Case 1: $\theta_h=0.3$ with fixed $\bD_h$}} \\
\addlinespace[0.5ex]
1.1 & 0.16 & 0.26 & 0.797 & 0.780 & \textbf{0.817} & 0.771 & \textbf{0.805} & 0.659 & \textbf{0.801} & 0.466 & \textbf{0.797} & 0.797 \\
1.2 & 0.18 & 0.28 & 0.771 & 0.707 & \textbf{0.769} & \textbf{0.773} & 0.771 & 0.545 & \textbf{0.765} & 0.357 & \textbf{0.769} & 0.747 \\
1.3 & 0.20 & 0.30 & 0.736 & 0.590 & \textbf{0.722} & 0.733 & \textbf{0.739} & 0.487 & \textbf{0.702} & 0.401 & \textbf{0.699} & 0.643 \\
1.4 & 0.22 & 0.32 & 0.715 & 0.533 & \textbf{0.638} & \textbf{0.712} & 0.710 & 0.608 & \textbf{0.653} & \textbf{0.773} & 0.638 & 0.555 \\
1.5 & 0.30 & 0.40 & 0.676 & \textbf{0.928} & 0.919 & \textbf{0.874} & 0.872 & \textbf{0.918} & 0.892 & \textbf{0.959} & 0.931 & 0.945 \\
1.6 & 0.34 & 0.44 & 0.654 & \textbf{0.724} & 0.722 & \textbf{0.689} & 0.684 & \textbf{0.721} & 0.714 & \textbf{0.891} & 0.811 & 0.848 \\
1.7 & 0.44 & 0.54 & 0.649 & 0.512 & \textbf{0.530} & 0.536 & \textbf{0.551} & 0.530 & \textbf{0.550} & 0.153 & \textbf{0.249} & 0.102 \\
\specialrule{\heavyrulewidth}{1pt}{1pt}
\multicolumn{13}{c}{\textbf{Case 2: $\theta_h=0.4$ with fixed $\bD_h$}} \\
\addlinespace[0.5ex]
2.1 & 0.24 & 0.34 & 0.694 & 0.690 & \textbf{0.702} & 0.697 & \textbf{0.702} & 0.599 & \textbf{0.709} & 0.444 & \textbf{0.695} & 0.696 \\
2.2 & 0.26 & 0.36 & 0.677 & 0.670 & \textbf{0.685} & 0.662 & \textbf{0.664} & 0.494 & \textbf{0.694} & 0.324 & \textbf{0.663} & 0.660 \\
2.3 & 0.30 & 0.40 & 0.676 & 0.544 & \textbf{0.630} & 0.666 & 0.666 & 0.452 & \textbf{0.630} & 0.409 & \textbf{0.628} & 0.576 \\
2.4 & 0.32 & 0.42 & 0.671 & 0.515 & \textbf{0.585} & \textbf{0.676} & 0.673 & \textbf{0.589} & 0.577 & \textbf{0.785} & 0.569 & 0.496 \\
2.5 & 0.40 & 0.50 & 0.654 & \textbf{0.910} & 0.901 & \textbf{0.848} & 0.847 & \textbf{0.900} & 0.879 & \textbf{0.945} & 0.917 & 0.934 \\
2.6 & 0.46 & 0.56 & 0.655 & \textbf{0.571} & 0.568 & \textbf{0.538} & 0.533 & \textbf{0.626} & 0.591 & \textbf{0.792} & 0.629 & 0.743 \\
2.7 & 0.54 & 0.64 & 0.663 & 0.523 & \textbf{0.566} & 0.571 & \textbf{0.589} & 0.515 & \textbf{0.583} & 0.182 & \textbf{0.297} & 0.118 \\
\specialrule{\heavyrulewidth}{1pt}{1pt}
\multicolumn{13}{c}{\textbf{Case 3: $\theta_h=0.4$ with random $\bD_h$}} \\
\addlinespace[0.5ex]
3.1 & 0.24 & 0.34 & 0.710 & 0.706 & \textbf{0.709} & 0.711 & 0.711 & 0.616 & \textbf{0.709} & 0.468 & \textbf{0.711} & 0.709 \\
3.2 & 0.26 & 0.36 & 0.695 & 0.680 & \textbf{0.693} & \textbf{0.697} & 0.696 & 0.523 & \textbf{0.696} & 0.351 & \textbf{0.696} & 0.689 \\
3.3 & 0.30 & 0.40 & 0.667 & 0.548 & \textbf{0.632} & \textbf{0.672} & 0.667 & 0.457 & \textbf{0.632} & 0.420 & \textbf{0.630} & 0.575 \\
3.4 & 0.32 & 0.42 & 0.664 & 0.515 & \textbf{0.584} & \textbf{0.674} & 0.673 & \textbf{0.602} & 0.592 & \textbf{0.783} & 0.584 & 0.500 \\
3.5 & 0.40 & 0.50 & 0.662 & \textbf{0.921} & 0.911 & \textbf{0.854} & 0.853 & \textbf{0.909} & 0.887 & \textbf{0.955} & 0.925 & 0.944 \\
3.6 & 0.46 & 0.56 & 0.650 & \textbf{0.574} & 0.571 & \textbf{0.542} & 0.534 & \textbf{0.627} & 0.596 & \textbf{0.791} & 0.643 & 0.745 \\
3.7 & 0.54 & 0.64 & 0.675 & 0.500 & \textbf{0.537} & 0.567 & \textbf{0.600} & 0.513 & \textbf{0.574} & 0.167 & \textbf{0.317} & 0.117 \\
\bottomrule
\end{tabular}%
}
\vspace{0.5ex}
\end{table}

When $\theta=0.34$ or $0.44$, the concurrent control rate exceeds the historical rate, so borrowing pulls the control estimate downward and can inflate the estimated treatment effect. Accordingly, WOW-gated methods yield slightly lower power than their non-gated counterparts, reflecting reduced reliance on incompatible historical information. This interpretation is supported by Figure~\ref{fig:RelativeBias}, where non-gated methods show considerable bias in these scenarios. By limiting borrowing when incompatibility is detected, WOW provides a safeguard against misleading inference and supports more credible integration of external evidence.

Cases~2 and~3 show similar patterns with $\theta_h=0.4$. Under prior--data conflict, WOW-gated methods mitigate the adverse impact of incompatible borrowing by preventing external information from exaggerating treatment evidence. When the concurrent and historical controls are compatible, WOW-gated methods maintain power comparable to their non-gated counterparts. The stochastic historical-data setting in Case~3 shows a similar pattern to Case~2, suggesting that the observed gains are not driven by fixing $D_h$. 

Additional operating characteristics under the common cutoff  $C=0.95$ are reported in Tables~S5--S6. These results support the calibrated-power findings: WOW-gated methods reduce the impact of incompatible borrowing on power and type I error, with especially clear improvements for more aggressive borrowing rules such as PIP.

\xx

\yx
\section{Real-Data Example}
\label{sec:real_data}

We considered treatment response as the primary binary endpoint in an RCT, with external placebo data as a potential source of borrowing. The example was adapted from the RBesT dataset \citep{RBEST}, which contains placebo response rates from studies of ankylosing spondylitis. For illustration, Study~6 was treated as the concurrent control dataset, with $n=20$ and $x=6$ responders, whereas Study~7 was treated as the external control dataset, with $n_h=78$ and $x_h=9$ responders. The corresponding observed response rates were 0.300 for the concurrent control and 0.115 for the historical source, indicating substantial discordance.

\begin{figure}[htbp]
\centering
\begin{minipage}[t]{0.48\textwidth}
    \centering
    \includegraphics[width=\linewidth]{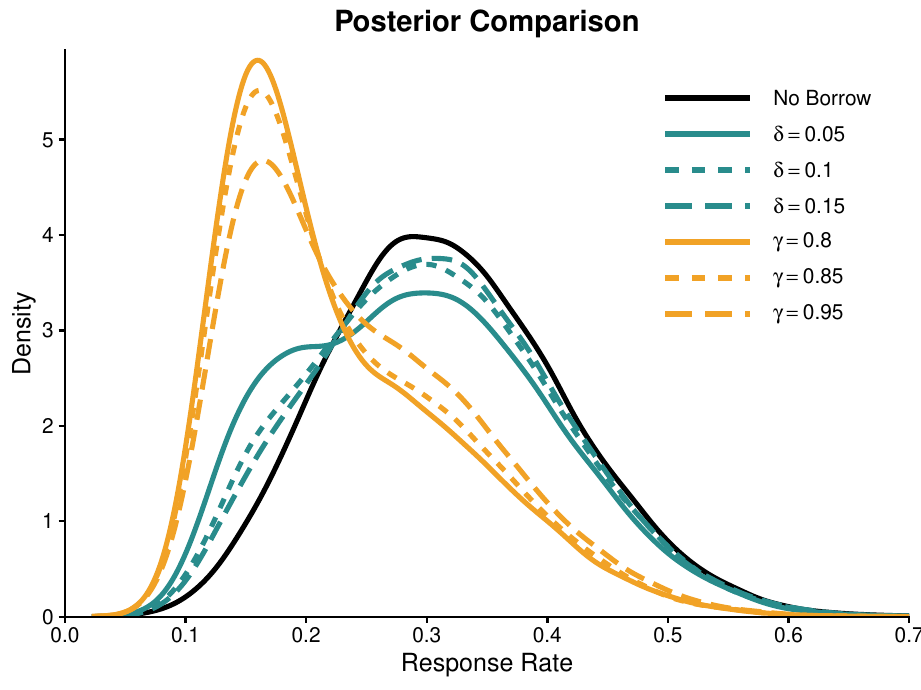}
    \\[2pt]
    {\footnotesize \textbf{(a)} Posterior comparison under no borrowing, SAM, and EB-rMAP.}
\end{minipage}\hfill
\begin{minipage}[t]{0.48\textwidth}
    \centering
    \includegraphics[width=\linewidth]{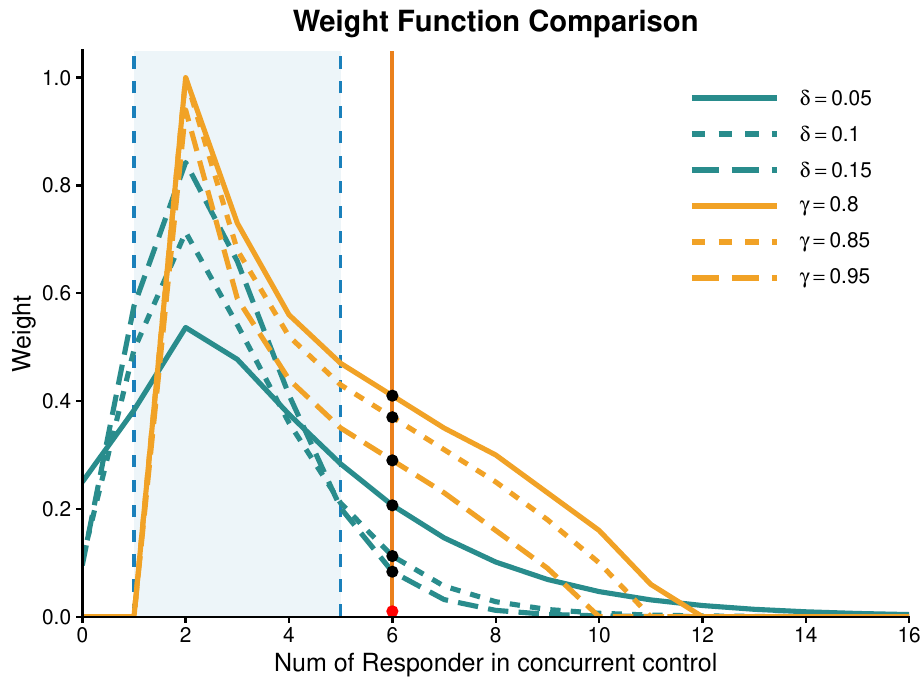}
    \\[2pt]
    {\footnotesize \textbf{(b)} Borrowing weights as a function of the concurrent control response count for varying SAM PPR cutoffs $\delta$ and EB-rMAP PPP thresholds $\gamma$.}
\end{minipage}

\caption{Real-data illustration using the RBesT dataset. 
(a) Posterior comparison under no borrowing, SAM, and EB-rMAP. 
(b) Borrowing weights as a function of the concurrent control response count for representative SAM and EB-rMAP settings.}
\label{fig:realdata_rbest}
\end{figure}

Figure~\ref{fig:realdata_rbest}(a) compares the posterior distributions of the concurrent control response rate under NP, SAM, and EB-rMAP. Because the observed historical response rate was much lower than the concurrent control response rate, borrowing shifted the posterior distribution toward the historical source and away from the no-borrowing posterior. This shift was particularly pronounced for EB-rMAP, indicating that adaptive weighting alone may still allow substantial borrowing from a discordant external control. Panel~(b) shows the borrowing weights as a function of the concurrent control response count $x$ under representative SAM and EB-rMAP hyperparameter settings. The vertical line marks the observed count $x=6$. Across these settings, both SAM and EB-rMAP assigned non-negligible borrowing weights at the observed value. In contrast, WOW rejected borrowing because the observed concurrent control count lay outside the WOW-supported borrowing region; consequently, the gated posterior coincided with the NP posterior. This example illustrates how WOW can prevent borrowing from an apparently incompatible historical source and provide a transparent safeguard against inappropriate external-data integration.

\xx

\section{Discussion}
\label{sec:conc}
This paper introduces WOW, a data-driven gating strategy that strengthens the robustness of external data borrowing using mixture priors in clinical trials. By leveraging WAIC-based model comparison, WOW serves as a preliminary assessment step to evaluate the compatibility between historical and concurrent trial data, \yx determining whether borrowing should proceed before applying any downstream borrowing rule.
This two-step structure emphasizes the modularity of WOW: the gating step is independent of the downstream borrowing method and can be integrated with fixed or adaptive weighting approaches. A practical advantage of the WAIC-based gate is that it avoids direct computation of marginal likelihoods or Bayes factors. This feature may be particularly useful in more complex nonconjugate settings, including many survival models, where closed-form marginal likelihoods are typically unavailable except in special cases such as simple exponential models. \xx Across simulations, WOW mitigates bias and error inflation in the presence of prior-data conflict while preserving comparable performance when borrowing is supported. 

In practical applications, WOW offers a transparent and reproducible approach that complements the growing demand for rigorous external evidence integration. Regulatory agencies, including the FDA, emphasize that borrowing should be justified by sound methodology with explicit attention to assumptions about data comparability. As described in FDA guidance, study-level exchangeability is critical for leveraging prior data while accounting for potential differences \citep{us2021considerations}. 
\yx  WOW addresses this need by providing a data-driven mechanism for assessing borrowing eligibility based on empirical data congruence. The RBesT binary responder example further shows that WOW can prevent borrowing from a clearly discordant historical source.
\xx  Although WOW may lead to modest efficiency loss when external and trial data are highly compatible, it improves robustness and protects against error inflation when they are not, aligning with regulatory expectations for reliable methodology. Moreover, WOW relies on aggregate-level statistics rather than patient-level data, simplifying implementation while supporting privacy protection and regulatory feasibility. This feature is particularly relevant in real-world evidence settings, where data alignment and justification are critical for regulatory acceptance.

\yx 
A limitation of the current work is that the theoretical development is presented under a simple i.i.d.\ setting without covariate adjustment; therefore, the compatibility assessment is conducted at the marginal rather than conditional level. Accordingly, the WAIC-based gate should be interpreted as an empirical compatibility check rather than as evidence of exchangeability. It does not eliminate potential bias from unmeasured confounding or replace clinical and regulatory justification for using external data. Extending WOW to covariate-adjusted settings is therefore an interesting direction for future work. For example, incorporating propensity score weighting, matching, or outcome-regression approaches could help account for covariate shift between historical and concurrent populations \citep{fu2023covariate} and broaden the practical applicability of the gate-then-borrow framework. Future work could also extend WOW to time-to-event endpoints, platform or multi-stage trial settings involving longitudinal or stage-wise borrowing, and other borrowing frameworks such as power priors and commensurate priors.
\xx

\section*{Acknowledgements}
Shouhao Zhou is supported in part by NIH Grant U24MD020517, Pennsylvania Department of Health TSF CURE Program, and Four Diamonds Faculty Research Award. 
Yanxun Xu is supported in part by National Institute of Health grants R01MH128085 and R01AI197147, and National Science Foundation grant  DMS-2610267.

\bibliographystyle{apalike}
\bibliography{reference.bib}

\label{lastpage}

\clearpage

\begingroup
\counterwithout{equation}{section}
\setcounter{section}{0}
\setcounter{equation}{0}
\setcounter{figure}{0}
\setcounter{table}{0}
\setcounter{algorithm}{0}

\renewcommand{\thesection}{\Alph{section}}
\renewcommand{\theequation}{S\arabic{equation}}
\renewcommand{\thefigure}{S\arabic{figure}}
\renewcommand{\thetable}{S\arabic{table}}
\renewcommand{\thealgorithm}{S\arabic{algorithm}}
\renewcommand{\figurename}{Figure}
\renewcommand{\tablename}{Table}

\begin{center}
  \Large\bfseries Supplementary Materials for ``WAIC-Optimized Weight Gating for Mixture Priors in External Data Borrowing''\\[0.75ex]
  \normalsize Shouhao Zhou, Qiuxin Gao, Chenqi Fu, and Yanxun Xu
\end{center}

\section{Proof of Theorem 2}
\label{sec:suppA}

\yx
\setcounter{theorem}{1}

\begin{theorem}[Minimization of WAIC]
For independent datasets $D$ and $D_h$, the  \(\text{WAIC}\) for a mixture model with borrowing weight $w_h$ is a quadratic concave function of the posterior mixture weight $w_h^*$. Since $w_h^*$ is a monotonic function of the prior borrowing weight $w_h$ mapping the interval $[0,1]$ onto itself, the WAIC achieves its minimum at the boundaries $w_h=0$ or $w_h=1$. 
\label{thm:waic-min}

\end{theorem}
\xx

\begin{proof}
With $f(y_i\mid\theta)$ denoting the log-likelihood contribution of observation \(y_i\), the WAIC \citep{watanabe2009algebraic} is defined as:

\begin{equation}
\label{eq:WAIC_org}
    \mathrm{WAIC}(w_h, D,D_h)
  = -2\sum_{i=1}^{n}\mathbb{E}_{p(\theta\mid D,D_h)}\!\bigl[f(y_i\mid\theta)\bigr]
    +2\sum_{i=1}^{n}\mathrm{Var}_{p(\theta\mid D,D_h)}\!\bigl[f(y_i\mid\theta)\bigr].
\end{equation}

The posterior distribution is a mixture of the complete borrowing posterior \( p_{h}(\theta \mid D, D_h) \) and the non-borrowing posterior \( p_{0}(\theta \mid D) \), with weight \( w_h^* \):
\[
p(\theta\mid D,D_h)=w_h^{*}\,p_{h}(\theta\mid D,D_h)+(1-w_h^{*})\,p_{0}(\theta\mid D).
\]

Substituting this mixture into the WAIC formula and expanding, we obtain:


\begin{equation}
\label{eq:WAIC_quad_long}
\begin{aligned}
\mathrm{WAIC}(w_h,D,D_h)
&= -2\sum_{i=1}^{n}\Bigl[\mathbb{E}_{p_{h}}\{f(y_i\mid\theta)\}
                         -\mathbb{E}_{p_{0}}\{f(y_i\mid\theta)\}\Bigr]^{2}(w_h^{*})^{2}\\
&\quad+2\sum_{i=1}^{n}\left[\mathbb{E}_{p_{h}}\{f^{2}(y_i\mid\theta)\}
      -\mathbb{E}_{p_{0}}\{f^{2}(y_i\mid\theta)\}
      -2\,\mathbb{E}_{p_{h}}\{f(y_i\mid\theta)\}\,\mathbb{E}_{p_{0}}\{f(y_i\mid\theta)\} \right. \\
&\hspace{30mm}
     \left. +\{\mathbb{E}_{p_{0}}\{f(y_i\mid\theta)\}\}^{2}
      -\mathbb{E}_{p_{h}}\{f(y_i\mid\theta)\}
      +\mathbb{E}_{p_{0}}\{f(y_i\mid\theta)\}\right]\,w_h^{*}
      + L,
\end{aligned}
\end{equation}
where $L$ is a function of $D$ and $D_h$ only.

\yx
Observe that the coefficient of the quadratic term in \eqref{eq:WAIC_quad_long} is non-negative, confirming that \(\mathrm{WAIC}\) is concave with respect to the posterior weight $w_h^*$. Consequently, its minimum over the unit interval $w_h^* \in [0,1]$ must occur at the boundaries $w_h^* = 0$ or $w_h^* = 1$.

Finally, we consider the relationship between $w_h^*$ and the prior weight $w_h$:
\[
w_h^* = \frac{w_h z_h}{w_h z_h + (1 - w_h) z_0}.
\]
Since the marginal likelihoods $z_h$ and $z_0$ are positive and independent of $w_h$, the derivative 
\[
\frac{d w_h^*}{d w_h} = \frac{z_h z_0}{(w_h z_h + (1 - w_h) z_0)^2} > 0.
\]
The derivative is strictly positive. This confirms that $w_h^*$ is a strictly increasing function of $w_h$ that maps $w_h \in [0,1]$ onto $w_h^* \in [0,1]$. Therefore, the boundary optima $w_h^* \in \{0, 1\}$ correspond directly to the prior weight boundaries $w_h \in \{0, 1\}$. This completes the proof that the gating decision relies solely on a comparison between the no-borrowing ($w_h=0$) and full-borrowing ($w_h=1$) scenarios.

\xx




\end{proof}

\section{Binary Endpoints: WAIC Derivation}
\phantomsection         
\label{sec:binaryWAIC}

The posterior distribution for the binary case is a mixture of two Beta distributions:
\begin{equation}
\label{eq:suppBpos}
  p(\theta \mid D, D_h) = \wst p_h(\theta \mid D,D_h) + (1-\wst) p_0(\theta \mid D),  
\end{equation}
where the components are defined as follows. The partial borrowing posterior is given by $p_h(\theta \mid D,D_h) = Beta(a_h, b_h)$ with parameters $a_h=
a+ x+x_h$ and $b_h=b +n +n_h -x-x_h$. The non-borrowing posterior is  $p_0(\theta \mid D) = Beta(a_0, b_0)$ with $a_0=a+x$ and $b_0=b+n-x$. The mixture weight takes the form $\wst=\frac{w_h z_h}{(w_h z_h+(1-w_h) z_0)}$, with $z_0=\frac{B(a_0, b_0)}{B(a, b)}$ and $z_h = \frac{B\left(a_h, b_h\right)}{B\left(a+x_h, b+n_h-x_h\right)}$, where $B(\cdot, \cdot)$ denotes the beta function.  For simplicity, we write expectations/variances as $\mathbb{E}_{p_h},\operatorname{Var}_{p_h}$, etc.

For a Bernoulli likelihood,  $\log f(y_i \mid \theta)=y_i\log\theta + (1-y_i)\log(1-\theta)$. 
Substituting the mixture posterior \eqref{eq:suppBpos} into \eqref{eq:WAIC_quad_long}  yields a quadratic expression in $\wst$:
\begin{equation}
\mathrm{WAIC}_B(\wst,D,D_h) = -I_1 \cdot {\wst}^2 + I_2 \cdot \wst + I_3,\label{eq: WAIC_binary}
\end{equation}
where 

$$
\begin{aligned}
I_1 \;=\;& 2 \Bigl[ (n-x)
           \Bigl\{\,
             \mathbb{E}_{p_h}\!\bigl\{\log(1-\theta)\bigr\}-
             \mathbb{E}_{p_0}\!\bigl\{\log(1-\theta)\bigr\}
           \Bigr\}^{2}
         + x\,
           \Bigl\{\,
             \mathbb{E}_{p_h}\!\bigl\{\log\theta\bigr\}-
             \mathbb{E}_{p_0}\!\bigl\{\log\theta\bigr\}
           \Bigr\}^{2} \Bigr];\\[6pt]
I_2 \;=\;& 2(n-x)\Bigl[
             \operatorname{Var}_{p_h}\!\bigl\{\log(1-\theta)\bigr\}-
             \operatorname{Var}_{p_0}\!\bigl\{\log(1-\theta)\bigr\} \\[-2pt]
         &\hphantom{(n-x)\Bigl[}+
             \Bigl\{\,
               \mathbb{E}_{p_h}\!\bigl\{\log(1-\theta)\bigr\}-
               \mathbb{E}_{p_0}\!\bigl\{\log(1-\theta)\bigr\}
             \Bigr\}^{2}
             +\mathbb{E}_{p_0}\!\bigl\{\log(1-\theta)\bigr\}-
              \mathbb{E}_{p_h}\!\bigl\{\log(1-\theta)\bigr\}
           \Bigr] \\[2pt]
         &+\,2x\Bigl[
             \operatorname{Var}_{p_h}\!\bigl\{\log\theta\bigr\}-
             \operatorname{Var}_{p_0}\!\bigl\{\log\theta\bigr\} \\[-2pt]
         &\hphantom{+\,x\Bigl[}+
             \Bigl\{\,
               \mathbb{E}_{p_h}\!\bigl\{\log\theta\bigr\}-
               \mathbb{E}_{p_0}\!\bigl\{\log\theta\bigr\}
             \Bigr\}^{2}
             +\mathbb{E}_{p_0}\!\bigl\{\log\theta\bigr\}-
              \mathbb{E}_{p_h}\!\bigl\{\log\theta\bigr\}
           \Bigr] ;
\end{aligned}
$$

$I_3$ corresponds to $L$ in \eqref{eq:WAIC_quad_long}, which is a constant independent of $\wst$.

The expectations and variances can be computed exactly using properties of the Beta distribution. 
Let $\psi(\,\cdot\,)$ and $\psi_1(\,\cdot\,)$ denote the digamma and trigamma functions, we have 
\(\psi(p)-\psi(p+q)=\displaystyle-\sum_{i=0}^{q-1}\tfrac1{p+i}\) and  
\(\psi_1(p)-\psi_1(p+q)=\displaystyle\sum_{i=0}^{q-1}\tfrac1{(p+i)^2}\). Then terms in the $I_1$ and  $I_2$ can be computed in the following closed–form:

\[
\begin{aligned}
\mathbb{E}_{p_h}\!\bigl[\log(1-\theta)\bigr]
  &=\psi(b_h)-\psi(a_h+b_h)
  =-\sum_{i=0}^{a_h-1}\frac{1}{b_h+i},\\
\mathbb{E}_{p_0}\!\bigl[\log(1-\theta)\bigr]
  &=\psi(b_0)-\psi(a_0+b_0)
  =-\sum_{i=0}^{a_0-1}\frac{1}{b_0+i},\\
\mathbb{E}_{p_h}\!\bigl[\log\theta\bigr]
  &=\psi(a_h)-\psi(a_h+b_h)
  =-\sum_{i=0}^{b_h-1}\frac{1}{a_h+i},\\
\mathbb{E}_{p_0}\!\bigl[\log\theta\bigr]
  &=\psi(a_0)-\psi(a_0+b_0)
  =-\sum_{i=0}^{b_0-1}\frac{1}{a_0+i},
\end{aligned}
\]

\[
\begin{aligned}
\operatorname{Var}_{p_h}\!\bigl[\log(1-\theta)\bigr]
  &=\psi_1(b_h)-\psi_1(a_h+b_h)
  =\sum_{i=0}^{a_h-1}\frac{1}{(b_h+i)^2},\\
\operatorname{Var}_{p_0}\!\bigl[\log(1-\theta)\bigr]
  &=\psi_1(b_0)-\psi_1(a_0+b_0)
  =\sum_{i=0}^{a_0-1}\frac{1}{(b_0+i)^2},\\
\operatorname{Var}_{p_h}\!\bigl[\log\theta\bigr]
  &=\psi_1(a_h)-\psi_1(a_h+b_h)
  =\sum_{i=0}^{b_h-1}\frac{1}{(a_h+i)^2},\\
\operatorname{Var}_{p_0}\!\bigl[\log\theta\bigr]
  &=\psi_1(a_0)-\psi_1(a_0+b_0)
  =\sum_{i=0}^{b_0-1}\frac{1}{(a_0+i)^2}.
\end{aligned}
\]

\yx 
To further simplify the expression of $I_1$ and $I_2$ to help the reader apply the approach, we define:
\[
\begin{aligned}
\Delta_{1-\theta}
&=
\Bigl[\psi^{(0)}(b_h)-\psi^{(0)}(a_h+b_h)\Bigr]
-
\Bigl[\psi^{(0)}(b_0)-\psi^{(0)}(a_0+b_0)\Bigr],\\[4pt]
\Delta_{\theta}
&=
\Bigl[\psi^{(0)}(a_h)-\psi^{(0)}(a_h+b_h)\Bigr]
-
\Bigl[\psi^{(0)}(a_0)-\psi^{(0)}(a_0+b_0)\Bigr],\\[4pt]
V_{1-\theta}
&=
\Bigl[\psi^{(1)}(b_h)-\psi^{(1)}(a_h+b_h)\Bigr]
-
\Bigl[\psi^{(1)}(b_0)-\psi^{(1)}(a_0+b_0)\Bigr],\\[4pt]
V_{\theta}
&=
\Bigl[\psi^{(1)}(a_h)-\psi^{(1)}(a_h+b_h)\Bigr]
-
\Bigl[\psi^{(1)}(a_0)-\psi^{(1)}(a_0+b_0)\Bigr].
\end{aligned}
\]
Therefore, the simplified expressions for \(I_1\) and \(I_2\) are
\[
\begin{aligned}
I_1
&=
2\Bigl\{(n-x)\Delta_{1-\theta}^{2}
+
x\Delta_{\theta}^{2}\Bigr\},\\[6pt]
I_2
&=
2(n-x)\Bigl\{V_{1-\theta}+\Delta_{1-\theta}^{2}-\Delta_{1-\theta}\Bigr\}
+
2x\Bigl\{V_{\theta}+\Delta_{\theta}^{2}-\Delta_{\theta}\Bigr\}.
\end{aligned}
\]
\xx

\section{Proof of Theorem 3}
\phantomsection         
\label{sec:suppproof3}

\begin{theorem}[Existence of a Single Connected Borrowing Region]
Given historical data $D_h$ with $(n_h, x_h)$ and a fixed concurrent control sample size $n$, there exists a single connected region \( G = [x_L^*, x_U^*] \subseteq [0, n] \) such that borrowing from the historical data is beneficial if and only if \( x \in G \), where \( x \) is the number of observed successes in the concurrent control.
\end{theorem}

\begin{proof}
According to Theorem \ref{thm:waic-min}, the WAIC-optimal weight $w_h^*$ must be either 0 or 1 due to the quadratic form's properties. For any observed number of responders $x$, given the concurrent control sample size $n$, historical sample size $n_h$, and historical responders $x_h$, we define the key comparison function:
\begin{equation}\label{eq:kx}
k(x) := \mathrm{WAIC}(1, D, D_h) - \mathrm{WAIC}(0, D, D_h) = -I_1 + I_2,
\end{equation}
where borrowing from historical data is beneficial  when $k(x) \leq 0$.

To prove the existence of a single connected borrowing region $G = [x_L^{*}, x_U^{*}] \subseteq [0, n]$, we establish two sufficient conditions:
\begin{enumerate}
\item First, we demonstrate that $k(x)$ is convex on $[0,n]$ by showing its second-order difference is non-negative. This convexity guarantees that $k(x)$ has at most one local minimum in $[0,n]$, which implies the equation $k(x) = 0$ has at most two solutions. 
\item Second, we prove there exists at least one point $\widetilde{x} \in [0,n]$ where $k(\widetilde{x}) < 0$, ensuring the borrowing region $G$ is non-empty. The intermediate value theorem applied to the continuous function $k(x)$ then guarantees the connectedness of $G$.
\end{enumerate}

\subsection{Proof of condition 1}
The second-order difference of $k(x)$ can be written as:
\begin{equation}
\Delta^{2}k(x)=[k(x+2)-k(x+1)]-[k(x+1)-k(x)], \quad x \in {0,1, \cdots, n-2}.
\end{equation}

According to \eqref{eq:kx}, we express $\Delta^{2}k(x)$ as the sum $g_1(x)+g_2(x)$, where $g_1(x)$ and $g_2(x)$ are:

\[
\begingroup
\setlength{\jot}{4pt}            
\begin{aligned}
    g_{1}(x)
  &= x\Bigl[
        - \frac{1}{\bigl(1+a+x+x_h\bigr)^{2}}
        + \frac{1}{\bigl(a+x+x_h\bigr)^{2}}
        + \frac{1}{\bigl(1+a+x\bigr)^{2}}
        - \frac{1}{\bigl(a+x\bigr)^{2}}
    \Bigr] \\[2pt]
  &\quad + x\Bigl[
        - \frac{1}{\,1+a+x+x_h}
        + \frac{1}{\,a+x+x_h}
        + \frac{1}{\,1+a+x}
        - \frac{1}{\,a+x}
    \Bigr] \\[2pt]
  &\quad + 2\Bigl[
        -\frac{1}{\bigl(1+a+x+x_h\bigr)^{2}}
        +\frac{1}{\bigl(1+a+x\bigr)^{2}}
        -\frac{1}{\,1+a+x+x_h}
        +\frac{1}{\,1+a+x}
    \Bigr], \\[2pt]
g_{2}(x)
  &= (n-x)\Bigl[
        \frac{1}{\bigl(b+n-x-2+n_{h}-x_{h}\bigr)^{2}}
      - \frac{1}{\bigl(b+n-x-1+n_{h}-x_{h}\bigr)^{2}} \\[2pt]
  &\quad - \frac{1}{\bigl(b+n-x-2\bigr)^{2}}
      + \frac{1}{\bigl(b+n-x-1\bigr)^{2}}
    \Bigr] \\[2pt]
  &\quad + (n-x)\Bigl[
        \frac{1}{\,b+n-x-2+n_{h}-x_{h}}
      - \frac{1}{\,b+n-x-1+n_{h}-x_{h}} \\[2pt]
   &\quad - \frac{1}{\,b+n-x-2}
      + \frac{1}{\,b+n-x-1}
    \Bigr] \\[2pt]
  &\quad + 2\Bigl[
        -\frac{1}{\bigl(b+n-x-2+n_{h}-x_{h}\bigr)^{2}}
        +\frac{1}{\bigl(b+n-x-2\bigr)^{2}} \\[2pt]
  &\quad -\frac{1}{\,b+n-x-2+n_{h}-x_{h}}
        +\frac{1}{\,b+n-x-2}
    \Bigr].
\end{aligned}
\endgroup
\]

We now verify that $g_1(x)\geq 0$ and $g_2(x)\geq 0$ for all valid $x$. For clarity, we treat $g_1$ and $g_2$ as functions of additional parameters:  $g_1(x) = g_1(x,x_h)$, and $g_2(x) = g_2(x, n_h - x_h)$. 

\subsubsection{Proof of \texorpdfstring{$g_1(x,x_h)\ge 0$}{g1(x,xh)>=0}}

We proceed by induction on $x_h$. Define:
\[
\begin{aligned}
g_{1}(x,x_{h}) &=
\begin{cases}
\displaystyle g_{1}(x,0), 
  &\text{when } x_{h}=0, \\[6pt]
\displaystyle g_{1}(x,1), 
  &\text{when } x_{h}=1, \\[6pt]
\displaystyle g_{1}(x,1)+\sum_{j=1}^{x_{h}-1}\Delta_{x_{h}} g_{1}(x,j), 
  & \text{when } 2\le x_{h}\le n_{h},
\end{cases} \\[10pt]
\text{where}\quad
\Delta_{x_{h}} g_{1}(x,j) &= g_{1}(x,j+1)-g_{1}(x,j).
\end{aligned}
\]

To prove that $g_1\left(x, x_h\right) \geq 0$, we establish three claims:
(a.1) $g_1(x, 0)=0$;
(a.2) $g_1(x, 1) \geq 0$;
(a.3) the forward difference $\Delta_{x_{h}} g_1(x, j)=g_1(x, j+1)-g_1(x, j)$ satisfies $\Delta_{x_{h}} g_1(x, j) \geq 0$ for all $j \geq 1$.

For claim (a.1), we expand  $g_1(x, 0)$ as  follows: 
$$
\begin{aligned}
    g_{1}(x, 0)
  &= x\Bigl[
        - \frac{1}{\bigl(1+a+x\bigr)^{2}}
        + \frac{1}{\bigl(a+x\bigr)^{2}}
        + \frac{1}{\bigl(1+a+x\bigr)^{2}}
        - \frac{1}{\bigl(a+x\bigr)^{2}}
    \Bigr] \\[2pt]
  &\quad + x\Bigl[
        - \frac{1}{\,1+a+x}
        + \frac{1}{\,a+x}
        + \frac{1}{\,1+a+x}
        - \frac{1}{\,a+x}
    \Bigr] \\[2pt]
  &\quad + 2\Bigl[
        -\frac{1}{\bigl(1+a+x\bigr)^{2}}
        +\frac{1}{\bigl(1+a+x\bigr)^{2}}
        -\frac{1}{\,1+a+x}
        +\frac{1}{\,1+a+x}
    \Bigr].\\[2pt]
\end{aligned}
$$
Observing the expansion, we see that $g_1(x, 0) = 0$. Each positive term cancels with its corresponding negative counterpart. For example, in the first line,  the term $\frac{1}{(1+a+x)^2}$ cancels with  $-\frac{1}{(1+a+x)^2}$. Thus $g_{1}(x, 0)=0$ for all $x$.  

For claim (a.2),  we expand $g_{1}(x, 1)$ as follows:
$$\begin{aligned} 
g_1(x,1)= & x\left[-\frac{1}{(a+x+2)^2}+\frac{1}{(a+x+1)^2}+\frac{1}{(a+x+1)^2}-\frac{1}{(a+x)^2}\right] \\ 
& +x\left[-\frac{1}{a+x+2}+\frac{1}{a+x+1}+\frac{1}{a+x+1}-\frac{1}{a+x}\right] \\ 
& +2\left[-\frac{1}{(a+x+2)^2}+\frac{1}{(a+x+1)^2}-\frac{1}{a+x+2}+\frac{1}{a+x+1}\right].
\end{aligned}
$$
After common-denominator expansion and simplification:
$$\begin{aligned} 
g_1(x,1) & =  \frac{M_0(x)}{(a+x)^2(1+a+x)^2(2+a+x)^2}, \\ 
\text{where} \\
M_0(x) &= a^4+5 a^3+5 a^2+(a-1)\left[x\left(3 a^2+6 a+2\right)+x^2(3 a+3)+x^3\right].
\end{aligned}
$$
With the Beta$(1,1)$ prior (i.e., $a=b=1$), the denominator is positive, and $M_{0}(x)\geq 0$ for $x\ge 0$. Therefore, $g_1(x,1)\geq 0$.

For claim (a.3), we expand the forward difference as:
$$
\begin{aligned}
   \Delta_{x_h} g_1\left(x, j\right) &= g_1\left(x, j+1\right) - g_1\left(x, j\right)\\
  &= x\Bigl[
        - \frac{1}{\bigl(2+a+x+j\bigr)^{2}}
        + \frac{1}{\bigl(1+a+x+j\bigr)^{2}}
        + \frac{1}{\bigl(1+a+x+j\bigr)^{2}}
        - \frac{1}{\bigl(a+x+j\bigr)^{2}}
    \Bigr] \\[2pt]
  &\quad + x\Bigl[
        - \frac{1}{\,2+a+x+j}
        + \frac{1}{\,1+a+x+j}
        + \frac{1}{\,1+a+x+j}
        - \frac{1}{\,a+x+j}
    \Bigr] \\[2pt]
  &\quad + 2\Bigl[
        -\frac{1}{\bigl(2+a+x+j\bigr)^{2}}
        -\frac{1}{2+a+x+j}
        +\frac{1}{\bigl(1+a+x+j\bigr)^{2}}
        +\frac{1}{\,1+a+x+j}
    \Bigr]. 
\end{aligned}
$$

After arrangement, we have:
\[
\begin{aligned}
\Delta_{x_h} g_1(x,j)
  &= 
  \frac{M_1(x)}
       {(x+a+j)^{2}(1+x+a+j)^{2}(2+x+a+j)^{2}}\,, \\[6pt]
M_1(x)
  &=
  2\left[\left((a+j)-1\right) x^3+3\left((a+j)^2+(a+j)-1\right) x^2 \right.\\
  &+\left. \left(3 (a+j)^3+9 (a+j)^2+2 (a+j)-2\right) x+(a+j)^2\left((a+j)^2+5 (a+j)+5\right)\right] .
\end{aligned}
\]

The denominator is always positive. So the sign of $\Delta_{x_h} g_1(x,j)$ is determined by its numerator $M_1(x)$. We can readily verify $M_1(x)$ is increasing in $x$ because its derivative is non-negative. Therefore, $$M_1(x) \ge M_1(0) = \left( 2(a+j)^4+10 (a+j)^3+10 (a+j)^2\right) > 0.$$




Combining these results, $g_1\left(x, x_h\right) \geq 0$ for all $x$ and $x_h$.

\subsubsection{Proof of \texorpdfstring{$g_2(x,n_h-x_h)\ge 0$}{g2(x,nh-xh)>=0}}

Similar to the proof of $g_1(x,x_h)$, we write the $g_2(x,n_h - x_h)$ as forward induction of $n_h - x_h$:
\[
\begin{aligned}
g_{2}(x,n_h- x_{h}) &=
\begin{cases}
\displaystyle g_{2}(x,0), 
  &\text{if } n_h - x_{h}=0, \\[6pt]
\displaystyle g_{2}(x,1), 
  &\text{if } n_h - x_{h}=1, \\[6pt]
\displaystyle g_{2}(x,1)+\sum_{j=1}^{n_h - x_{h}-1}
          \Delta_{n_h - x_{h}} g_{2}(x,j), 
  &\text{if } 2\le n_h - x_{h}\le n_{h},
\end{cases} \\[10pt]
\text{where}\quad
\Delta_{n_h - x_{h}} g_{2}(x,j) &= g_{2}(x,j+1)-g_{2}(x,j).
\end{aligned}
\]

We establish three analogous claims:
(b.1) $g_2(x, 0)=0$;
(b.2) $g_2(x, 1) \geq 0$;
(b.3) the forward difference $\Delta_{n_h - x_{h}} g_2(x, j)=g_2(x, j+1)-g_2(x, j)$ satisfies $\Delta_{n_h - x_{h}} g_2(x, j) \geq 0$ for all $j \geq 1$.

For claim (b.1), expand the $g_2(x, 0)$ as 
$$
\begin{aligned}
g_{2}(x,0)
  &= (n-x)\Bigl[
        \frac{1}{\bigl(b+n-x-2\bigr)^{2}}
      - \frac{1}{\bigl(b+n-x-1\bigr)^{2}} \\[2pt]
  &\quad - \frac{1}{\bigl(b+n-x-2\bigr)^{2}}
      + \frac{1}{\bigl(b+n-x-1\bigr)^{2}}
    \Bigr] \\[2pt]
  &\quad + (n-x)\Bigl[
        \frac{1}{\,b+n-x-2}
      - \frac{1}{\,b+n-x-1} \\[2pt]
   &\quad - \frac{1}{\,b+n-x-2}
      + \frac{1}{\,b+n-x-1}
    \Bigr] \\[2pt]
  &\quad + 2\Bigl[
        -\frac{1}{\bigl(b+n-x-2\bigr)^{2}}
        +\frac{1}{\bigl(b+n-x-2\bigr)^{2}} \\[2pt]
  &\quad -\frac{1}{\,b+n-x-2}
        +\frac{1}{\,b+n-x-2}
    \Bigr].
\end{aligned}
$$
The explicit expansion of $g_{2}(x,0)$ shows that every positive term is cancelled by its negative counterpart, yielding $g_{2}(x,0)=0$ for all $x$. 

For claim (b.2), taking $n_h - x_h = 1$ into $g_{2}(x, n_h - x_h)$ to get:
$$\begin{aligned} 
g_{2}(x, 1)
  &= (n-x)\Bigl[
        \frac{1}{\bigl(b+n-x-1\bigr)^{2}}
      - \frac{1}{\bigl(b+n-x\bigr)^{2}} \\[2pt]
  &\quad - \frac{1}{\bigl(b+n-x-2\bigr)^{2}}
      + \frac{1}{\bigl(b+n-x-1\bigr)^{2}}
    \Bigr] \\[2pt]
  &\quad + (n-x)\Bigl[
        \frac{1}{\,b+n-x-1}
      - \frac{1}{\,b+n-x} \\[2pt]
   &\quad - \frac{1}{\,b+n-x-2}
      + \frac{1}{\,b+n-x-1}
    \Bigr] \\[2pt]
  &\quad + 2\Bigl[
        -\frac{1}{\bigl(b+n-x-1\bigr)^{2}}
        +\frac{1}{\bigl(b+n-x-2\bigr)^{2}} \\[2pt]
  &\quad -\frac{1}{\,b+n-x-1}
        +\frac{1}{\,b+n-x-2}
    \Bigr].
\end{aligned}
$$
Denote $u = n-x-2$, after common- expansion and simplification:
\[
\begin{aligned}
g_{2}\bigl(x,\,1\bigr)
  &=\frac{M_{3}\bigl(u\bigr)}
          {(b+n-x-2)^{2}\,(b+n-x-1)^{2}\,(b+n-x)^{2}},
\\[6pt]
\text{where} \\\qquad
M_{3}\bigl(u\bigr)
  &= (2 b-2) u^3+\left(6 b^2+6 b-6\right) u^2+\left(6 b^3+18 b^2+4 b-4\right) u+\left(2 b^4+10 b^3+10 b^2\right).
\end{aligned}
\]
 Since the second order difference domain requires $0\le x \le n-2$, so $u = n-x-2 \ge 0$. The denominator is positive for all $n -x -2\ge 0$, and $M_3(u)$ is increasing for $u\ge 0$, so 
$M_3(u) > M_3(0)= 2 b^4+10 b^3+10 b^2 > 0$.

For claim (b.3), we expand the $\Delta_{n_h - x_h} g_2\left(x, j\right)$ as:

$$
\begin{aligned}
   \Delta_{n_h - x_h} g_2\left(x, j\right)
  &= g_2\left(x, j + 1\right)  -  g_2\left(x, j \right)\\
  & = (n-x)\Bigl[
        \frac{1}{\bigl(b+n-x-1+j\bigr)^{2}}
      - \frac{1}{\bigl(b+n-x-2+j\bigr)^{2}} \\[2pt]
  &\quad - \frac{1}{\bigl(b+n-x+j\bigr)^{2}}
      + \frac{1}{\bigl(b+n-x-1+j\bigr)^{2}}
    \Bigr] \\[2pt]
  &\quad + (n-x)\Bigl[
        \frac{1}{\,b+n-x-1+j}
      - \frac{1}{\,b+n-x-2+j} \\[2pt]
   &\quad - \frac{1}{\,b+n-x+j}
      + \frac{1}{\,b+n-x-1+j}
    \Bigr] \\[2pt]
  &\quad + 2\Bigl[
        -\frac{1}{\bigl(b+n-x-1+j\bigr)^{2}}
        +\frac{1}{\bigl(b+n-x-2+j\bigr)^{2}} \\[2pt]
  &\quad -\frac{1}{\,b+n-x-1+j}
        +\frac{1}{\,b+n-x-2+j}
    \Bigr].
\end{aligned}
$$

Denote $ u = n -x- 2 , m = b+j $, after some algebraic rearrangement:
\[
\begin{aligned}
\Delta_{n_h - x_h} g_2(x,j)
  &= 
\frac{M_4(u)}{(b+n-x-2+j)^2(b+n-x-1+j)^2(b+n-x+j)^2}, \\[6pt]
M_4(u)
  &=
  2(m-1) u^3+2\left(3 m^2+3 m-3\right) u^2+2\left(3 m^3+9 m^2+2 m-2\right) u\\
  &+2\left(m^4+5 m^3+5 m^2\right).
\end{aligned}
\]

Since $j \ge 1$, we have $m = b+j \ge 2$. Therefore the denominator of $\Delta_{n_h - x_h} g_2\left(x, j\right)$ is always positive. So the sign of $\Delta_{n_h - x_h} g_2\left(x, n_h - x_h\right)$ is determined by the numerator $M_4(u)$. And we take the derivative of $M_4(u)$ over $u$ as:
$$
\begin{aligned}
\frac{d}{du} M_4\bigl(u\bigr) 
&= 2\left[3(m-1) u^2+2\left(3 m^2+3 m-3\right) u+\left(3 m^3+9 m^2+2 m-2\right)\right] .
\end{aligned}
$$

Because $u \ge 0$  and \(m \ge 2\), every term in $\frac{d}{du} M_4\bigl(u\bigr) $ is non-negative, 
so $M_4(u)$ is increasing in $u$. Hence $M_4(u) \geq M_4(0)=2\left(m^4+5 m^3+5 m^2\right)>0$. The numerator of $\Delta_{n_h-x_h} g_2(x, j)$ is therefore strictly positive, and its denominator is positive as shown above; consequently $\Delta_{n_h-x_h} g_2(x, j) \geq 0$ for every $j \geq 1$.
Combining this with $g_2(x, 1) \geq 0$ yields

$$
g_2\left(x, n_h-x_h\right)=g_2(x, 1)+\sum_{j=1}^{n_h-x_h-1} \Delta_{n_h-x_h} g_2(x, j) \geq 0 \quad \text { for } 2 \leq n_h -x_h \leq n_h .
$$

Together with $g_2(x, 0)=0$ established in (b.1) and $g_2(x,1) \ge 0$ established in (b.2), we conclude that $g_2\left(x, n_h-x_h\right) \geq 0$.

\end{proof}

\subsection{Proof of condition 2}
\label{lemma1}
\begin{proof}
To determine $\widetilde{x}$ such that $k(\widetilde{x}) < 0$, recall that this condition is equivalent to  $\mathrm{WAIC}(1, D, D_h) < \mathrm{WAIC}(0, D, D_h)$, indicating that borrowing is beneficial at $\widetilde{x}$. A natural candidate for $\widetilde{x}$ is the point when the historical mean aligns precisely with the posterior mean of the concurrent control arm under a non-informative prior.  


We express the condition as:
\begin{equation}
\label{eq:tildx}
    \frac{n+n_h+a+b}{x+x_h+a}= \frac{n+a+b}{x+a} = \lambda \quad (\lambda >1),
\end{equation}
whose solution is $$\widetilde{x} = \frac{x_h(n + a+ b)}{n_h} - a.$$ Since $\widetilde{x}$ must satisfy  $0\le \widetilde{x} \le n$, substituting into this constraint yields the admissible range for $x_h$:  $$ \frac{a n_h}{(n+a+b)} \le x_h\le  \frac{(n+a)n_h}{(n+a+b)} .$$ We denote this set of valid $x_h$ values as 
$$
\mathcal{A}_h=\left\{x_h :  \frac{a n_h}{(n+a+b)}  \leq x_h \leq  \frac{(n+a) n_h}{(n+a+b)}  \right\}.
$$

Then we complete the proof in two steps:
\begin{itemize}
\item Case 1: if \( x_h \notin \mathcal{A}_h \), then either \( x_h < \frac{a n_h}{n + a + b} \), where \( k(\widetilde{x} = 0) < 0 \), or \( x_h > \frac{(n + a) n_h}{n + a + b} \), where \( k(\widetilde{x} = n) < 0 \).  
\item Case 2: if $x_h \in \mathcal{A}_h$, the solution $\tilx$ to \eqref{eq:tildx} satisfies $k(\tilx) < 0$.
\end{itemize}


\subsubsection{Case \texorpdfstring{$x_h \notin \mathcal{A}_h$}{xh not in Ah}}

We begin with the condition 
\begin{equation}
\label{eq:C.2.1.cond.1}
      x_h <  \frac{a n_h}{n+a+b}. 
\end{equation}
By \eqref{eq:kx}, we have
$$
\begin{aligned}
    k(0) =  
        &-\psi(b + n + n_h - x_h) + \psi(a + b + n + n_h) 
        + \psi_1(b + n + n_h - x_h) - \psi_1(a + b + n + n_h) \\
        & - \left( -\psi(b + n) + \psi(a + b + n) 
        + \psi_1(b + n) - \psi_1(a + b + n) \right).
\end{aligned}
$$

Using the recurrence relations of the digamma and trigamma functions, \( \psi(z+1) = \psi(z) + \frac{1}{z} \) and \( \psi_1(z+1) = \psi_1(z) - \frac{1}{z^2} \), define  
\[
S(t) = \psi(t+1) - \psi_1(t+1) - \bigl(\psi(t) - \psi_1(t)\bigr) = \frac{1}{t} + \frac{1}{t^2}.
\]  

For a Beta(1,1) prior (i.e., \( a = b = 1 \)), we expand \( k(0) \) as follows:  
\[
\begin{aligned}
k(0) &= \left( \sum_{i=n+2}^{n+n_h+1} S(i) \right) - \sum_{i=n+1}^{n+n_h-x_h} S(i) \\
     &= -S(n+1) - \left[ \sum_{i=n+2}^{n+n_h-x_h} S(i) - \left( \sum_{i=n+2}^{n+n_h+1} S(i) \right) \right] \\
     &= -S(n+1) + \left( \sum_{i=n+n_h-x_h+1}^{n+n_h+1} S(i) \right).
\end{aligned}
\]

When \( x_h = 0 \), the summation simplifies to:  
\[
\begin{aligned}
k(0) &= -S(n+1) + S(n+n_h+1).
\end{aligned}
\]  
Since \( S(t) \) is a decreasing function, it follows that \( S(n+1) > S(n+n_h+1) \), and therefore \( k(0) < 0 \).  

When \( x_h > 0 \), proving \( k(0) < 0 \) reduces to verifying:  
\begin{equation}  
\label{eq:C.2.1.2}  
\sum_{i=n+n_h-x_h+1}^{n+n_h+1} S(i) < S(n+1).
\end{equation}  

In this case, the summation consists of \( x_h + 1 \) terms. Since \( S(t) \) is decreasing, the maximum value in the summation occurs at the smallest index \( i = n + n_h - x_h + 1 \). Thus, an upper bound for the summation is:  
\[
\sum_{i=n+n_h-x_h+1}^{n+n_h+1} S(i) \leq (x_h + 1) S(n + n_h - x_h + 1).  
\]  

A sufficient condition for \eqref{eq:C.2.1.2} to hold is:  
\[
(x_h + 1) S(n + n_h - x_h + 1) < S(n + 1).
\]

Substituting the explicit form $S(t)=\frac{1}{t}+\frac{1}{t^2}$ into the inequality yields:

$$
(x_h+1)\Bigl[\frac{1}{n+1+n_h-x_h}
+\frac{1}{(n+1+n_h-x_h)^2}\Bigr]<\frac{1}{n+1}+\frac{1}{(n+1)^2},
$$

which is equivalent to

\begin{equation}
\label{eq:C.2.1.1}
      x_h+1<\frac{(n+2)}{(n+1)^2} \times \frac{\left(n+1+n_h-x_h\right)^2}{n+2+n_h-x_h}.
\end{equation}

By \eqref{eq:C.2.1.cond.1} we have \(x_h < \frac{n_h}{n+2}\), hence \(n_h > x_h(n+2)\), and therefore
\[
n+1+n_h-x_h > n+1+(n+2)x_h - x_h = (n+1)(x_h+1).
\]
Consider the function \(\eta(x)=\frac{x^2}{x+1}\), which is strictly increasing for \(x>0\). Using this monotonicity gives a lower bound for the right-hand side of \eqref{eq:C.2.1.1}:
\[
\frac{((n+1)+n_h-x_h)^2}{(n+1)+n_h-x_h+1}
= \eta(n+1+n_h-x_h)
> \eta\bigl((n+1)(x_h+1)\bigr)
= \frac{(n+1)^2(x_h+1)^2}{(n+1)(x_h+1)+1}.
\]
Consequently, a sufficient condition for \eqref{eq:C.2.1.1} is
\[
x_h+1 \;<\; \frac{n+2}{(n+1)^2}\,
\frac{(n+1)^2(x_h+1)^2}{(n+1)(x_h+1)+1}.
\]
Since \(x_h\ge 0\) implies \(x_h+1>0\), dividing both sides by \(x_h+1\) yields
\[
1 \;<\; \frac{(n+2)(x_h+1)}{(n+1)(x_h+1)+1}
\;\Longleftrightarrow\;
0 \;<\; x_h,
\]
and thus \(k(0)<0\).

Second, we show that if
\begin{equation}
\label{eq:C.2.2.cond.1}
x_h > \frac{(n + a) n_h}{n + a + b},
\end{equation}
then $k(n) < 0$. By \eqref{eq:kx}, 
$$
\begin{aligned}
    k(n) =  &-\psi(a + n + x_h) + \psi(a + b + n + n_h) 
        + \psi_1(a + n + x_h) - \psi_1(a + b + n + n_h) \\
        &- \left( -\psi(a + n) + \psi(a + b + n) 
        + \psi_1(a + n) - \psi_1(a + b + n) \right). \\
\end{aligned}
$$

With prior Beta$(1,1)$ and the same definition of $S(t)$, we obtain:
$$
\begin{aligned}
     k(n) &= \sum_{i=n+2}^{n+n_h+1} S(i) - \sum_{i=n+1}^{n+x_h} S(i)   \\
       & = \left(  \sum_{i=n+2}^{n+n_h+1} S(i)\right) - \left( S(n+1) + \sum_{i=n+2}^{n+x_h} S(i) \right)  \\
&= \left(  \sum_{i=n+x_h+1}^{n+n_h+1} S(i)\right)  -     S(n+1). 
\end{aligned}
$$
When $n_h - x_h = 0$, we have
$$
\begin{aligned}
     k(n) &= S(n+1+n_h)  -     S(n+1). 
\end{aligned}
$$
Since $S(t)$ is a decreasing function, it follows that $S\left(n+1+n_h\right)<S(n+1)$, and therefore $k(n)<0$.

When $n_h - x_h > 0$, showing $k(n)<0$ is equivalent to:
\begin{equation}
\label{eq:C.2.2.1}
\sum_{i = n + x_h + 1}^{n + n_h + 1} S(i) < S(n + 1).
\end{equation}
The summation comprises $n_h - x_h+1$ terms. Since $S(t)$ is decreasing, the maximum value in the sum occurs at the smallest index $i=n+x_h+1$. Hence, an upper bound of the left-hand side of \eqref{eq:C.2.2.1} is
$$
\left(n_h-x_h+1\right) S\left(n+x_h+1\right)<S(n+1),
$$
Substituting $S(t)=1 / t+1 / t^2$ gives
$$
\left(n_h-x_h+1\right)\left(\frac{1}{n+1+x_h}+\frac{1}{\left(n+1+x_h\right)^2}\right)<\frac{1}{n+1}+\frac{1}{(n+1)^2},
$$
which is equivalent to
\begin{equation}
\label{eq:C.2.2.2}
n_h-x_h+1<\frac{n+2}{(n+1)^2} \frac{\left(n+1+x_h\right)^2}{n+2+x_h}.
\end{equation}

By \eqref{eq:C.2.2.cond.1}, $x_h>\frac{(n+1) n_h}{n+2}$, implying $x_h>(n+1)\left(n_h-x_h\right)$, and thus
$$
n+1+x_h>n+1+(n+1)\left(n_h-x_h\right)=(n+1)\left(n_h-x_h+1\right) .
$$
By monotonicity of $\eta(x)$, the right-hand side of the \eqref{eq:C.2.2.2} admits a lower bound:
$$
\frac{((n+1) +x_h)^2}{n+1+x_h + 1} = \eta(n+1+x_h) > \eta((n+1)(n_h-x_h+1)) =  \frac{(n+1)^2(n_h-x_h+1)^2}{(n+1)(n_h-x_h+1)+1}.
$$

Consequently, a sufficient condition for \eqref{eq:C.2.2.2} is
$$
n_h-x_h+1<\frac{n+2}{(n+1)^2} \frac{(n+1)^2\left(n_h-x_h+1\right)^2}{(n+1)\left(n_h-x_h+1\right)+1}.
$$
Since $n_h-x_h+1>0$, dividing both sides by $n_h-x_h+1$ yields

$$
1<\frac{(n+2)\left(n_h-x_h+1\right)}{(n+1)\left(n_h-x_h+1\right)+1} \Longleftrightarrow n_h-x_h>0
$$

and hence $k(n)<0$.


\subsubsection{Case \texorpdfstring{$x_h \in \mathcal{A}_h$}{xh in Ah}}

Following the notation in Section \ref{sec:binaryWAIC},   $k(\tilx)$ can be expressed as:

\[
\begin{aligned}
k(\tilx) =& -2 \Bigl[ \sum_{i=1}^{n}\mathbb{E}_{p_h}\!\bigl\{f(y_i\mid\theta)\bigr\} - \sum_{i=1}^{n}\mathbb{E}_{p_0}\!\bigl\{f(y_i\mid\theta)\bigr\}\Bigr] \\
   & +2\Bigl[\sum_{i=1}^{n}\mathrm{Var}_{p_h}\!\bigl\{f(y_i\mid\theta)\bigr\} - \sum_{i=1}^{n}\mathrm{Var}_{p_0}\!\bigl\{f(y_i\mid\theta)\bigr\}\Bigr] \\
    = & -2 \Bigl[ (n - \tilx) \Bigl\{ \mathbb{E}_{p_{h}}\!\bigl\{\log(1-\theta)\bigr\}-
  \mathbb{E}_{p_{0}}\!\bigl\{\log(1-\theta)\bigr\} \Bigr\}+
 \tilx \Bigl\{ \mathbb{E}_{p_{h}}\!\bigl\{\log\theta\bigr\}-
  \mathbb{E}_{p_{0}}\!\bigl\{\log\theta\bigr\} \Bigr\} \Bigr]\\
    & + 2 \Bigl[ (n - \tilx) \Bigl\{ \operatorname{Var}_{p_{h}}\!\bigl\{\log(1-\theta)\bigr\}
      -\operatorname{Var}_{p_{0}}\!\bigl\{\log(1-\theta)\bigr\} \Bigr\} + \tilx  \Bigl\{ \operatorname{Var}_{p_{h}}\!\bigl\{\log\theta\bigr\}
      -\operatorname{Var}_{p_{0}}\!\bigl\{\log\theta\bigr\} \Bigr\} \Bigr].
\end{aligned}
\]

Denote
$$
\begin{aligned}
  E_0(\tilx) &=
  \mathbb{E}_{p_{h}}\!\bigl[\log\theta\bigr]-
  \mathbb{E}_{p_{0}}\!\bigl[\log\theta\bigr],\\
E_1(\tilx) &=
  \mathbb{E}_{p_{h}}\!\bigl[\log(1-\theta)\bigr]-
  \mathbb{E}_{p_{0}}\!\bigl[\log(1-\theta)\bigr],\\
V_0(\tilx) &=
      \operatorname{Var}_{p_{h}}\!\bigl[\log\theta\bigr]
      -\operatorname{Var}_{p_{0}}\!\bigl[\log\theta\bigr],\\
V_1(\tilx) &=
      \operatorname{Var}_{p_{h}}\!\bigl[\log(1-\theta)\bigr]
      -\operatorname{Var}_{p_{0}}\!\bigl[\log(1-\theta)\bigr]. \\
\end{aligned}
$$

Next, we show that each component of $k(\tilx)$ is negative, i.e., 
$$
\begin{aligned}
-\Bigl[ \mathbb{E}_{p_{h}}\!\bigl\{\log\theta\bigr\}-
  \mathbb{E}_{p_{0}}\!\bigl\{\log\theta\bigr\} \Bigr] = -E_0(\tilx) &< 0, \\
-\Bigl[ \mathbb{E}_{p_{h}}\!\bigl\{\log(1-\theta)\bigr\}-
  \mathbb{E}_{p_{0}}\!\bigl\{\log(1-\theta)\bigr\} \Bigr] = - E_1(\tilx) &< 0, \\
\Bigl[ \operatorname{Var}_{p_{h}}\!\bigl\{\log\theta\bigr\}
      -\operatorname{Var}_{p_{0}}\!\bigl\{\log\theta\bigr\} \Bigr] = V_0(\tilx) &< 0, \\
\Bigl[ \operatorname{Var}_{p_{h}}\!\bigl\{\log(1-\theta)\bigr\}
      -\operatorname{Var}_{p_{0}}\!\bigl\{\log(1-\theta)\bigr\} \Bigr] = V_1(\tilx) &< 0. \\
\end{aligned}
$$

Recall that 
$$
a_0 = a + \tilx,\quad a_h = a + \tilx + x_h, \quad b_0 = b +n - \tilx, \quad b_h = b +n +n_h - \tilx - x_h .
$$
By construction of the scaling constant in \eqref{eq:tildx}, the following relationships hold:
\begin{equation}\label{eq:properties}
a_h + b_h = \lambda a_h, \quad
a_0 + b_0 = \lambda a_0, \quad
a_h + b_h = \frac{\lambda}{\lambda -1}b_h,\quad
a_0 + b_0 = \frac{\lambda}{\lambda -1}b_0, \quad
\lambda > 1.
\end{equation}

Using these relationships, the explicit forms of $E_0(\tilx)$ and $E_1(\tilx)$ are given by:
\[
\begin{aligned}
E_0(\tilx)&=\psi(a_h)-\psi(a_h + b_h) -  \left( \psi(a_0)-\psi(a_0+b_0) \right),\\
E_1(\tilx)&=\psi(b_h)-\psi(a_h + b_h) -  \left( \psi(b_0)-\psi(a_0 + b_0) \right).
\end{aligned}                                                              \]

By substituting \eqref{eq:properties} into \( E_0(\tilx) \), we obtain:  
\[
E_0(\tilx) = \psi(a_h) - \psi(\lambda a_h) - \bigl[\psi(a_0) - \psi(\lambda a_0)\bigr].
\]  

It is clear that \( E_0(\tilx) \) can be interpreted as the difference between two instances of the function, evaluated at \( a_0 \) and \( a_h \), respectively. To simplify the notation, let \( H_0(t) = \psi(t) - \psi(\lambda t) \), and consider the forward difference:  
\[
\begin{aligned}
H_0(x+1) - H_0(x)
  &= \psi(x+1) - \psi(\lambda(x+1)) - \bigl[\psi(x) - \psi(\lambda x)\bigr] \\
  &= \frac{1}{x} - \int_{0}^{\infty} \frac{e^{-\lambda x t} \bigl(1 - e^{-\lambda t}\bigr)}{1 - e^{-t}} dt \\
  &\geq \frac{1}{x} - \int_{0}^{\infty} \lambda e^{-\lambda x t} dt = \frac{1}{x} - \frac{1}{\lambda x} > 0,
\end{aligned}
\]  
where the bound \( \bigl(1 - e^{-\lambda t}\bigr) / \bigl(1 - e^{-t}\bigr) \leq \lambda \) holds for \( \lambda > 1 \). The expressions \( \psi(\lambda x) \) and \( \psi(\lambda(x+1)) \) are evaluated using the integral representation:  
\[
\psi(z) = \int_0^{\infty} \left(\frac{e^{-t}}{t} - \frac{e^{-z t}}{1 - e^{-t}}\right) dt.
\]  

From the condition \( x_h \in \mathcal{A}_h \), it follows that \( a n_h / (n + a + b) \leq x_h \). This guarantees \( x_h > 0 \) since \( a n_h / (n + a + b) > 0 \). Therefore, \( a_0 = a + \widetilde{x} < a + x_h + \widetilde{x} = a_h \).

Since \( H_0(x) \) is strictly increasing and \( a_0 < a_h \), we have \( H_0(a_0) < H_0(a_h) \). Hence,  
\[
E_0(\tilx) = H_0(a_h) - H_0(a_0) > 0.
\]

By substituting \eqref{eq:properties} into \( E_1(\tilx) \), we obtain:  
\[
E_1(\tilx) = \bigl[\psi(b_h) - \psi\bigl(\frac{\lambda}{\lambda - 1} b_h\bigr)\bigr] - \bigl[\psi(b_0) - \psi\bigl(\frac{\lambda}{\lambda - 1} b_0\bigr)\bigr].
\]  

Let \( \lambda^{\prime} = \frac{\lambda}{\lambda - 1} > 1 \). Similar to the proof that \( H_0(x) \) is increasing, the function \( H_1(x) = \psi(x) - \psi(\lambda^{\prime} x) \), where \( \lambda^{\prime} > 1 \), is also increasing.  

From the condition \( x_h \in \mathcal{A}_h \), we know \( x_h \leq \frac{(n + a) n_h}{n + a + b} \). This ensures \( x_h < n_h \), which implies \( n_h - x_h > 0 \). As a result, \( b_h = b + n + n_h - x - x_h > b + n - x = b_0 \).  

Because \( H_1(x) \) is increasing and \( b_h > b_0 \), we conclude:  
\[
E_1(\tilx) = H_1(b_h) - H_1(b_0) > 0.  
\]

By the notation in Section \ref{sec:binaryWAIC}, the explicit forms of \( V_0(\tilx) \) and \( V_1(\tilx) \) are given by:  
\[
\begin{aligned}
V_0(\tilx) &= \psi_1(a_h) - \psi_1(a_h + b_h) - \bigl[\psi_1(a_0) - \psi_1(a_0 + b_0)\bigr], \\
V_1(\tilx) &= \psi_1(b_h) - \psi_1(a_h + b_h) - \bigl[\psi_1(b_0) - \psi_1(a_0 + b_0)\bigr].
\end{aligned}
\]

Substituting \eqref{eq:properties} into \( V_0(\tilx) \), we find:  
\[
V_0(\tilx) = \bigl[\psi_1(a_h) - \psi_1(\lambda a_h)\bigr] - \bigl[\psi_1(a_0) - \psi_1(\lambda a_0)\bigr].
\]  

It is straightforward to observe that \( V_0(\tilx) \) represents the difference between two instances of the function, evaluated at \( a_0 \) and \( a_h \), respectively. 

To simplify the notation, define 
\( H_2(x)=\psi_1(x)-\psi_1(\lambda x) \). Using the recurrence
\(\psi_1(x+1)-\psi_1(x)=-1/x^2\) and the integral representation of
\(\psi_1(\cdot)\), we have
\[
\begin{aligned}
H_2(x+1)-H_2(x)
&=
\bigl[\psi_1(x+1)-\psi_1(x)\bigr]
-
\bigl[\psi_1(\lambda(x+1))-\psi_1(\lambda x)\bigr] \\
&=
-\frac{1}{x^2}
+
\int_{0}^{\infty}
\frac{t e^{-\lambda x t}\bigl(1-e^{-\lambda t}\bigr)}
     {1-e^{-t}}\,dt \\
&\le
-\frac{1}{x^2}
+
\lambda\int_{0}^{\infty} t e^{-\lambda x t}\,dt \\
&=
-\frac{1}{x^2}+\frac{1}{\lambda x^2}<0 ,
\end{aligned}
\]
where the bound
\[
\frac{1-e^{-\lambda t}}{1-e^{-t}}\le \lambda
\]
holds for \(\lambda>1\). Therefore, \(H_2(x)\) is strictly decreasing
along unit increments. Since \(a_h-a_0=x_h>0\), we have
\(H_2(a_h)<H_2(a_0)\). Hence,
\[
V_0(\tilx)=H_2(a_h)-H_2(a_0)<0.
\]



Similarly, substituting \eqref{eq:properties} into \( V_1(\tilx) \), we obtain:  
\[
V_1(\tilx) = \bigl[\psi_1(b_h) - \psi_1\bigl(\frac{\lambda}{\lambda - 1} b_h\bigr)\bigr] - \bigl[\psi_1(b_0) - \psi_1\bigl(\frac{\lambda}{\lambda - 1} b_0\bigr)\bigr].
\]  

Denote \( \lambda^{\prime} = \frac{\lambda}{\lambda - 1} > 1 \). Similar to the proof for \( H_2(x) \), the function \( H_3(x) = \psi_1(x) - \psi_1(\lambda^{\prime} x) \), where \( \lambda^{\prime} > 1 \), is strictly decreasing. Since \( b_h > b_0 \), we conclude:  
\[
V_1(\tilx) = H_3(b_h) - H_3(b_0) < 0.
\]

To summarize, we have demonstrated that:  
\[
\begin{aligned}
-\Bigl[ \mathbb{E}_{p_{h}}\!\bigl\{\log(1-\theta)\bigr\} - \mathbb{E}_{p_{0}}\!\bigl\{\log(1-\theta)\bigr\} \Bigr] = -E_1(\tilx) &< 0, \\
-\Bigl[ \mathbb{E}_{p_{h}}\!\bigl\{\log\theta\bigr\} - \mathbb{E}_{p_{0}}\!\bigl\{\log\theta\bigr\} \Bigr] = -E_0(\tilx) &< 0, \\
\Bigl[ \operatorname{Var}_{p_{h}}\!\bigl\{\log(1-\theta)\bigr\} - \operatorname{Var}_{p_{0}}\!\bigl\{\log(1-\theta)\bigr\} \Bigr] = V_1(\tilx) &< 0, \\
\Bigl[ \operatorname{Var}_{p_{h}}\!\bigl\{\log\theta\bigr\} - \operatorname{Var}_{p_{0}}\!\bigl\{\log\theta\bigr\} \Bigr] = V_0(\tilx) &< 0.
\end{aligned}
\]  
Thus, \( k(\tilx) < 0 \).

\end{proof}

\yx
\section{Binary-Endpoint Illustration of the WOW Procedure}

\renewcommand{\thealgorithm}{S\arabic{algorithm}}
\setcounter{algorithm}{0}

\begin{algorithm}[htbp]
\small
\caption{Implementation of Algorithm~1 for Binary Endpoint Borrowing}
\label{alg:wow-general}
\begin{algorithmic}[1]
\Require Historical data \(D_h=\{x_h, n_h\}\), concurrent control data \(D=\{x, n\}\),
sample space \(\Omega=\{0,\cdots,n\}\), priors \(\pi_0(\theta)=\text{Beta}(a,b)\) and \(\pi_h(\theta\mid D_h)=\text{Beta}(a + x_h, b + n_h - x_h)\),
and a user-preferred borrowing rule \(\mathcal A\) (e.g, \(\mathcal A_{SAM}\) for SAM prior)

\vspace{12pt}   
\Statex \textbf{Step 1: WAIC-optimized gating}

\State Evaluate
\[
k(x)=\mathrm{WAIC}_B(1,D,D_h)-\mathrm{WAIC}_B(0,D,D_h)=-I_1+I_2
\]
over possible \(x\in\Omega\), and construct the  borrowing exclusion region
\[
    \Omega_0
    =
    \{x\in\Omega:
    k(x)>0\}.
\]
\State Set the borrowing-supported region as
$G=\Omega\setminus\Omega_0$.

\vspace{12pt}   
\Statex \textbf{Step 2: external borrowing only if the gate is open.}

\If{\(x\in\Omega_0\)}
    \State Gate closed: set \(w_h=0\) and \(\wst=0\).
    \State \textbf{assign}
    \[
        p^\star(\theta\mid D,D_h) \leftarrow p_0(\theta\mid D).
    \]
\Else
    \State Gate open: compute the downstream prior weight
    \[
        w_h \leftarrow w_h^{\mathcal A}(D,D_h)
    \]
    using the prespecified rule \(\mathcal A\).
    \State Compute the corresponding posterior mixture weight
    \[
        \wst
        =
        \frac{w_h z_h}
        {w_h z_h+(1-w_h)z_0}.
    \]
    \State \textbf{assign}
    \[
        p^\star(\theta\mid D,D_h)
        \leftarrow
        \wst\,p_h(\theta\mid D,D_h)
        +
        (1-\wst)\,p_0(\theta\mid D).
    \]
\EndIf
\vspace{12pt}   

\State \Return $p^{\star}(\theta \mid D, D_h)$ \Statex (If {\(x\) is not observed under the prospective planning setting}, then \Return \(G=\Omega\setminus\Omega_0\) instead to apply WOW-\(\mathcal A\) for prospective use.)
\end{algorithmic}
\end{algorithm}
Algorithm~S1 illustrates the implementation of Algorithm~1 for the binary-endpoint setting. By Theorem~2, the WAIC gate does not require optimization over \(w_h\in[0,1]\); instead, it compares only the two boundary specifications \(w_h=0\) and \(w_h=1\). For binary endpoints, Theorem~3 further shows that the borrowing-supported region is a single interval \(G=[x_L^*,x_U^*]\), which can be pre-tabulated before trial
conduct. The downstream borrowing rule  \(\mathcal A\) (e.g, \(\mathcal A_{SAM}\) for SAM prior) is invoked only when the gate is open; otherwise, the analysis defaults to the no-borrowing posterior \(p_0(\theta\mid D)\).
In implementation,  \(\mathcal A\) can be any fixed-weight rule, such as \(\mathcal A_{Mix50}\) or by another data-adaptive rule, such as \(\mathcal A_{EB-rMAP}\). \xx


\section{Binary-Endpoint Illustration of the Borrowing Region and Downstream Weights}

Figure~2 in the main text displays the borrowing regions and weight profiles for the SAM and EB-rMAP priors across varying historical sample sizes \( n_h \in \{75, 150, 600\} \). In each case, the historical control group is assumed to have a response rate of 0.4, yielding \( x_h = n_h \times 0.4 \) responders. For the concurrent control group, we fix the sample size at \( n = 150 \). For the WOW gating procedure, the identified borrowing regions \(G=[x_L^*, x_U^*]\) are [43,78], [46,74], and [49,71] for \(n_h=75, 150\), and \(600\), respectively. 
As the size of the historical dataset increases, the borrowing region becomes narrower with fixed concurrent control samples, reflecting greater precision of the historical information and a stricter threshold for compatibility with the concurrent control data. In contrast, smaller historical sample sizes yield broader borrowing regions, permitting borrowing across a wider range of observed outcomes in the current trial.

The WOW gating step addresses several key limitations of existing adaptive borrowing methods. In the SAM prior (upper panel), the borrowing decision depends solely on the user-specified clinical threshold \(\delta\), not on the sample size of the historical data. As a result, SAM does not adapt to the precision of the historical information: it may assign substantial weight to noisy historical data or discount highly informative data, simply based on proximity in point estimates. This behavior is concerning, as larger historical datasets should warrant stricter compatibility before borrowing is allowed when the concurrent control sample size is fixed. Moreover, SAM’s borrowing weights are highly sensitive to the choice of \(\delta\), making the approach unstable and challenging to calibrate. In contrast, the WOW gating strategy adaptively adjusts the borrowing region based on both the degree of compatibility and the informativeness of the historical data. As the historical sample size \(n_h\) increases, WOW tightens the borrowing region progressively, permitting borrowing only when the historical data show strong compatibility with the concurrent control. 

The EB-rMAP prior (lower panel) offers more robustness to its tuning parameter \(\gamma\), but introduces a different issue: its borrowing behavior is asymmetric. Even when the current response count is close to the historical mean (e.g., slightly below 60 when \(\theta_h = 0.4\)), EB-rMAP often assigns weights close to zero, leading to overly cautious rejection of compatible data. This asymmetry can limit power in realistic settings. The WOW gating strategy offers a principled, data-driven, and interpretable solution that uses predictive performance to determine whether borrowing is admissible before applying a downstream mixture-prior procedure. By ensuring borrowing occurs only when supported by the data, WOW avoids instability from manual thresholding, accounts for the precision of historical information, and aligns with the intuitive principle that larger, more reliable datasets should be subject to stricter compatibility criteria. WOW provides a practical and transparent framework for adaptive borrowing that balances statistical rigor with real-world usability.

 \section{Gating Strategy for Continuous Endpoints}

\subsection{Mixture prior and posterior}

Let the concurrent control observations be
\(y_1,\dots ,y_n \stackrel{\mathrm{iid}}{\sim} N(\theta,\sigma^2)\),
with \(\sigma^2\) known. Denote $\overline y=\frac{\sum_{i=1}^{n}y_i}{n}$.  Historical data
\(D_h=\{y_{h1},\dots ,y_{h n_h}\}\) have empirical mean
\(\overline y_h\) and unbiased variance estimate \(s^2\).  Following the structure of the binary case, we assign to \(\theta\) a normal–mixture prior
\[
\pi(\theta)\;=\;w_h\,\pi_h(\theta)\;+\;(1-w_h)\,\pi_0(\theta),
\]
where
\[
\pi_h(\theta)=N\!\left(\overline y_h,\;\frac{s^2}{n_h}\right),
\qquad
\pi_0(\theta)=N\!\bigl(\overline y_h,\sigma_0^{2}\bigr),
\quad
\sigma_0^{2}\gg s^2/n_h .
\]
The informative component \(\pi_h(\theta)\) is the approximate posterior distribution of 
\(\theta\) under the Jeffreys prior
\(p(\theta,\sigma^2)\propto\sigma^{-2}\), derived solely from the historical data \(D_h\). The variance $\sigma_h^{2}=s^2/n_h$ reflects the precision of the historical empirical mean $\overline y_h$. The noninformative component $\pi_0(\theta)$ represents minimal prior information.

The posterior can be derived as the mixture: 

\begin{equation}
\label{eq:suppDmix}
p(\theta\mid D,D_h)
  =w_h^{\ast}\,p_h(\theta\mid D,D_h)
   +(1-w_h^{\ast})\,p_0(\theta\mid D). 
\end{equation}

Here 
\[
p_h(\theta\mid D,D_h)=N(\mu_h,\tau_h^{2}),
\qquad
p_0(\theta\mid D)=N(\mu_0,\tau_0^{2}),
\]
where
\[
\mu_h
  =\frac{\overline y_h/\sigma_h^{2}+n\overline y/\sigma^{2}}
         {1/\sigma_h^{2}+n/\sigma^{2}},
\quad
\mu_0
  =\frac{\overline y_h/\sigma_0^{2}+n\overline y/\sigma^{2}}
         {1/\sigma_0^{2}+n/\sigma^{2}},
\]
\[
\bigl(\tau_h\bigr)^{-2}
      =\frac{1}{\sigma_h^{2}}+\frac{n}{\sigma^{2}},
\qquad
\bigl(\tau_0\bigr)^{-2}
      =\frac{1}{\sigma_0^{2}}+\frac{n}{\sigma^{2}}.
\]

 The posterior borrowing weight is 
\[
w_h^{\ast}
   =\frac{w_h\,z_h}{w_h\,z_h+(1-w_h)\,z_0},
\]
where 
\[
z_h=\phi\!\left(\overline y;\,\overline y_h,\;\sigma_h^{2}+\sigma^{2}/n\right),
\quad
z_0=\phi\!\left(\overline y;\,\overline y_h,\;\sigma_0^{2}+\sigma^{2}/n\right),
\]
with \(\phi(\cdot;\mu,\tau^{2})\) the normal density.



\subsection{Continuous endpoints: derivation of WAIC}

For $y_i\sim N(\theta, \sigma^2)$ with known $\sigma^2$, the log-likelihood is: $\log f(y_i \mid \theta) = -\frac{1}{2}\log(2\pi\sigma^2) - \frac{(y_i - \theta)^2}{2\sigma^2}.
$  Substituting the mixture posterior (\ref{eq:suppDmix}) in (\ref{eq:WAIC_quad_long}), we get the WAIC$_C$ for the continuous endpoints in quadratic form expression: 
\begin{equation}\label{eq:WAIC_continuous}
\mathrm{WAIC}_C\bigl(w_h,D,D_h\bigr)
   = -I_{4}{w_{h}^{*}}^{2}
     + I_{5}\,w_h^{\ast}
     + I_{6},
\end{equation}
where

$$
\begin{aligned}
I_4 \;=\;& 2 \Bigl[ 
          \sum_{i=1}^{n}\Bigl\{\mathbb{E}_{p_h}\!\bigl\{f(y_i\mid\theta)\bigr\}\Bigr\}^{2}
         + \sum_{i=1}^{n}\Bigl\{\mathbb{E}_{p_0}\!\bigl\{f(y_i\mid\theta)\bigr\}\Bigr\}^{2}
           -2 \sum_{i=1}^{n}
   \mathbb{E}_{p_h}\!\bigl\{f(y_i\mid\theta)\bigr\}\,
   \mathbb{E}_{p_0}\!\bigl\{f(y_i\mid\theta)\bigr\} \Bigr];\\[6pt]
I_5 \;=\;& 2\Bigl[
             \sum_{i=1}^{n}\mathbb{E}_{p_h}\!\bigl\{f^{2}(y_i\mid\theta)\bigr\} - 
             \sum_{i=1}^{n}\mathbb{E}_{p_0}\!\bigl\{f^{2}(y_i\mid\theta)\bigr\} -
             2 \sum_{i=1}^{n}
   \mathbb{E}_{p_h}\!\bigl\{f(y_i\mid\theta)\bigr\}\,
   \mathbb{E}_{p_0}\!\bigl\{f(y_i\mid\theta)\bigr\} \\[-2pt]
         & + \sum_{i=1}^{n}\Bigl\{\mathbb{E}_{p_h}\!\bigl\{f(y_i\mid\theta)\bigr\}\Bigr\}^{2} 
          - \sum_{i=1}^{n}\mathbb{E}_{p_h}\!\bigl\{f(y_i\mid\theta)\bigr\}
          +  \sum_{i=1}^{n}\mathbb{E}_{p_0}\!\bigl\{f(y_i\mid\theta)\bigr\}
           \Bigr].\\[2pt]
\end{aligned}
$$

$I_6$ is a constant independent of $\wst$. The terms in $I_4$ and $I_5$ can be computed in the following closed-form:

\[
\begin{aligned}
\sum_{i=1}^{n}\mathbb{E}_{p_h}[f(y_i\mid\theta)]
  &= -\,n\log\sigma
    -\frac{1}{2\sigma^{2}}
       \Bigl\{n{\tau_h}^{2}+n\mu_h^{2}-2\mu_h \sum_{i=1}^{n}y_i+\sum_{i=1}^{n}y_i^2\Bigr\},\\
\sum_{i=1}^{n}\mathbb{E}_{p_0}[f(y_i\mid\theta)]
  &= -\,n\log\sigma
    -\frac{1}{2\sigma^{2}}
       \Bigl\{n{\tau_0}^{2}+n\mu_0^{2}-2\mu_0 \sum_{i=1}^{n}y_i+\sum_{i=1}^{n}y_i^2\Bigr\},\\
\sum_{i=1}^{n}\mathbb{E}_{p_h}[f^{2}(y_i\mid\theta)]
 &= n\log^{2}\sigma
    +\frac{\log\sigma}{\sigma^{2}}
       \bigl[\sum_{i=1}^{n}y_i^2-2\mu_h \sum_{i=1}^{n}y_i+n(\tau_h^{2}+\mu_h^{2})\bigr]\\
 &\quad
    +\frac{1}{4\sigma^{4}}
       \Bigl[\sum_{i=1}^{n}y_i^4-4\mu_h \sum_{i=1}^{n}y_i^3
             +6(\tau_h^{2}+\mu_h^{2})\sum_{i=1}^{n}y_i^2
             -4\mu_h(3\tau_h^{2}+\mu_h^{2})\sum_{i=1}^{n}y_i\\
 &\hspace{46mm}
             +n\bigl(3\tau_h^{4}+6\mu_h^{2}\tau_h^{2}+\mu_h^{4}\bigr)
       \Bigr],\\
\sum_{i=1}^{n}\mathbb{E}_{p_0}[f^{2}(y_i\mid\theta)]
 &= n\log^{2}\sigma
    +\frac{\log\sigma}{\sigma^{2}}
       \bigl[\sum_{i=1}^{n}y_i^2-2\mu_0 \sum_{i=1}^{n}y_i+n(\tau_0^{2}+\mu_0^{2})\bigr]\\
 &\quad
    +\frac{1}{4\sigma^{4}}
       \Bigl[\sum_{i=1}^{n}y_i^4-4\mu_0 \sum_{i=1}^{n}y_i^3
             +6(\tau_0^{2}+\mu_0^{2})\sum_{i=1}^{n}y_i^2
             -4\mu_0(3\tau_0^{2}+\mu_0^{2})\sum_{i=1}^{n}y_i\\
 &\hspace{46mm}
             +n\bigl(3\tau_0^{4}+6\mu_0^{2}\tau_0^{2}+\mu_0^{4}\bigr)
       \Bigr],\\
\end{aligned}
\]

\[
\begin{aligned}
\sum_{i=1}^{n}
   \mathbb{E}_{p_h}[f(y_i\mid\theta)]\,
   \mathbb{E}_{p_0}[f(y_i\mid\theta)]
 &= n \Bigl[\log\sigma+\{\mu_h^{2}+\tau_h^{2}\}/(2\sigma^{2})\Bigr]\Bigl[\log\sigma+\{\mu_0^{2}+\tau_0^{2}\}/(2\sigma^{2})\Bigr] \\
    &-\frac{\Bigl[\log\sigma+\{\mu_h^{2}+\tau_h^{2}\}/(2\sigma^{2})\Bigr]\mu_0}{\sigma^{2}}\,\Bigl(\sum_{i=1}^{n}y_i\Bigr) \\[4pt]
    &-\frac{\Bigl[\log\sigma+\{\mu_0^{2}+\tau_0^{2}\}/(2\sigma^{2})\Bigr]\mu_h}{\sigma^{2}}\,\Bigl(\sum_{i=1}^{n}y_i\Bigr) \\[4pt]
 &+\Bigl[
         \frac{\log\sigma+\{\mu_h^{2}+\tau_h^{2}\}/(2\sigma^{2})}{2\sigma^{2}}
      \Bigr]\Bigl(\sum_{i=1}^{n}y_i^4\Bigr)\\
&+\Bigl[\frac{\log\sigma+\{\mu_0^{2}+\tau_0^{2}\}/(2\sigma^{2})}{2\sigma^{2}}
      \Bigr]\Bigl(\sum_{i=1}^{n}y_i^4\Bigr)\\
&+\frac{\mu_h\mu_0}{\sigma^{4}}\Bigl(\sum_{i=1}^{n}y_i^4\Bigr)\\
    &-\frac{\mu_h+\mu_0}{2\sigma^{4}}\,\Bigl(\sum_{i=1}^{n}y_i^3\Bigr)
    +\frac{1}{4\sigma^{4}}\Bigl(\sum_{i=1}^{n}y_i^4\Bigr).
\end{aligned}
\]

\subsection{Continuous endpoints: WOW gating strategy}  

Extending the approach from the binary endpoint case, we propose a pre-specified borrowing rule for continuous endpoints. Specifically, the borrowing region is determined by comparing \( \text{WAIC}(0, D, D_h) \) and \( \text{WAIC}(1, D, D_h) \). Before observing the data \( D \), the sample mean \( \overline{y} \) and standard deviation \( \sigma \) are random variables that depend on the realizations of \( D \). For every possible realization of \( D \), thresholds \( \overline{y}_L(\sigma) \) and \( \overline{y}_U(\sigma) \) are computed to define the borrowing region based on the historical data \( D_h \). Borrowing occurs if the sample mean satisfies \( \overline{y} \in [\overline{y}_L(\sigma), \overline{y}_U(\sigma)] \), given the corresponding observed standard deviation \( \sigma \).  

Importantly, as in the binary case, for each fixed value of \( \sigma \), the borrowing region remains a single continuous interval as determined by WAIC-based compatibility. This ensures that borrowing decisions account for both central tendency (mean) and variability (standard deviation) in the observed data. Furthermore, similar to the binary case, the WOW gating strategy is independent of downstream borrowing methods, providing a robust, principled framework for data-driven borrowing decisions.

\section{Simulation Study Details and Additional Binary Endpoint Results}
\label{sec:supp-simulation-binary}
\yx
This section provides implementation details and additional results for the binary-endpoint simulation study in Section 5 of the main text.  

The comparator methods were SAM, EB-rMAP, Mix50, PIP, their WOW-gated versions, no borrowing (NP), and test-then-pool (TTP). The tuning parameters were $\delta=0.15$ for SAM, $\gamma=0.8$ for EB-rMAP, and $\alpha=0.05$ for PIP. TTP used Fisher's exact test of $H_0:\theta_h=\theta$ at level $0.05$, with full pooling when the null was not rejected and no borrowing otherwise. 

Performance for estimating  the concurrent control response rate \(\theta\) was summarized by relative bias, ratio MSE, absolute bias, coverage probability, and interval score. Relative bias was computed as the difference between the posterior mean under a given method and the posterior mean under NP; ratio MSE was defined as \(\mathrm{MSE}(\mathrm{method})/\mathrm{MSE}(\mathrm{NP})\); and absolute bias was defined as \(|\mathbb{E}(\theta\mid D,D_h)-\theta|\).

Interval-estimation performance for \(\theta\) was evaluated using coverage probability and interval score in two settings with \(\theta_h=0.4\) and \(n_h=600\): one with fixed historical data and one with stochastic historical data. Coverage probability was defined as the proportion of simulation replicates in which the nominal 95\% posterior credible interval contained the true value of \(\theta\). The interval score for a 95\% interval \([L,U]\) was computed as
\[
(U-L) + \frac{2}{0.05}(L-\theta)\mathbf{1}\{\theta<L\}
+ \frac{2}{0.05}(\theta-U)\mathbf{1}\{\theta>U\},
\]
where smaller values indicate better interval performance.

For the treatment effect \(\theta_t-\theta\), we evaluated bias and coverage in the same two settings with \(\theta_h=0.4\) and \(n_h=600\), considering both fixed and stochastic historical data.

Figure~\ref{fig:supp:bin:theta-absolute-bias} reports absolute bias for the concurrent control response rate $\theta$. The NP curve in  Figure~\ref{fig:supp:bin:theta-absolute-bias} is not flat because NP estimates \(\theta\) using the posterior mean under a uniform \(\mathrm{Beta}(1,1)\) prior, \(\hat{\theta}=(x+1)/(n+2)\). Since this estimator is shrunk toward the prior mean \(1/2\), its bias can be computed: $
\mathbb{E}[\hat{\theta}]-\theta=\frac{1-2\theta}{n+2}.
$
Thus, the absolute bias is proportional to \(|1/2-\theta|\), decreasing as \(\theta\) approaches \(0.5\). This pattern reflects shrinkage induced by the baseline prior and is unrelated to the historical rate \(\theta_h\), since NP performs no borrowing.
Consistent with the relative-bias and ratio-MSE results in the main text, the WOW-gated methods generally reduce absolute bias under prior--data conflict, especially when the historical sample size is large. Tables~\ref{tab:supp:bin:theta-c-coverage} and~\ref{tab:supp:bin:theta-c-interval-score} report coverage probability and interval score for $\theta$, respectively. 
Tables~\ref{tab:supp:bin:treatment-effect-bias} and~\ref{tab:supp:bin:treatment-effect-coverage} report bias and coverage for the treatment effect \(\theta_t-\theta\). 
For trial-level operating characteristics, Tables~\ref{tab:supp:bin:prespecified-power} and~\ref{tab:supp:bin:type1-error} report pre-specified power and type I error under the fixed posterior probability cutoff $C=0.95$. In all supplementary tables, boldface indicates the better-performing method within each original/WOW-gated pair for a given scenario. 
Overall, these supplementary results support the main calibrated-power findings: WOW-gated methods reduce the adverse impact of incompatible borrowing while preserving comparable performance when the historical and concurrent controls are compatible.

\begin{figure}[htbp]
\centering
\includegraphics[width=1\textwidth]{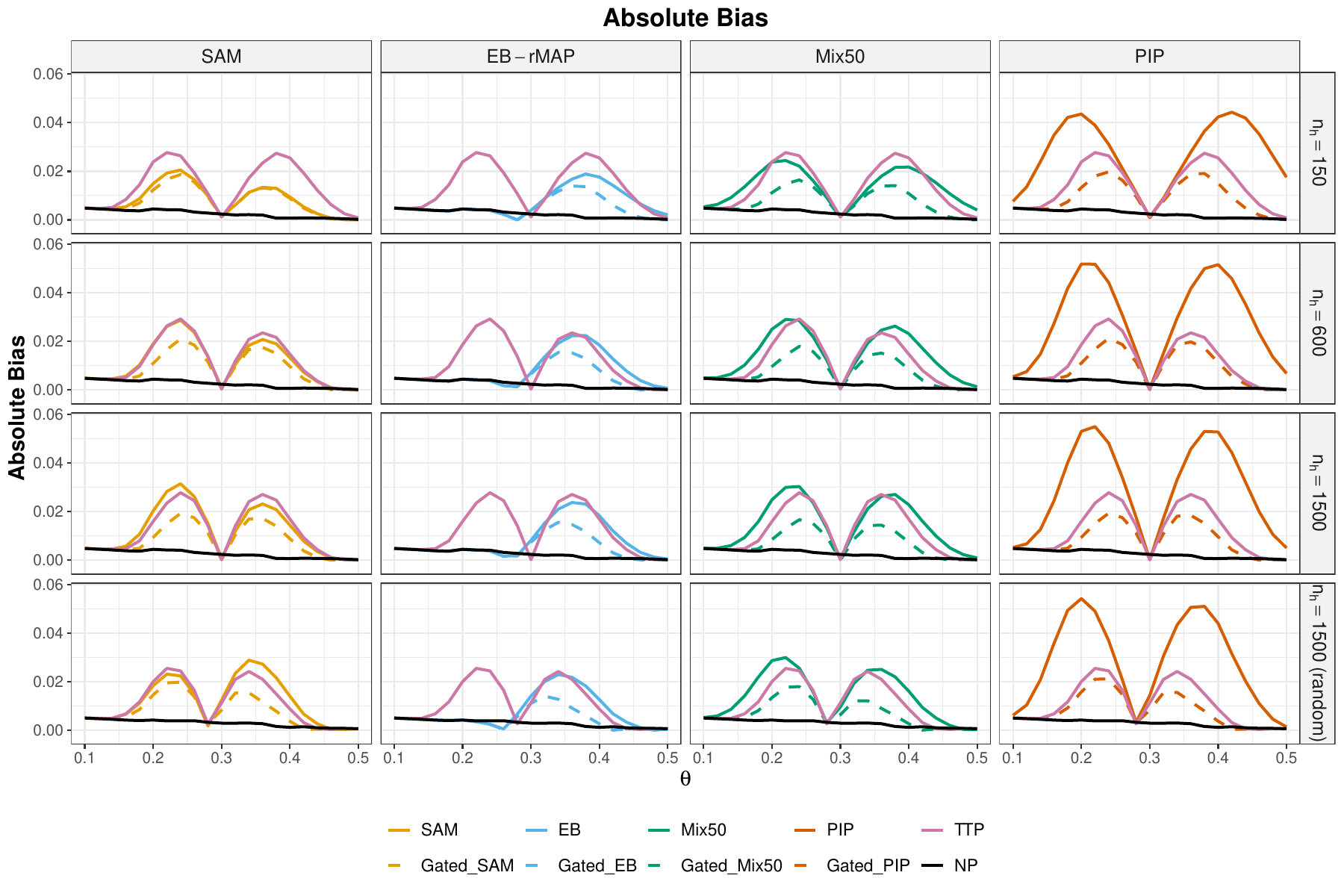}
\caption{Absolute bias in estimating the concurrent control response rate $\theta$ for the binary endpoint. Original and WOW-gated versions of SAM, EB-rMAP, Mix50, and PIP are compared across historical sample sizes, with an additional stochastic historical-data scenario. TTP is included as a benchmark.}
\label{fig:supp:bin:theta-absolute-bias}
\end{figure}

\begin{table}[htbp]
\centering
\caption{Coverage probability for estimating the concurrent control response rate $\theta$ in the binary endpoint simulations.}
\label{tab:supp:bin:theta-c-coverage}
\resizebox{\textwidth}{!}{%
\begin{tabular}{ccccccccccccc}
\toprule
Scenario & $\theta$ & $\theta_t$ & NP & SAM & Gated SAM & EB-rMAP & Gated EB-rMAP & Mix50 & Gated Mix50 & PIP & Gated PIP & TTP \\
\midrule
\multicolumn{13}{c}{\textbf{Case 1: $\theta_h=0.4$ with fixed $\bD_h$}} \\
\addlinespace[0.5ex]
1.1 & 0.30 & 0.40 & 0.949 & 0.767 & \textbf{0.768} & \textbf{0.949} & 0.947 & \textbf{0.866} & 0.861 & 0.503 & \textbf{0.704} & 0.583 \\
1.2 & 0.32 & 0.42 & 0.961 & 0.637 & 0.637 & 0.922 & 0.922 & \textbf{0.820} & 0.815 & 0.328 & \textbf{0.510} & 0.372 \\
1.3 & 0.40 & 0.50 & 0.950 & \textbf{0.985} & 0.950 & \textbf{0.969} & 0.950 & \textbf{0.993} & 0.950 & \textbf{1.000} & 0.950 & 0.979 \\
1.4 & 0.46 & 0.56 & 0.950 & 0.514 & 0.514 & \textbf{0.773} & 0.762 & \textbf{0.815} & 0.798 & 0.214 & \textbf{0.338} & 0.208 \\
1.5 & 0.54 & 0.64 & 0.950 & 0.917 & \textbf{0.918} & 0.927 & \textbf{0.929} & 0.902 & \textbf{0.919} & 0.728 & \textbf{0.921} & 0.864 \\
\specialrule{\heavyrulewidth}{1pt}{1pt}
\multicolumn{13}{c}{\textbf{Case 2: $\theta_h=0.4$ with random $\bD_h$}} \\
\addlinespace[0.5ex]
2.1 & 0.30 & 0.40 & 0.951 & 0.769 & \textbf{0.770} & \textbf{0.950} & 0.949 & \textbf{0.864} & 0.860 & 0.511 & \textbf{0.715} & 0.590 \\
2.2 & 0.32 & 0.42 & 0.962 & 0.652 & 0.652 & 0.917 & 0.917 & \textbf{0.824} & 0.818 & 0.325 & \textbf{0.524} & 0.375 \\
2.3 & 0.40 & 0.50 & 0.948 & \textbf{0.984} & 0.948 & \textbf{0.969} & 0.949 & \textbf{0.996} & 0.948 & \textbf{0.999} & 0.948 & 0.979 \\
2.4 & 0.46 & 0.56 & 0.951 & \textbf{0.511} & 0.509 & \textbf{0.770} & 0.760 & \textbf{0.813} & 0.796 & 0.216 & \textbf{0.322} & 0.209 \\
2.5 & 0.54 & 0.64 & 0.951 & 0.921 & 0.921 & 0.930 & 0.930 & 0.905 & \textbf{0.922} & 0.735 & \textbf{0.923} & 0.866 \\
\bottomrule
\end{tabular}%
}
\vspace{0.5ex}

\end{table}

\begin{table}[htbp]
\centering
\caption{Interval score for the concurrent control response rate \(\theta\) in the binary endpoint simulations.}
\label{tab:supp:bin:theta-c-interval-score}
\resizebox{\textwidth}{!}{%
\begin{tabular}{ccccccccccccc}
\toprule
Scenario & $\theta$ & $\theta_t$ & NP & SAM & Gated SAM & EB-rMAP & Gated EB-rMAP & Mix50 & Gated Mix50 & PIP & Gated PIP & TTP \\
\midrule
\multicolumn{13}{c}{\textbf{Case 1: $\theta_h=0.4$ with fixed $\bD_h$}} \\
\addlinespace[0.5ex]
1.1 & 0.30 & 0.40 & 0.171 & 0.507 & \textbf{0.503} & 0.205 & 0.205 & 0.288 & \textbf{0.271} & 1.086 & \textbf{0.741} & 0.971 \\
1.2 & 0.32 & 0.42 & 0.170 & 0.617 & \textbf{0.614} & 0.244 & 0.244 & 0.293 & \textbf{0.284} & 0.989 & \textbf{0.804} & 0.955 \\
1.3 & 0.40 & 0.50 & 0.187 & \textbf{0.098} & 0.110 & \textbf{0.127} & 0.139 & \textbf{0.104} & 0.132 & \textbf{0.072} & 0.106 & 0.098 \\
1.4 & 0.46 & 0.56 & 0.191 & 0.462 & \textbf{0.460} & 0.310 & \textbf{0.308} & 0.260 & 0.260 & 0.568 & \textbf{0.556} & 0.584 \\
1.5 & 0.54 & 0.64 & 0.191 & 0.294 & \textbf{0.291} & 0.241 & \textbf{0.231} & 0.267 & \textbf{0.247} & 0.796 & \textbf{0.372} & 0.558 \\
\specialrule{\heavyrulewidth}{1pt}{1pt}
\multicolumn{13}{c}{\textbf{Case 2: $\theta_h=0.4$ with random $\bD_h$}} \\
\addlinespace[0.5ex]
2.1 & 0.30 & 0.40 & 0.170 & 0.511 & \textbf{0.507} & 0.202 & 0.202 & 0.291 & \textbf{0.274} & 1.071 & \textbf{0.722} & 0.957 \\
2.2 & 0.32 & 0.42 & 0.172 & 0.601 & \textbf{0.597} & 0.248 & 0.248 & 0.297 & \textbf{0.288} & 0.987 & \textbf{0.788} & 0.951 \\
2.3 & 0.40 & 0.50 & 0.184 & \textbf{0.095} & 0.108 & \textbf{0.125} & 0.137 & \textbf{0.103} & 0.130 & \textbf{0.072} & 0.104 & 0.096 \\
2.4 & 0.46 & 0.56 & 0.187 & 0.459 & \textbf{0.457} & 0.311 & \textbf{0.309} & \textbf{0.258} & 0.259 & 0.569 & \textbf{0.559} & 0.584 \\
2.5 & 0.54 & 0.64 & 0.187 & 0.289 & \textbf{0.286} & 0.235 & \textbf{0.225} & 0.262 & \textbf{0.241} & 0.793 & \textbf{0.365} & 0.555 \\
\bottomrule
\end{tabular}%
}
\vspace{0.5ex}

\end{table}

\xx

\begin{table}[htbp]
\centering
\caption{Bias for the treatment effect $\theta_t-\theta$ in the binary endpoint simulations. Values are reported on the $10^{-2}$ scale.}
\label{tab:supp:bin:treatment-effect-bias}
\resizebox{\textwidth}{!}{%
\begin{tabular}{ccccccccccccc}
\toprule
Scenario & $\theta$ & $\theta_t$ & NP & SAM & Gated SAM & EB-rMAP & Gated EB-rMAP & Mix50 & Gated Mix50 & PIP & Gated PIP & TTP \\
\midrule
\multicolumn{13}{c}{\textbf{Case 1: $\theta_h=0.4$ with fixed $\bD_h$}} \\
\addlinespace[0.5ex]
1.1 & 0.30 & 0.40 & -0.28 & -1.83 & \textbf{-1.33} & -0.29 & -0.29 & -2.69 & \textbf{-1.17} & -5.51 & \textbf{-1.44} & -2.23 \\
1.2 & 0.32 & 0.42 & -0.36 & -2.49 & \textbf{-1.95} & -0.39 & -0.39 & -3.07 & \textbf{-1.71} & -5.34 & \textbf{-2.07} & -2.97 \\
1.3 & 0.40 & 0.50 & -0.36 & -0.20 & \textbf{-0.19} & 0.44 & \textbf{0.39} & -0.20 & -0.20 & \textbf{-0.15} & -0.17 & -0.18 \\
1.4 & 0.46 & 0.56 & -0.09 & 2.13 & \textbf{1.67} & 1.90 & \textbf{1.31} & 2.58 & \textbf{1.44} & 4.28 & \textbf{1.83} & 2.68 \\
1.5 & 0.54 & 0.64 & 0.12 & 0.57 & \textbf{0.34} & 0.77 & \textbf{0.28} & 1.53 & \textbf{0.31} & 4.24 & \textbf{0.37} & 0.74 \\
\specialrule{\heavyrulewidth}{1pt}{1pt}
\multicolumn{13}{c}{\textbf{Case 2: $\theta_h=0.4$ with random $\bD_h$}} \\
\addlinespace[0.5ex]
2.1 & 0.30 & 0.40 & -0.27 & -1.78 & \textbf{-1.27} & -0.28 & -0.28 & -2.67 & \textbf{-1.12} & -5.49 & \textbf{-1.37} & -2.18 \\
2.2 & 0.32 & 0.42 & -0.25 & -2.36 & \textbf{-1.78} & -0.28 & -0.28 & -2.95 & \textbf{-1.55} & -5.24 & \textbf{-1.91} & -2.87 \\
2.3 & 0.40 & 0.50 & -0.20 & -0.08 & \textbf{-0.07} & 0.57 & \textbf{0.52} & -0.08 & -0.08 & \textbf{-0.05} & -0.06 & -0.06 \\
2.4 & 0.46 & 0.56 & -0.11 & 2.14 & \textbf{1.72} & 1.89 & \textbf{1.34} & 2.56 & \textbf{1.49} & 4.25 & \textbf{1.90} & 2.65 \\
2.5 & 0.54 & 0.64 & 0.05 & 0.50 & \textbf{0.26} & 0.69 & \textbf{0.20} & 1.44 & \textbf{0.24} & 4.15 & \textbf{0.30} & 0.67 \\
\bottomrule
\end{tabular}%
}
\vspace{0.5ex}

\end{table}

\begin{table}[htbp]
\centering
\caption{Coverage for the treatment effect \(\theta_t-\theta\) in the binary endpoint simulations.}
\label{tab:supp:bin:treatment-effect-coverage}
\resizebox{\textwidth}{!}{%
\begin{tabular}{ccccccccccccc}
\toprule
Scenario & $\theta$ & $\theta_t$ & NP & SAM & Gated SAM & EB-rMAP & Gated EB-rMAP & Mix50 & Gated Mix50 & PIP & Gated PIP & TTP \\
\midrule
\multicolumn{13}{c}{\textbf{Case 1: $\theta_h=0.4$ with fixed $\bD_h$}} \\
\addlinespace[0.5ex]
1.1 & 0.30 & 0.40 & 0.950 & 0.815 & 0.815 & \textbf{0.945} & 0.944 & \textbf{0.879} & 0.876 & 0.606 & \textbf{0.760} & 0.666 \\
1.2 & 0.32 & 0.42 & 0.950 & \textbf{0.767} & 0.765 & \textbf{0.927} & 0.926 & \textbf{0.868} & 0.854 & 0.629 & \textbf{0.707} & 0.641 \\
1.3 & 0.40 & 0.50 & 0.953 & \textbf{0.959} & 0.953 & \textbf{0.953} & 0.947 & \textbf{0.971} & 0.957 & \textbf{0.970} & 0.952 & 0.960 \\
1.4 & 0.46 & 0.56 & 0.957 & \textbf{0.823} & 0.818 & \textbf{0.893} & 0.881 & \textbf{0.903} & 0.886 & 0.761 & \textbf{0.781} & 0.754 \\
1.5 & 0.54 & 0.64 & 0.955 & 0.924 & \textbf{0.931} & 0.933 & \textbf{0.942} & 0.915 & \textbf{0.935} & 0.752 & \textbf{0.918} & 0.871 \\
\specialrule{\heavyrulewidth}{1pt}{1pt}
\multicolumn{13}{c}{\textbf{Case 2: $\theta_h=0.4$ with random $\bD_h$}} \\
\addlinespace[0.5ex]
2.1 & 0.30 & 0.40 & 0.952 & 0.816 & \textbf{0.818} & 0.945 & 0.945 & \textbf{0.880} & 0.879 & 0.613 & \textbf{0.767} & 0.674 \\
2.2 & 0.32 & 0.42 & 0.954 & \textbf{0.774} & 0.773 & 0.930 & 0.930 & \textbf{0.875} & 0.864 & 0.634 & \textbf{0.714} & 0.647 \\
2.3 & 0.40 & 0.50 & 0.955 & \textbf{0.966} & 0.956 & \textbf{0.957} & 0.951 & \textbf{0.975} & 0.960 & \textbf{0.974} & 0.955 & 0.962 \\
2.4 & 0.46 & 0.56 & 0.955 & \textbf{0.819} & 0.817 & \textbf{0.888} & 0.879 & \textbf{0.901} & 0.884 & 0.763 & \textbf{0.774} & 0.752 \\
2.5 & 0.54 & 0.64 & 0.956 & 0.922 & \textbf{0.929} & 0.935 & \textbf{0.941} & 0.916 & \textbf{0.934} & 0.753 & \textbf{0.918} & 0.868 \\
\bottomrule
\end{tabular}%
}
\vspace{0.5ex}

\end{table}

\begin{table}[htbp]
\centering
\caption{Power for binary endpoints under the pre-specified posterior probability cutoff $C=0.95$.}
\label{tab:supp:bin:prespecified-power}
\resizebox{\textwidth}{!}{%
\begin{tabular}{ccccccccccccc}
\toprule
Scenario & $\theta$ & $\theta_t$ & NP & SAM & Gated SAM & EB-rMAP & Gated EB-rMAP & Mix50 & Gated Mix50 & PIP & Gated PIP & TTP \\
\midrule
\multicolumn{13}{c}{\textbf{Case 1: $\theta_h=0.3$ with fixed $\bD_h$}} \\
\addlinespace[0.5ex]
1.1 & 0.16 & 0.26 & 0.777 & 0.764 & \textbf{0.778} & 0.777 & 0.777 & 0.632 & \textbf{0.776} & 0.433 & \textbf{0.777} & 0.777 \\
1.2 & 0.18 & 0.28 & 0.750 & 0.691 & \textbf{0.751} & 0.752 & 0.752 & 0.514 & \textbf{0.750} & 0.298 & \textbf{0.749} & 0.731 \\
1.3 & 0.20 & 0.30 & 0.725 & 0.580 & \textbf{0.693} & 0.724 & \textbf{0.726} & 0.439 & \textbf{0.697} & 0.252 & \textbf{0.694} & 0.641 \\
1.4 & 0.22 & 0.32 & 0.715 & 0.523 & \textbf{0.644} & 0.718 & 0.718 & 0.449 & \textbf{0.649} & 0.338 & \textbf{0.643} & 0.560 \\
1.5 & 0.30 & 0.40 & 0.672 & 0.915 & 0.915 & 0.884 & 0.884 & 0.882 & 0.882 & \textbf{0.943} & 0.928 & 0.935 \\
1.6 & 0.34 & 0.44 & 0.652 & \textbf{0.837} & 0.827 & \textbf{0.830} & 0.809 & \textbf{0.843} & 0.814 & \textbf{0.965} & 0.854 & 0.887 \\
1.7 & 0.44 & 0.54 & 0.646 & 0.647 & 0.647 & \textbf{0.651} & 0.647 & \textbf{0.653} & 0.647 & \textbf{0.684} & 0.646 & 0.648 \\
\specialrule{\heavyrulewidth}{1pt}{1pt}
\multicolumn{13}{c}{\textbf{Case 2: $\theta_h=0.4$ with fixed $\bD_h$}} \\
\addlinespace[0.5ex]
2.1 & 0.24 & 0.34 & 0.701 & 0.697 & \textbf{0.700} & \textbf{0.702} & 0.701 & 0.603 & \textbf{0.700} & 0.448 & \textbf{0.701} & 0.700 \\
2.2 & 0.26 & 0.36 & 0.686 & 0.670 & \textbf{0.685} & 0.686 & \textbf{0.687} & 0.503 & \textbf{0.686} & 0.325 & \textbf{0.685} & 0.680 \\
2.3 & 0.30 & 0.40 & 0.673 & 0.539 & \textbf{0.628} & 0.672 & 0.672 & 0.375 & \textbf{0.629} & 0.207 & \textbf{0.624} & 0.575 \\
2.4 & 0.32 & 0.42 & 0.654 & 0.484 & \textbf{0.566} & \textbf{0.661} & 0.658 & 0.401 & \textbf{0.572} & 0.285 & \textbf{0.563} & 0.487 \\
2.5 & 0.40 & 0.50 & 0.628 & 0.893 & 0.893 & 0.850 & 0.850 & \textbf{0.858} & 0.857 & \textbf{0.919} & 0.906 & 0.916 \\
2.6 & 0.46 & 0.56 & 0.649 & \textbf{0.753} & 0.725 & \textbf{0.743} & 0.714 & \textbf{0.786} & 0.722 & \textbf{0.935} & 0.732 & 0.806 \\
2.7 & 0.54 & 0.64 & 0.673 & \textbf{0.675} & 0.674 & 0.677 & 0.677 & \textbf{0.692} & 0.673 & \textbf{0.723} & 0.674 & 0.675 \\
\specialrule{\heavyrulewidth}{1pt}{1pt}
\multicolumn{13}{c}{\textbf{Case 3: $\theta_h=0.4$ with random $\bD_h$}} \\
\addlinespace[0.5ex]
3.1 & 0.24 & 0.34 & 0.702 & 0.697 & \textbf{0.701} & \textbf{0.703} & 0.702 & 0.604 & \textbf{0.702} & 0.455 & \textbf{0.701} & 0.701 \\
3.2 & 0.26 & 0.36 & 0.689 & 0.671 & \textbf{0.688} & 0.688 & 0.688 & 0.510 & \textbf{0.687} & 0.321 & \textbf{0.688} & 0.684 \\
3.3 & 0.30 & 0.40 & 0.666 & 0.544 & \textbf{0.629} & 0.667 & 0.667 & 0.376 & \textbf{0.630} & 0.214 & \textbf{0.627} & 0.574 \\
3.4 & 0.32 & 0.42 & 0.658 & 0.496 & \textbf{0.582} & \textbf{0.664} & 0.662 & 0.412 & \textbf{0.587} & 0.295 & \textbf{0.579} & 0.496 \\
3.5 & 0.40 & 0.50 & 0.641 & \textbf{0.900} & 0.899 & 0.855 & 0.855 & \textbf{0.862} & 0.861 & \textbf{0.930} & 0.913 & 0.924 \\
3.6 & 0.46 & 0.56 & 0.645 & \textbf{0.753} & 0.730 & \textbf{0.740} & 0.714 & \textbf{0.790} & 0.726 & \textbf{0.937} & 0.737 & 0.803 \\
3.7 & 0.54 & 0.64 & 0.663 & 0.663 & \textbf{0.664} & \textbf{0.667} & 0.664 & \textbf{0.679} & 0.662 & \textbf{0.713} & 0.663 & 0.663 \\
\bottomrule
\end{tabular}%
}
\vspace{0.5ex}

\end{table}

\begin{table}[htbp]
\centering
\caption{Type I error for binary endpoints under the pre-specified posterior probability cutoff $C=0.95$.}
\label{tab:supp:bin:type1-error}
\resizebox{\textwidth}{!}{%
\begin{tabular}{ccccccccccccc}
\toprule
Scenario & $\theta$ & $\theta_t$ & NP & SAM & Gated SAM & EB-rMAP & Gated EB-rMAP & Mix50 & Gated Mix50 & PIP & Gated PIP & TTP \\
\midrule
\multicolumn{13}{c}{\textbf{Case 1: $\theta_h=0.3$ with fixed $\bD_h$}} \\
\addlinespace[0.5ex]
1.1 & 0.16 & 0.16 & 0.043 & 0.043 & \textbf{0.043} & 0.043 & \textbf{0.043} & \textbf{0.042} & 0.043 & \textbf{0.042} & 0.043 & 0.043 \\
1.2 & 0.18 & 0.18 & 0.045 & \textbf{0.044} & 0.045 & 0.046 & 0.045 & \textbf{0.041} & 0.045 & \textbf{0.034} & 0.044 & 0.046 \\
1.3 & 0.20 & 0.20 & 0.048 & 0.048 & 0.048 & 0.048 & 0.048 & \textbf{0.039} & 0.048 & \textbf{0.025} & 0.048 & 0.048 \\
1.4 & 0.22 & 0.22 & 0.053 & 0.051 & 0.052 & 0.054 & \textbf{0.052} & \textbf{0.025} & 0.053 & \textbf{0.010} & 0.052 & 0.052 \\
1.5 & 0.30 & 0.30 & 0.050 & \textbf{0.034} & 0.045 & 0.059 & 0.059 & \textbf{0.028} & 0.041 & \textbf{0.032} & 0.046 & 0.037 \\
1.6 & 0.34 & 0.34 & 0.050 & \textbf{0.182} & 0.183 & 0.149 & 0.150 & 0.137 & 0.136 & 0.218 & \textbf{0.211} & 0.212 \\
1.7 & 0.44 & 0.44 & 0.049 & 0.079 & \textbf{0.071} & 0.082 & \textbf{0.067} & 0.094 & \textbf{0.070} & 0.243 & \textbf{0.075} & 0.095 \\
\specialrule{\heavyrulewidth}{1pt}{1pt}
\multicolumn{13}{c}{\textbf{Case 2: $\theta_h=0.4$ with fixed $\bD_h$}} \\
\addlinespace[0.5ex]
2.1 & 0.24 & 0.24 & 0.055 & \textbf{0.054} & 0.055 & \textbf{0.053} & 0.056 & \textbf{0.053} & 0.054 & \textbf{0.051} & 0.054 & 0.054 \\
2.2 & 0.26 & 0.26 & 0.059 & 0.058 & 0.058 & \textbf{0.058} & 0.059 & \textbf{0.057} & 0.058 & \textbf{0.051} & 0.058 & 0.058 \\
2.3 & 0.30 & 0.30 & 0.049 & \textbf{0.049} & 0.050 & \textbf{0.049} & 0.050 & \textbf{0.034} & 0.050 & \textbf{0.015} & 0.049 & 0.050 \\
2.4 & 0.32 & 0.32 & 0.046 & \textbf{0.044} & 0.046 & \textbf{0.046} & 0.048 & \textbf{0.020} & 0.048 & \textbf{0.006} & 0.047 & 0.046 \\
2.5 & 0.40 & 0.40 & 0.038 & \textbf{0.033} & 0.043 & 0.052 & 0.051 & \textbf{0.027} & 0.038 & \textbf{0.028} & 0.043 & 0.034 \\
2.6 & 0.46 & 0.46 & 0.049 & \textbf{0.243} & 0.244 & 0.172 & 0.172 & 0.174 & 0.174 & 0.358 & \textbf{0.288} & 0.335 \\
2.7 & 0.54 & 0.54 & 0.056 & 0.086 & \textbf{0.077} & 0.081 & \textbf{0.069} & 0.105 & \textbf{0.073} & 0.295 & \textbf{0.083} & 0.122 \\
\specialrule{\heavyrulewidth}{1pt}{1pt}
\multicolumn{13}{c}{\textbf{Case 3: $\theta_h=0.4$ with random $\bD_h$}} \\
\addlinespace[0.5ex]
3.1 & 0.24 & 0.24 & 0.047 & \textbf{0.046} & 0.047 & \textbf{0.046} & 0.047 & \textbf{0.046} & 0.047 & \textbf{0.043} & 0.047 & 0.046 \\
3.2 & 0.26 & 0.26 & 0.048 & \textbf{0.047} & 0.048 & \textbf{0.047} & 0.048 & 0.047 & 0.047 & \textbf{0.041} & 0.047 & 0.047 \\
3.3 & 0.30 & 0.30 & 0.049 & 0.049 & 0.049 & \textbf{0.049} & 0.050 & \textbf{0.035} & 0.049 & \textbf{0.016} & 0.049 & 0.049 \\
3.4 & 0.32 & 0.32 & 0.048 & \textbf{0.045} & 0.048 & 0.048 & 0.048 & \textbf{0.020} & 0.048 & \textbf{0.006} & 0.048 & 0.047 \\
3.5 & 0.40 & 0.40 & 0.041 & \textbf{0.033} & 0.041 & 0.052 & \textbf{0.051} & \textbf{0.025} & 0.036 & \textbf{0.026} & 0.041 & 0.033 \\
3.6 & 0.46 & 0.46 & 0.048 & 0.247 & \textbf{0.246} & 0.169 & \textbf{0.168} & 0.173 & 0.173 & 0.360 & \textbf{0.297} & 0.338 \\
3.7 & 0.54 & 0.54 & 0.048 & 0.082 & \textbf{0.073} & 0.075 & \textbf{0.064} & 0.100 & \textbf{0.069} & 0.289 & \textbf{0.076} & 0.119 \\
\bottomrule
\end{tabular}%
}
\vspace{0.5ex}

\end{table}

\section{Simulation of Continuous Endpoint}

\yx

This section provides implementation details and additional results for the continuous-endpoint simulation study. We first describe the continuous-outcome data-generating model and design choices, then present results in parallel with the binary-endpoint simulations: point-estimation performance for the concurrent control parameter, interval-estimation performance, treatment-effect bias and coverage, calibrated power, and power and type I error under a pre-specified posterior probability cutoff.
\xx 
\subsection{Simulation setup}
For the continuous endpoint simulations, historical control data were generated from \(D_h \sim \mathcal N(\theta_h,\sigma^2)\) with \(\theta_h=0\), and concurrent control data were generated from \(D \sim \mathcal N(\theta,\sigma^2)\). Varying \(\theta\) away from \(\theta_h\) induced different levels of prior--data conflict. Across all continuous simulations, the standard deviation was fixed at \(\sigma=3\), and the non-informative prior was \(\pi_0(\theta)=\mathcal N(0,10^2)\). \yx We considered both fixed and stochastic historical-data settings. In the fixed setting, the historical sufficient statistics were fixed deterministically, yielding the informative prior \(\pi_h(\theta)=\mathcal N(0,\sigma^2/n_h)\). In the stochastic setting, a historical dataset was generated in each replicate, and the informative prior was constructed from the resulting historical sample mean and variance. \xx 

The same borrowing methods and benchmarks as in the binary endpoint simulations were evaluated: SAM, EB-rMAP, Mix50, PIP, their WOW-gated counterparts, no borrowing (NP), and test-then-pool (TTP). The tuning parameters were \(\delta=0.15\) for SAM, \(\gamma=0.8\) for EB-rMAP, and \(\alpha=0.05\) for PIP. For TTP, we used a two-sided \(z\)-test at level \(0.05\) to test \(H_0:\theta_h=\theta\); if the null hypothesis was not rejected, the historical and concurrent controls were pooled, and otherwise no borrowing was used.

For point-estimation performance of the concurrent control mean \(\theta\), we fixed the concurrent control sample size at \(n=80\) and the treatment-arm sample size at \(n_t=160\). We considered historical sample sizes \(n_h=100,900,\) and \(3600\), which correspond to historical standard errors \(\sigma/\sqrt{n_h}=0.30,0.10,\) and \(0.05\), respectively. We also included an additional stochastic historical-data setting with \(n_h=3600\). The concurrent control mean \(\theta\) was varied from \(-2\) to \(2\), with 2000 simulation replicates at each setting. Methods were evaluated using the same point-estimation metrics as in the binary endpoint simulations: relative bias against NP, ratio MSE against NP, and absolute bias. Interval-estimation performance for \(\theta\) was evaluated using coverage probability and interval score in two settings with \(\theta_h=0\) and \(n_h=900\): one with fixed historical data and one with stochastic historical data. Similarly, treatment-effect bias and coverage were evaluated in two settings with \(d=0.40\), \(\theta_h=0\), and \(n_h=900\), again considering both fixed and stochastic historical data.

\yx
For trial-level operating characteristics, we considered two fixed historical-data settings with standardized effect sizes \(d=0.40\) and \(d=0.23\), and one additional stochastic historical-data setting with \(d=0.40\), all with \(\theta_h=0\) and \(n_h=900\). Concurrent control outcomes were generated from \(\mathcal N(\theta,\sigma^2)\), with \(\theta\) varied around \(0\) to induce different levels of prior--data conflict. Treatment outcomes were generated from \(\mathcal N(\theta_t,\sigma^2)\), where \(d=(\theta_t-\theta)/\sigma\). For \(d=0.23\), we used 2:1 randomization with \(n_t=300\) treatment and \(n=150\) control subjects. For \(d=0.40\), we used smaller sample sizes with \(n_t=160\) and \(n=80\). For calibrated power, we used the same posterior-probability decision rule as in the binary endpoint simulations, with method-specific cutoffs chosen to control type I error at 5\% under \(\theta_t=\theta\). Power and type I error were also evaluated under the pre-specified posterior probability cutoff \(C=0.95\).
\xx

\subsection{Simulation results}
\yx
Figure~\ref{fig:supp:cont:theta-relative-bias} reports relative bias in estimating the concurrent control mean \(\theta\) over \(\theta\in[-2,2]\) with \(\theta_h=0\). Each panel compares the original and WOW-gated versions of SAM, EB-rMAP, Mix50, and PIP, with TTP included as a benchmark, under three fixed historical sample sizes and one stochastic historical-data scenario at \(n_h=3600\). As in the binary endpoint simulations, non-gated methods show increasing bias as \(\theta\) departs from \(\theta_h\), especially under moderate prior--data conflict. PIP behaves similarly near \(\theta=\theta_h\) but often has larger bias away from \(\theta_h\), reflecting its sensitivity to discordant data. EB-rMAP is nearly symmetric because its PPP-based weight is direction invariant for continuous outcomes. Across methods, WOW gating consistently reduces bias and generally performs better than TTP. This benefit becomes more pronounced as \(n_h\) increases. At \(n_h=3600\), non-gated methods show bias exceeding \(0.20\) in magnitude, whereas the WOW-gated versions remain below \(0.12\). The stochastic historical-data scenario in the final row follows the fixed \(n_h=3600\) pattern, with slightly noisier curves because the historical data are regenerated in each replicate.
\xx

\begin{figure}[htbp]
\centering
\includegraphics[width=1\textwidth]{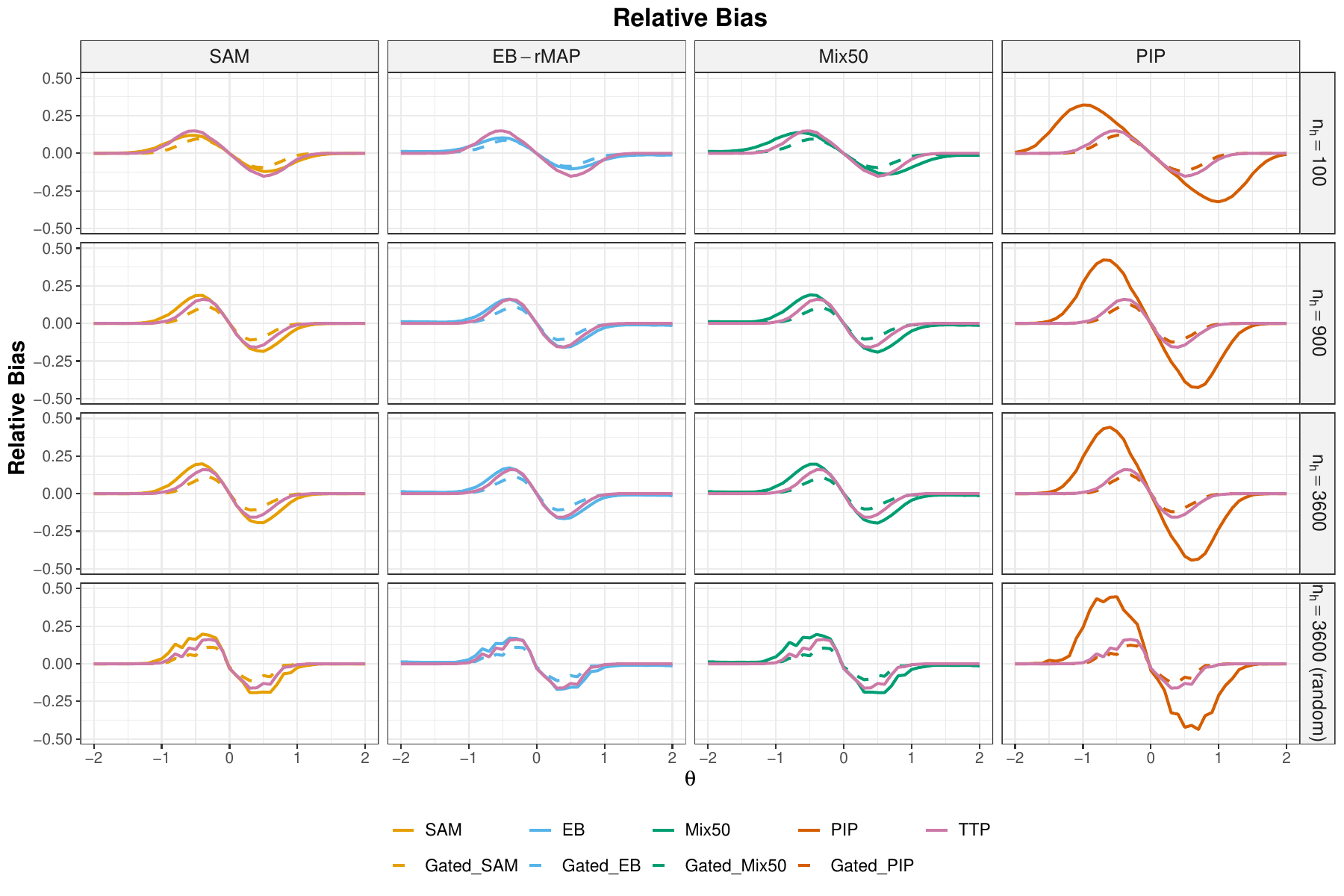}
\caption{
Relative bias in estimating the concurrent control $\theta$ for the continuous endpoint. Original and WOW-gated versions of SAM, EB-rMAP, Mix50, and PIP are compared across historical sample sizes, with an additional stochastic historical-data scenario. TTP is included as a benchmark.}
\label{fig:supp:cont:theta-relative-bias}
\end{figure}

Figure~\ref{fig:supp:cont:theta-ratio-mse} presents the corresponding ratio MSE results under the same settings. The patterns closely mirror those for relative bias: WOW-gated methods reduce MSE in regions with substantial prior--data conflict while maintaining comparable performance to the original methods when \(\theta\) is close to \(\theta_h\). The ratio MSE curves have similar shapes across methods, reflecting their shared mixture-prior borrowing structure. The stochastic historical-data scenario shows the same overall pattern.

\begin{figure}[htbp]
\centering
\includegraphics[width=1\textwidth]{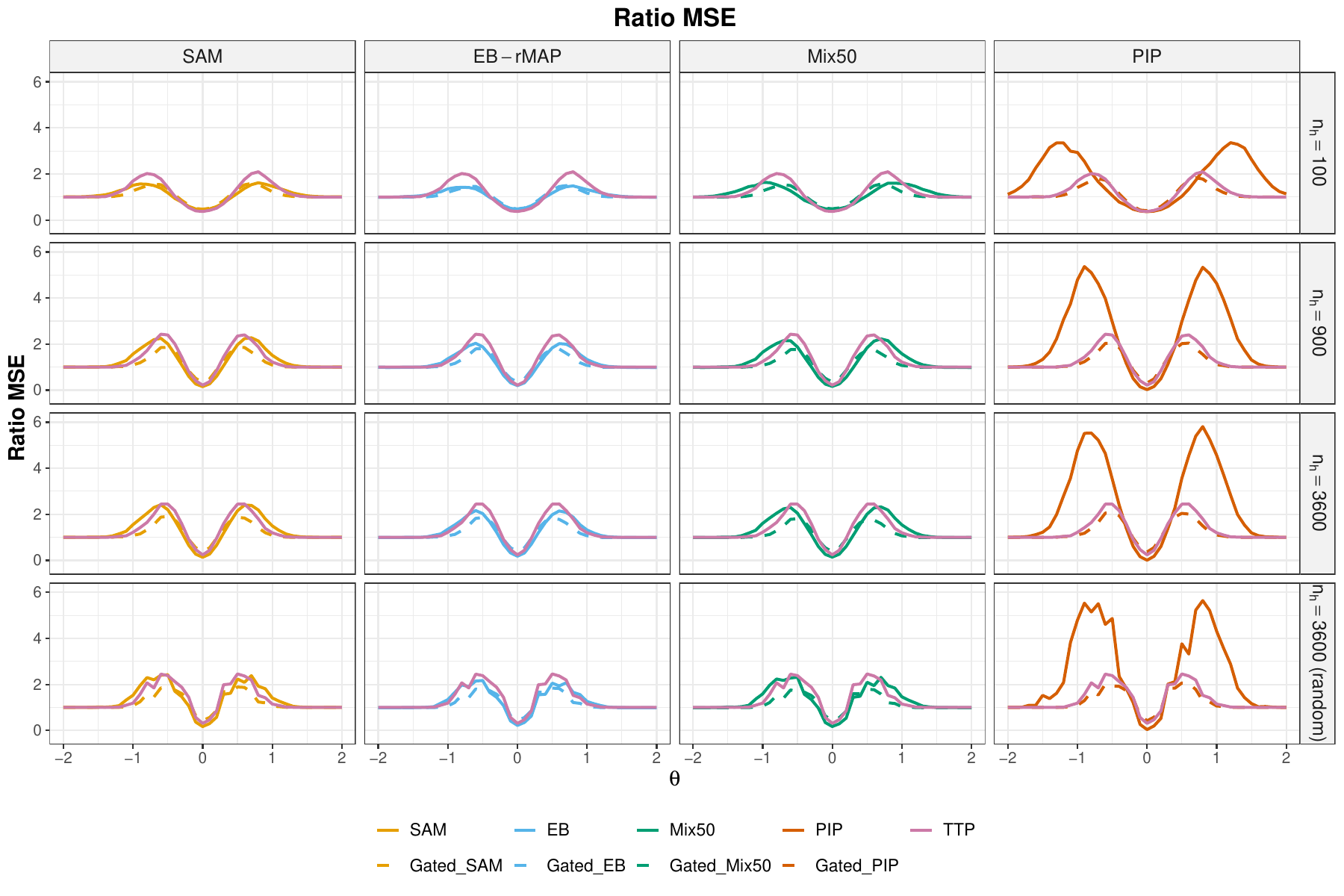}

\caption{Ratio of MSE in estimating the concurrent control response rate $\theta$ for the continuous endpoint. Original and WOW-gated versions of SAM, EB-rMAP, Mix50, and PIP are compared across historical sample sizes, with an additional stochastic historical-data scenario. TTP is included as a benchmark.}
\label{fig:supp:cont:theta-ratio-mse}
\end{figure}

\yx The absolute bias results are reported in  Figure~\ref{fig:supp:cont:theta-absolute-bias}. Consistent with the relative-bias results, the WOW-gated methods generally exhibit smaller absolute bias than their non-gated counterparts under prior--data conflict, with the advantage becoming more pronounced as the historical sample size increases. Tables~\ref{tab:supp:cont:theta-c-coverage} and~\ref{tab:supp:cont:theta-c-interval-score} report coverage probability and interval score for the concurrent control parameter \(\theta\), respectively. Tables~\ref{tab:supp:cont:treatment-effect-bias} and~\ref{tab:supp:cont:treatment-effect-coverage} report bias and coverage for the treatment effect \(\theta_t-\theta\), respectively.

\begin{figure}[htbp]
\centering
\includegraphics[width=1\textwidth]{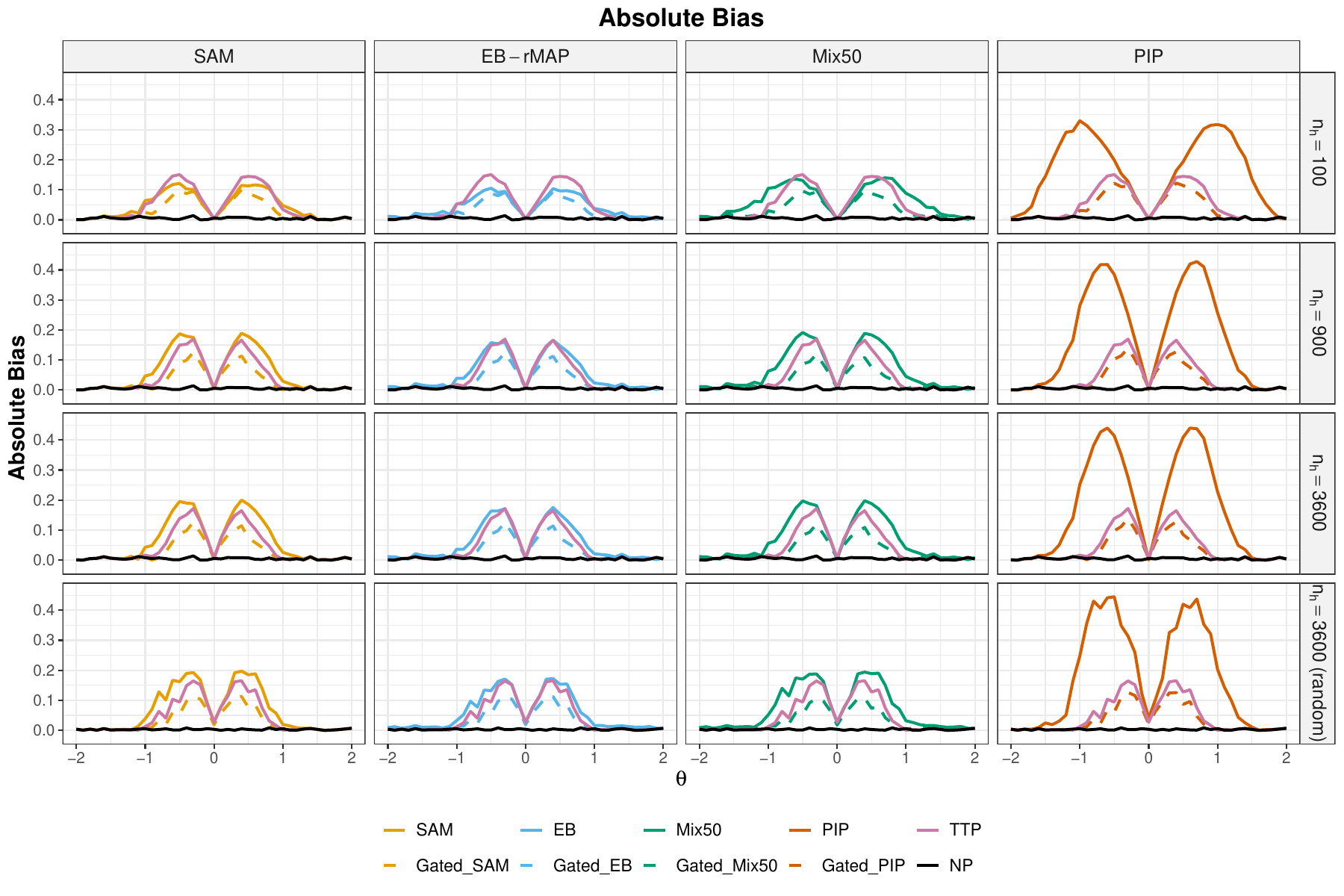}

\caption{Absolute bias in estimating the concurrent control response rate $\theta$ for the continuous endpoint. Original and WOW-gated versions of SAM, EB-rMAP, Mix50, and PIP are compared across historical sample sizes, with an additional stochastic historical-data scenario. TTP is included as a benchmark.}

\label{fig:supp:cont:theta-absolute-bias}
\end{figure}

\begin{table}[htbp]
\centering
\caption{Coverage probability for estimating the concurrent control response rate $\theta$ in the continuous endpoint simulations.}
\label{tab:supp:cont:theta-c-coverage}
\resizebox{\textwidth}{!}{%
\begin{tabular}{ccccccccccccc}
\toprule
Scenario & $\theta$ & $\theta_t$ & NP & SAM & Gated SAM & EB-rMAP & Gated EB-rMAP & Mix50 & Gated Mix50 & PIP & Gated PIP & TTP \\
\midrule
\multicolumn{13}{c}{\textbf{Case 1: $d=0.4$ with fixed $\bD_h$}} \\
\addlinespace[0.5ex]
1.1 & -1.0 & 0.2 & 0.953 & 0.872 & \textbf{0.874} & \textbf{0.898} & 0.897 & 0.891 & \textbf{0.894} & 0.552 & \textbf{0.860} & 0.806 \\
1.2 & -0.8 & 0.4 & 0.954 & 0.807 & 0.807 & 0.838 & 0.838 & \textbf{0.855} & 0.853 & 0.397 & \textbf{0.706} & 0.631 \\
1.3 & 0.0 & 1.2 & 0.953 & \textbf{0.999} & 0.953 & \textbf{0.993} & 0.955 & \textbf{0.999} & 0.955 & \textbf{1.000} & 0.953 & 0.960 \\
1.4 & 0.8 & 2.0 & 0.939 & \textbf{0.777} & 0.774 & \textbf{0.808} & 0.807 & 0.826 & 0.826 & 0.373 & \textbf{0.672} & 0.583 \\
1.5 & 1.0 & 2.2 & 0.953 & 0.881 & \textbf{0.883} & 0.900 & 0.900 & 0.901 & 0.901 & 0.565 & \textbf{0.872} & 0.819 \\
\specialrule{\heavyrulewidth}{1pt}{1pt}
\multicolumn{13}{c}{\textbf{Case 2: $d=0.4$ with random $\bD_h$}} \\
\addlinespace[0.5ex]
2.1 & -1.0 & 0.2 & 0.950 & 0.852 & \textbf{0.853} & 0.884 & \textbf{0.885} & 0.878 & 0.878 & 0.519 & \textbf{0.835} & 0.771 \\
2.2 & -0.8 & 0.4 & 0.949 & \textbf{0.798} & 0.797 & 0.827 & 0.827 & \textbf{0.845} & 0.844 & 0.390 & \textbf{0.700} & 0.621 \\
2.3 & 0.0 & 1.2 & 0.944 & \textbf{0.974} & 0.945 & \textbf{0.969} & 0.946 & \textbf{0.975} & 0.947 & \textbf{0.980} & 0.951 & 0.952 \\
2.4 & 0.8 & 2.0 & 0.942 & \textbf{0.762} & 0.761 & 0.789 & \textbf{0.790} & \textbf{0.824} & 0.822 & 0.328 & \textbf{0.624} & 0.532 \\
2.5 & 1.0 & 2.2 & 0.948 & 0.879 & \textbf{0.883} & 0.897 & \textbf{0.898} & 0.896 & 0.896 & 0.595 & \textbf{0.881} & 0.829 \\
\bottomrule
\end{tabular}%
}
\vspace{0.5ex}

\end{table}

\begin{table}[htbp]
\centering
\caption{Interval score for the concurrent control mean \(\theta\) in the continuous endpoint simulations.}

\label{tab:supp:cont:theta-c-interval-score}
\resizebox{\textwidth}{!}{%
\begin{tabular}{ccccccccccccc}
\toprule
Scenario & $\theta$ & $\theta_t$ & NP & SAM & Gated SAM & EB-rMAP & Gated EB-rMAP & Mix50 & Gated Mix50 & PIP & Gated PIP & TTP \\
\midrule
\multicolumn{13}{c}{\textbf{Case 1: $d=0.4$ with fixed $\bD_h$}} \\
\addlinespace[0.5ex]
1.1 & -1.0 & 0.2 & 1.535 & 2.836 & \textbf{2.701} & 2.772 & \textbf{2.679} & 2.344 & \textbf{2.193} & 11.774 & \textbf{5.074} & 6.637 \\
1.2 & -0.8 & 0.4 & 1.559 & 3.469 & \textbf{3.357} & 3.765 & \textbf{3.678} & 2.553 & \textbf{2.429} & 13.215 & \textbf{7.409} & 9.040 \\
1.3 & 0.0 & 1.2 & 1.542 & \textbf{0.608} & 0.834 & \textbf{0.580} & 0.780 & \textbf{0.769} & 0.983 & \textbf{0.382} & 0.648 & 0.616 \\
1.4 & 0.8 & 2.0 & 1.633 & 3.888 & \textbf{3.787} & 4.250 & \textbf{4.174} & 2.835 & \textbf{2.714} & 13.804 & \textbf{8.109} & 10.016 \\
1.5 & 1.0 & 2.2 & 1.568 & 2.679 & \textbf{2.537} & 2.551 & \textbf{2.462} & 2.261 & \textbf{2.103} & 11.404 & \textbf{4.540} & 6.100 \\
\specialrule{\heavyrulewidth}{1pt}{1pt}
\multicolumn{13}{c}{\textbf{Case 2: $d=0.4$ with random $\bD_h$}} \\
\addlinespace[0.5ex]
2.1 & -1.0 & 0.2 & 1.562 & 3.006 & \textbf{2.874} & 2.935 & \textbf{2.846} & 2.425 & \textbf{2.275} & 12.310 & \textbf{5.475} & 7.216 \\
2.2 & -0.8 & 0.4 & 1.610 & 3.617 & \textbf{3.507} & 3.912 & \textbf{3.828} & 2.661 & \textbf{2.540} & 13.409 & \textbf{7.564} & 9.263 \\
2.3 & 0.0 & 1.2 & 1.609 & \textbf{0.703} & 0.923 & \textbf{0.697} & 0.867 & \textbf{0.849} & 1.063 & \textbf{0.412} & 0.729 & 0.668 \\
2.4 & 0.8 & 2.0 & 1.595 & 3.938 & \textbf{3.846} & 4.450 & \textbf{4.374} & 2.790 & \textbf{2.687} & 13.709 & \textbf{8.481} & 10.272 \\
2.5 & 1.0 & 2.2 & 1.589 & 2.751 & \textbf{2.608} & 2.631 & \textbf{2.540} & 2.337 & \textbf{2.178} & 10.857 & \textbf{4.347} & 5.905 \\
\bottomrule
\end{tabular}%
}
\vspace{0.5ex}

\end{table}

\begin{table}[htbp]
\centering
\caption{Bias for the treatment effect \(\theta_t-\theta\) in the continuous endpoint simulations.}

\label{tab:supp:cont:treatment-effect-bias}
\resizebox{\textwidth}{!}{%
\begin{tabular}{ccccccccccccc}
\toprule
Scenario & $\theta$ & $\theta_t$ & NP & SAM & Gated SAM & EB-rMAP & Gated EB-rMAP & Mix50 & Gated Mix50 & PIP & Gated PIP & TTP \\
\midrule
\multicolumn{13}{c}{\textbf{Case 1: $d=0.4$ with fixed $\bD_h$}} \\
\addlinespace[0.5ex]
1.1 & -1.0 & 0.2 & -0.003 & -0.168 & \textbf{-0.046} & -0.111 & \textbf{-0.052} & -0.166 & \textbf{-0.052} & -0.564 & \textbf{-0.054} & -0.085 \\
1.2 & -0.8 & 0.4 & 0.007 & -0.223 & \textbf{-0.080} & -0.155 & \textbf{-0.081} & -0.210 & \textbf{-0.079} & -0.572 & \textbf{-0.094} & -0.138 \\
1.3 & 0.0 & 1.2 & 0.008 & 0.004 & 0.004 & 0.003 & 0.003 & 0.005 & \textbf{0.004} & 0.008 & \textbf{0.004} & 0.002 \\
1.4 & 0.8 & 2.0 & 0.006 & 0.236 & \textbf{0.098} & 0.171 & \textbf{0.099} & 0.222 & \textbf{0.096} & 0.565 & \textbf{0.113} & 0.165 \\
1.5 & 1.0 & 2.2 & -0.021 & 0.139 & \textbf{0.015} & 0.082 & \textbf{0.022} & 0.138 & \textbf{0.022} & 0.532 & \textbf{0.022} & 0.053 \\
\specialrule{\heavyrulewidth}{1pt}{1pt}
\multicolumn{13}{c}{\textbf{Case 2: $d=0.4$ with random $\bD_h$}} \\
\addlinespace[0.5ex]
2.1 & -1.0 & 0.2 & 0.005 & -0.172 & \textbf{-0.043} & -0.111 & \textbf{-0.049} & -0.168 & \textbf{-0.048} & -0.563 & \textbf{-0.051} & -0.088 \\
2.2 & -0.8 & 0.4 & 0.001 & -0.227 & \textbf{-0.086} & -0.160 & \textbf{-0.087} & -0.214 & \textbf{-0.084} & -0.568 & \textbf{-0.100} & -0.146 \\
2.3 & 0.0 & 1.2 & 0.001 & 0.095 & \textbf{0.070} & 0.086 & \textbf{0.070} & 0.086 & \textbf{0.064} & 0.131 & \textbf{0.079} & 0.094 \\
2.4 & 0.8 & 2.0 & 0.003 & 0.249 & \textbf{0.109} & 0.180 & \textbf{0.107} & 0.232 & \textbf{0.104} & 0.557 & \textbf{0.125} & 0.180 \\
2.5 & 1.0 & 2.2 & -0.017 & 0.131 & \textbf{0.015} & 0.079 & \textbf{0.024} & 0.132 & \textbf{0.023} & 0.525 & \textbf{0.021} & 0.051 \\
\bottomrule
\end{tabular}%
}
\vspace{0.5ex}

\end{table}

\begin{table}[htbp]
\centering
\caption{Coverage for the treatment effect $\theta_t-\theta$ in the continuous endpoint simulations.}
\label{tab:supp:cont:treatment-effect-coverage}
\resizebox{\textwidth}{!}{%
\begin{tabular}{ccccccccccccc}
\toprule
Scenario & $\theta$ & $\theta_t$ & NP & SAM & Gated SAM & EB-rMAP & Gated EB-rMAP & Mix50 & Gated Mix50 & PIP & Gated PIP & TTP \\
\midrule
\multicolumn{13}{c}{\textbf{Case 1: $d=0.4$ with fixed $\bD_h$}} \\
\addlinespace[0.5ex]
1.1 & -1.0 & 0.2 & 0.950 & 0.886 & \textbf{0.894} & 0.906 & \textbf{0.909} & 0.904 & \textbf{0.910} & 0.584 & \textbf{0.857} & 0.808 \\
1.2 & -0.8 & 0.4 & 0.953 & \textbf{0.835} & 0.833 & 0.845 & \textbf{0.846} & \textbf{0.879} & 0.877 & 0.488 & \textbf{0.742} & 0.679 \\
1.3 & 0.0 & 1.2 & 0.949 & \textbf{0.965} & 0.948 & \textbf{0.957} & 0.945 & \textbf{0.969} & 0.952 & \textbf{0.957} & 0.943 & 0.944 \\
1.4 & 0.8 & 2.0 & 0.942 & 0.814 & \textbf{0.816} & 0.820 & \textbf{0.823} & \textbf{0.852} & 0.849 & 0.463 & \textbf{0.707} & 0.630 \\
1.5 & 1.0 & 2.2 & 0.955 & 0.896 & \textbf{0.905} & 0.913 & \textbf{0.915} & 0.913 & \textbf{0.920} & 0.589 & \textbf{0.873} & 0.822 \\
\specialrule{\heavyrulewidth}{1pt}{1pt}
\multicolumn{13}{c}{\textbf{Case 2: $d=0.4$ with random $\bD_h$}} \\
\addlinespace[0.5ex]
2.1 & -1.0 & 0.2 & 0.948 & 0.872 & \textbf{0.881} & 0.890 & \textbf{0.895} & 0.895 & \textbf{0.900} & 0.558 & \textbf{0.838} & 0.777 \\
2.2 & -0.8 & 0.4 & 0.946 & 0.826 & 0.826 & \textbf{0.839} & 0.837 & \textbf{0.869} & 0.867 & 0.483 & \textbf{0.733} & 0.667 \\
2.3 & 0.0 & 1.2 & 0.940 & \textbf{0.943} & 0.918 & \textbf{0.933} & 0.913 & \textbf{0.952} & 0.928 & \textbf{0.931} & 0.907 & 0.909 \\
2.4 & 0.8 & 2.0 & 0.945 & \textbf{0.807} & 0.803 & \textbf{0.816} & 0.813 & \textbf{0.854} & 0.848 & 0.463 & \textbf{0.687} & 0.613 \\
2.5 & 1.0 & 2.2 & 0.950 & 0.888 & \textbf{0.899} & 0.905 & \textbf{0.910} & 0.905 & \textbf{0.912} & 0.612 & \textbf{0.878} & 0.827 \\
\bottomrule
\end{tabular}%
}
\vspace{0.5ex}

\end{table}

\yx 

 Table~\ref{tab:supp:cont:calibrated-power-revised} reports calibrated power for NP, SAM, EB-rMAP, Mix50, PIP, their WOW-gated counterparts, and TTP. Results are shown for three scenarios: fixed historical data with \(d=0.40\) in Case~1, fixed historical data with \(d=0.23\) in Case~2, and stochastic historical data with \(d=0.40\) in Case~3. \xx

\begin{table}[htbp]
\centering

\caption{Power for continuous endpoints under the calibrated posterior probability cutoff $C$.}

\label{tab:supp:cont:calibrated-power-revised}
\resizebox{\textwidth}{!}{%
\begin{tabular}{ccccccccccccc}
\toprule
Scenario & $\theta$ & $\theta_t$ & NP & SAM & Gated SAM & EB-rMAP & Gated EB-rMAP & Mix50 & Gated Mix50 & PIP & Gated PIP & TTP \\
\midrule
\multicolumn{13}{c}{\textbf{Case 1: $d=0.4$ with fixed $\bD_h$}} \\
\addlinespace[0.5ex]
1.1 & -1.5 & -0.3 & 0.895 & 0.832 & \textbf{0.897} & 0.876 & \textbf{0.902} & 0.837 & \textbf{0.910} & 0.578 & \textbf{0.887} & 0.887 \\
1.2 & -1.3 & -0.1 & 0.909 & 0.769 & \textbf{0.906} & 0.820 & \textbf{0.912} & 0.757 & \textbf{0.901} & 0.467 & \textbf{0.896} & 0.892 \\
1.3 & -1.2 & 0.0 & 0.890 & 0.721 & \textbf{0.886} & 0.805 & \textbf{0.891} & 0.723 & \textbf{0.900} & 0.410 & \textbf{0.892} & 0.880 \\
1.4 & -1.0 & 0.2 & 0.895 & 0.712 & \textbf{0.870} & 0.784 & \textbf{0.875} & 0.725 & \textbf{0.872} & 0.550 & \textbf{0.874} & 0.845 \\
1.5 & 0.0 & 1.2 & 0.896 & \textbf{0.983} & 0.967 & \textbf{0.975} & 0.965 & \textbf{0.980} & 0.964 & \textbf{1.000} & 0.971 & 0.980 \\
1.6 & 0.2 & 1.4 & 0.881 & \textbf{0.926} & 0.898 & \textbf{0.894} & 0.879 & \textbf{0.924} & 0.900 & \textbf{0.987} & 0.896 & 0.927 \\
1.7 & 0.5 & 1.7 & 0.905 & \textbf{0.779} & 0.714 & \textbf{0.681} & 0.657 & \textbf{0.832} & 0.782 & \textbf{0.859} & 0.628 & 0.708 \\
\specialrule{\heavyrulewidth}{1pt}{1pt}
\multicolumn{13}{c}{\textbf{Case 2: $d=0.23$ with fixed $\bD_h$}} \\
\addlinespace[0.5ex]
2.1 & -1.2 & -0.5 & 0.737 & 0.725 & \textbf{0.739} & 0.734 & \textbf{0.743} & 0.714 & \textbf{0.744} & 0.572 & \textbf{0.751} & 0.754 \\
2.2 & -1.0 & -0.3 & 0.708 & 0.626 & \textbf{0.710} & 0.647 & \textbf{0.713} & 0.595 & \textbf{0.712} & 0.320 & \textbf{0.711} & 0.705 \\
2.3 & -0.9 & -0.2 & 0.755 & 0.633 & \textbf{0.753} & 0.673 & \textbf{0.756} & 0.587 & \textbf{0.758} & 0.282 & \textbf{0.750} & 0.747 \\
2.4 & -0.8 & -0.1 & 0.749 & 0.572 & \textbf{0.743} & 0.619 & \textbf{0.745} & 0.535 & \textbf{0.744} & 0.280 & \textbf{0.748} & 0.743 \\
2.5 & 0.0 & 0.7 & 0.750 & \textbf{0.959} & 0.939 & \textbf{0.948} & 0.934 & \textbf{0.952} & 0.932 & \textbf{0.994} & 0.953 & 0.971 \\
2.6 & 0.2 & 0.9 & 0.731 & \textbf{0.821} & 0.803 & \textbf{0.800} & 0.792 & \textbf{0.821} & 0.795 & \textbf{0.963} & 0.832 & 0.884 \\
2.7 & 0.4 & 1.1 & 0.773 & \textbf{0.631} & 0.596 & \textbf{0.564} & 0.545 & \textbf{0.679} & 0.625 & \textbf{0.823} & 0.602 & 0.691 \\
\specialrule{\heavyrulewidth}{1pt}{1pt}
\multicolumn{13}{c}{\textbf{Case 3: $d=0.4$ with random $\bD_h$}} \\
\addlinespace[0.5ex]
3.1 & -1.5 & -0.3 & 0.894 & 0.805 & \textbf{0.894} & 0.857 & \textbf{0.898} & 0.800 & \textbf{0.896} & 0.527 & \textbf{0.892} & 0.892 \\
3.2 & -1.3 & -0.1 & 0.899 & 0.824 & \textbf{0.898} & 0.860 & \textbf{0.900} & 0.819 & \textbf{0.899} & 0.587 & \textbf{0.899} & 0.898 \\
3.3 & -1.2 & 0.0 & 0.896 & 0.725 & \textbf{0.889} & 0.812 & \textbf{0.892} & 0.731 & \textbf{0.892} & 0.412 & \textbf{0.893} & 0.882 \\
3.4 & -1.0 & 0.2 & 0.896 & 0.713 & \textbf{0.870} & 0.787 & \textbf{0.875} & 0.724 & \textbf{0.875} & 0.489 & \textbf{0.872} & 0.840 \\
3.5 & 0.0 & 1.2 & 0.899 & \textbf{0.962} & 0.938 & \textbf{0.947} & 0.931 & \textbf{0.960} & 0.938 & \textbf{0.995} & 0.937 & 0.954 \\
3.6 & 0.2 & 1.4 & 0.888 & \textbf{0.909} & 0.876 & \textbf{0.878} & 0.862 & \textbf{0.913} & 0.883 & \textbf{0.980} & 0.867 & 0.903 \\
3.7 & 0.5 & 1.7 & 0.885 & \textbf{0.747} & 0.696 & \textbf{0.626} & 0.612 & \textbf{0.794} & 0.744 & \textbf{0.832} & 0.636 & 0.720 \\
\bottomrule
\end{tabular}%
}
\vspace{0.5ex}

\end{table}

In Case 1, WOW-gated methods consistently outperform their non-gated counterparts when prior–data conflict is present with  $\theta<\theta_h$, where aggressive borrowing leads to underestimated treatment effects. \yx For instance, Gated Mix50, Gated EB-rMAP,  Gated SAM and Gated PIP achieve substantially higher power than their original versions. \xx  When $\theta = 0$,  where historical and current data align, gated methods yield slightly lower power relative to their non-gated counterparts, but the difference remains small, demonstrating that WOW preserves much of the efficiency when borrowing is appropriate.  In settings where $\theta>\theta_h$, borrowing tends to underestimate the control response and inflate the treatment effect. Under these conditions, WOW-gated methods adopt a more conservative borrowing approach, yielding slightly lower power compared to their non-gated versions. This trade-off reflects a deliberate focus on robustness to mitigate the risk of overestimating treatment effects due to incompatibility.  Non-gated methods, as demonstrated in  Figure~\ref{fig:supp:cont:theta-relative-bias}, produce substantial bias when incompatible data are borrowed, further justifying the need for WOW’s safeguards against misleading inference. \yx Case 2 ($d=0.23$) and Case 3 ($d =0.4$) show similar trends. \xx  WOW consistently reduces bias and delivers stable power performance across varying levels of prior–data conflict.

\yx
Tables~\ref{tab:supp:cont:prespecified-power} and~\ref{tab:supp:cont:type1-error} report the corresponding pre-specified power and type I error for the continuous endpoint under the common cutoff $C=0.95$, and show patterns similar to calibrated power results.
\xx

\begin{table}[htbp]
\centering
\caption{Power for continuous endpoints under the pre-specified posterior probability cutoff $C=0.95$.}

\label{tab:supp:cont:prespecified-power}
\resizebox{\textwidth}{!}{%
\begin{tabular}{ccccccccccccc}
\toprule
Scenario & $\theta$ & $\theta_t$ & NP & SAM & Gated SAM & EB-rMAP & Gated EB-rMAP & Mix50 & Gated Mix50 & PIP & Gated PIP & TTP \\
\midrule
\multicolumn{13}{c}{\textbf{Case 1: $d=0.4$ with fixed $\bD_h$}} \\
\addlinespace[0.5ex]
1.1 & -1.5 & -0.3 & 0.885 & 0.820 & \textbf{0.886} & 0.855 & \textbf{0.880} & 0.807 & \textbf{0.881} & 0.578 & \textbf{0.887} & 0.888 \\
1.2 & -1.3 & -0.1 & 0.898 & 0.760 & \textbf{0.896} & 0.825 & \textbf{0.894} & 0.750 & \textbf{0.894} & 0.446 & \textbf{0.896} & 0.891 \\
1.3 & -1.2 & 0.0 & 0.899 & 0.722 & \textbf{0.894} & 0.805 & \textbf{0.891} & 0.712 & \textbf{0.891} & 0.392 & \textbf{0.893} & 0.884 \\
1.4 & -1.0 & 0.2 & 0.897 & 0.706 & \textbf{0.870} & 0.778 & \textbf{0.871} & 0.702 & \textbf{0.871} & 0.415 & \textbf{0.870} & 0.840 \\
1.5 & 0.0 & 1.2 & 0.900 & \textbf{0.978} & 0.967 & \textbf{0.973} & 0.966 & \textbf{0.973} & 0.964 & \textbf{1.000} & 0.971 & 0.980 \\
1.6 & 0.2 & 1.4 & 0.884 & \textbf{0.956} & 0.927 & \textbf{0.943} & 0.928 & \textbf{0.953} & 0.930 & \textbf{0.996} & 0.929 & 0.946 \\
1.7 & 0.5 & 1.7 & 0.903 & \textbf{0.926} & 0.908 & \textbf{0.919} & 0.911 & \textbf{0.929} & 0.912 & \textbf{0.981} & 0.908 & 0.914 \\
\specialrule{\heavyrulewidth}{1pt}{1pt}
\multicolumn{13}{c}{\textbf{Case 2: $d=0.23$ with fixed $\bD_h$}} \\
\addlinespace[0.5ex]
2.1 & -1.2 & -0.5 & 0.741 & 0.727 & \textbf{0.740} & 0.724 & \textbf{0.735} & 0.702 & \textbf{0.737} & 0.562 & \textbf{0.741} & 0.741 \\
2.2 & -1.0 & -0.3 & 0.738 & 0.646 & \textbf{0.737} & 0.668 & \textbf{0.731} & 0.607 & \textbf{0.733} & 0.329 & \textbf{0.740} & 0.739 \\
2.3 & -0.9 & -0.2 & 0.755 & 0.630 & \textbf{0.754} & 0.664 & \textbf{0.751} & 0.577 & \textbf{0.751} & 0.261 & \textbf{0.752} & 0.754 \\
2.4 & -0.8 & -0.1 & 0.739 & 0.557 & \textbf{0.733} & 0.599 & \textbf{0.730} & 0.509 & \textbf{0.730} & 0.199 & \textbf{0.738} & 0.730 \\
2.5 & 0.0 & 0.7 & 0.735 & \textbf{0.941} & 0.935 & \textbf{0.935} & 0.933 & \textbf{0.929} & 0.925 & \textbf{0.984} & 0.953 & 0.966 \\
2.6 & 0.2 & 0.9 & 0.732 & \textbf{0.895} & 0.861 & \textbf{0.879} & 0.858 & \textbf{0.890} & 0.859 & \textbf{0.991} & 0.865 & 0.902 \\
2.7 & 0.4 & 1.1 & 0.741 & \textbf{0.841} & 0.777 & \textbf{0.820} & 0.778 & \textbf{0.845} & 0.779 & \textbf{0.967} & 0.775 & 0.802 \\
\specialrule{\heavyrulewidth}{1pt}{1pt}
\multicolumn{13}{c}{\textbf{Case 3: $d=0.4$ with random $\bD_h$}} \\
\addlinespace[0.5ex]
3.1 & -1.5 & -0.3 & 0.894 & 0.805 & \textbf{0.895} & 0.850 & \textbf{0.890} & 0.794 & \textbf{0.891} & 0.521 & \textbf{0.893} & 0.893 \\
3.2 & -1.3 & -0.1 & 0.898 & 0.818 & \textbf{0.897} & 0.856 & \textbf{0.893} & 0.804 & \textbf{0.894} & 0.535 & \textbf{0.897} & 0.897 \\
3.3 & -1.2 & 0.0 & 0.898 & 0.727 & \textbf{0.891} & 0.807 & \textbf{0.889} & 0.722 & \textbf{0.890} & 0.393 & \textbf{0.891} & 0.881 \\
3.4 & -1.0 & 0.2 & 0.898 & 0.713 & \textbf{0.870} & 0.776 & \textbf{0.871} & 0.714 & \textbf{0.870} & 0.452 & \textbf{0.868} & 0.837 \\
3.5 & 0.0 & 1.2 & 0.900 & \textbf{0.964} & 0.946 & \textbf{0.955} & 0.944 & \textbf{0.959} & 0.944 & \textbf{0.996} & 0.944 & 0.958 \\
3.6 & 0.2 & 1.4 & 0.890 & \textbf{0.950} & 0.919 & \textbf{0.938} & 0.921 & \textbf{0.947} & 0.921 & \textbf{0.993} & 0.920 & 0.937 \\
3.7 & 0.5 & 1.7 & 0.892 & \textbf{0.924} & 0.900 & \textbf{0.914} & 0.902 & \textbf{0.926} & 0.903 & \textbf{0.980} & 0.898 & 0.905 \\
\bottomrule
\end{tabular}%
}
\vspace{0.5ex}

\end{table}

\begin{table}[htbp]
\centering
\caption{Type I error for continuous endpoints under the pre-specified posterior probability cutoff $C=0.95$.}

\label{tab:supp:cont:type1-error}
\resizebox{\textwidth}{!}{%
\begin{tabular}{ccccccccccccc}
\toprule
Scenario & $\theta$ & $\theta_t$ & NP & SAM & Gated SAM & EB-rMAP & Gated EB-rMAP & Mix50 & Gated Mix50 & PIP & Gated PIP & TTP \\
\midrule
\multicolumn{13}{c}{\textbf{Case 1: $d=0.4$ with fixed $\bD_h$}} \\
\addlinespace[0.5ex]
1.1 & -1.5 & -1.5 & 0.050 & \textbf{0.049} & 0.050 & 0.043 & 0.043 & 0.044 & 0.044 & \textbf{0.050} & 0.051 & 0.051 \\
1.2 & -1.3 & -1.3 & 0.048 & 0.048 & 0.048 & 0.046 & 0.046 & 0.045 & 0.046 & \textbf{0.045} & 0.049 & 0.049 \\
1.3 & -1.2 & -1.2 & 0.053 & 0.053 & 0.053 & \textbf{0.050} & 0.051 & \textbf{0.050} & 0.051 & \textbf{0.043} & 0.051 & 0.051 \\
1.4 & -1.0 & -1.0 & 0.051 & \textbf{0.048} & 0.050 & \textbf{0.046} & 0.048 & \textbf{0.043} & 0.050 & \textbf{0.031} & 0.047 & 0.047 \\
1.5 & 0.0 & 0.0 & 0.052 & \textbf{0.035} & 0.051 & \textbf{0.042} & 0.053 & \textbf{0.031} & 0.050 & \textbf{0.037} & 0.054 & 0.050 \\
1.6 & 0.2 & 0.2 & 0.052 & 0.102 & 0.102 & \textbf{0.111} & 0.114 & \textbf{0.085} & 0.086 & 0.143 & \textbf{0.133} & 0.139 \\
1.7 & 0.5 & 0.5 & 0.049 & 0.262 & \textbf{0.255} & 0.252 & 0.253 & 0.196 & \textbf{0.195} & 0.498 & \textbf{0.362} & 0.410 \\
\specialrule{\heavyrulewidth}{1pt}{1pt}
\multicolumn{13}{c}{\textbf{Case 2: $d=0.23$ with fixed $\bD_h$}} \\
\addlinespace[0.5ex]
2.1 & -1.2 & -1.2 & 0.051 & 0.051 & 0.051 & 0.049 & 0.049 & 0.049 & 0.049 & 0.048 & 0.049 & 0.049 \\
2.2 & -1.0 & -1.0 & 0.059 & 0.059 & 0.059 & \textbf{0.056} & 0.058 & \textbf{0.055} & 0.058 & \textbf{0.056} & 0.059 & 0.060 \\
2.3 & -0.9 & -0.9 & 0.050 & \textbf{0.049} & 0.051 & 0.048 & 0.048 & \textbf{0.047} & 0.048 & \textbf{0.043} & 0.051 & 0.052 \\
2.4 & -0.8 & -0.8 & 0.048 & \textbf{0.047} & 0.048 & 0.045 & 0.045 & \textbf{0.043} & 0.046 & \textbf{0.026} & 0.047 & 0.046 \\
2.5 & 0.0 & 0.0 & 0.051 & \textbf{0.032} & 0.048 & \textbf{0.036} & 0.049 & \textbf{0.029} & 0.045 & \textbf{0.029} & 0.050 & 0.041 \\
2.6 & 0.2 & 0.2 & 0.053 & \textbf{0.141} & 0.142 & \textbf{0.151} & 0.155 & \textbf{0.120} & 0.121 & 0.193 & \textbf{0.180} & 0.185 \\
2.7 & 0.4 & 0.4 & 0.043 & 0.234 & \textbf{0.233} & 0.229 & \textbf{0.226} & 0.193 & \textbf{0.189} & 0.462 & \textbf{0.317} & 0.357 \\
\specialrule{\heavyrulewidth}{1pt}{1pt}
\multicolumn{13}{c}{\textbf{Case 3: $d=0.4$ with random $\bD_h$}} \\
\addlinespace[0.5ex]
3.1 & -1.5 & -1.5 & 0.051 & 0.051 & 0.051 & 0.046 & 0.046 & 0.046 & 0.046 & \textbf{0.048} & 0.051 & 0.051 \\
3.2 & -1.3 & -1.3 & 0.050 & \textbf{0.049} & 0.050 & 0.047 & 0.047 & \textbf{0.045} & 0.047 & \textbf{0.035} & 0.048 & 0.049 \\
3.3 & -1.2 & -1.2 & 0.051 & 0.051 & 0.051 & 0.047 & 0.047 & 0.047 & 0.047 & \textbf{0.043} & 0.049 & 0.048 \\
3.4 & -1.0 & -1.0 & 0.051 & 0.050 & 0.050 & 0.047 & 0.047 & 0.047 & \textbf{0.046} & \textbf{0.042} & 0.048 & 0.048 \\
3.5 & 0.0 & 0.0 & 0.051 & \textbf{0.054} & 0.066 & \textbf{0.061} & 0.071 & \textbf{0.050} & 0.060 & \textbf{0.062} & 0.074 & 0.070 \\
3.6 & 0.2 & 0.2 & 0.051 & \textbf{0.118} & 0.119 & 0.128 & 0.128 & \textbf{0.097} & 0.099 & 0.164 & \textbf{0.151} & 0.158 \\
3.7 & 0.5 & 0.5 & 0.053 & 0.268 & \textbf{0.264} & 0.256 & 0.256 & 0.205 & 0.205 & 0.580 & \textbf{0.376} & 0.441 \\
\bottomrule
\end{tabular}%
}
\vspace{0.5ex}

\end{table}

\endgroup

\end{document}